\documentclass[AIF,Unicode,manuscript]{cedram}

\usepackage{bm} 
\usepackage{mathrsfs} 
\usepackage{stmaryrd} 

\usepackage{caption}
\usepackage{subcaption}	
\usepackage[edges]{forest}

\hypersetup{colorlinks,citecolor=blue,filecolor=blacklinkcolor=red, urlcolor=magenta,linkcolor=red}
\newcommand{\bbR}{\mathbb{R}}
\newcommand{\bbC}{\mathbb{C}}

\DeclareMathOperator*{\slim}{s\,--\,lim}
\equalenv{rem}{rema}
\newcommand{\bra}[1]{\langle {#1} \rangle}
\newcommand{\eqk}{\underset{K}{\sim}}
\newcommand{\geqk}{\underset{K}{\gtrsim}}
\newcommand{\supp}{\textrm{supp}~}
\newcommand{\slashh}{H\mkern-13mu/}

\title[]{Scattering theory for Dirac fields near an extreme Kerr-de Sitter black hole}
\alttitle{Théorie de la diffusion pour champs de Dirac au voisinage d'un trou noir extrême de type Kerr-de Sitter.}

\author{\firstname{Jack} \middlename{A.} \lastname{Borthwick}}
\address{Univ. Brest\\ UMR CNRS 6205\\ Laboratoire de Mathématiques de Bretagne Atlantique\\ 6. av Victor Le Gorgeu\\ CS 93837\\ 29238 BREST cedex 3}
\email[]{jack.borthwick@math.cnrs.fr}
 \thanks{Il m’aurait sans doute été impossible d’achever ce travail en un temps fini sans les innombrables discussions avec mon directeur de thèse Jean-Philippe Nicolas, pour lesquelles je le remercie. Je souhaiterais également remercier Thierry Daudé, dont la thèse constitue la base de ce travail, le temps qu’il a pu m’accorder pour répondre à mes questions et doutes, a une valeur inestimable. I would also like to thank the referee who carefully read the original lengthy manuscript and whose very detailed report contained valuable advice.  }



\begin{abstract} 
 In this paper, we construct a scattering theory for classical massive Dirac fields near the “double” horizon of an extreme Kerr-de Sitter blackhole. Our main tool is the existence of a conjugate operator in the sense of Mourre theory. Additionally, despite the fact that effects of the rotation are ``amplified'' near the double horizon, we show that one can still reduce our study to a 1-dimensional problem through an appropriate decomposition of the Hilbert space.
 \end{abstract}

\begin{altabstract}
Dans cet article, nous développons une théorie du scattering pour des champs de Dirac massifs ou non en métrique Kerr-de Sitter extrême, dans la région située entre l’horizon (double) du trou noir et l’horizon cosmologique. L’outil principal de la construction est l’existence d’un opérateur conjugué au sens de la théorie de Mourre. Par ailleurs, bien que les effets de la rotation soient amplifiées au voisinage de l’horizon double, nous montrons qu’il est néanmoins possible de se ramener à un problème de scattering unidimensionnelle moyennant une décomposition ad-hoc de l’espace de Hilbert.
\end{altabstract}    
    
\begin{document}
\newpage

\maketitle

\section{Introduction}
Over the past two decades or so there has been quite a bit of mathematical interest in scattering theories for particles in black-hole type geometries. This is useful for the understanding of these geometries and the detection of black holes but also in the study of Quantum Field Theory on curved spacetimes, see for example~\cite{Dappiaggi:2009fx,gerard:hal-02939993}.

For rotating black holes, due to super-radiance, it is well known that the usual energy functional of integer spin particle fields, described for instance by the wave or Klein-Gordon equation, is no longer positive-definite, this leads to obvious technical difficulties that have nevertheless been overcome in a handful of situations, such as the Klein-Gordon equation on (De Sitter) Kerr spacetimes~\cite{Georgescu:2017vl,Hafner:2003aa} and the wave equation on Kerr spacetime~\cite{Dafermos:2014jwa}.

On the other hand, for Dirac fields, there is still a conserved current which leads to a natural Hilbert space framework adapted to a spectral theory approach. Scattering theories for massive or massless Dirac fields have been constructed in this manner in the exterior region of Reissner-Nordström, slow Kerr and Kerr-Newman black holes~\cite{thdaude,Nicolas:2004aa}. More recently, there has been interest in non-asymptotically flat backgrounds such as Schwarzschild-de Sitter~\cite{idelonriton:tel-01370116}, slow Kerr Newman-de Sitter~\cite{Daude:2016aa} and slow Kerr-Newman-AdS~\cite{Belgiorno:2010aa} black holes. 

In this paper we study the case of an \emph{extreme} Kerr-de Sitter black hole in a region situated between what we will refer to as a ``double'' horizon and a usual ``simple'' one (the cosmological horizon). The ``double'' horizon is the hypersurface resulting from the coincidence of the two inner black hole horizons\footnote{which occurs for special choices of the parameters of the family. }, and differs quite significantly from the exterior horizon of, for instance, Kerr spacetime.  The extreme case is of particular interest for the understanding of mechanisms behind stability/instability of black hole type spacetimes as it presents features of both types. This is analysed thoroughly in the case of an Extreme Reissner-Nordström black hole in~\cite{Aretakis:2011aa,Aretakis:2011hc}, and complemented by the remarks in~\cite{Batic:2007jb} on the asymptotic behavior to the wave equation. Regarding the Dirac equation, an integral representation of the Dirac propagator in the extreme Kerr metric is derived in~\cite{Bizon:2012we}. Our main result is the asymptotic completeness of the Dirac operator in an extreme Kerr-de Sitter black hole, this is the conjunction of the absence of singular continuous spectrum, and Theorem~\ref{th:final_result}, formulated in Section~\ref{fullscattering}, which provides a complete description of the asymptotic behaviour of all states in the absolutely continuous subspace by establishing the existence of a direct sum decomposition into states that scatter towards either the cosmological or blackhole horizon. This is achieved by comparison with simplified dynamics:
\begin{itemize}
\item $H_{-\infty}= \Gamma^1D_{r^*} + \frac{a}{r_+^2+a^2}-\frac{a}{r_e^2+a^2}D_\phi$ at the extreme black hole horizon,
\item $H_{+\infty}=\Gamma^1D_{r^*} + \frac{a}{r_+^2+a^2}-\frac{a}{r_+^2+a^2}D_\phi$ at the cosmological horizon.
\end{itemize}
In the context of understanding stability/instability features of extreme black hole space times, this result can perhaps be seen as a stability feature of the extreme Kerr-de Sitter blackhole.
  
Our global strategy follows closely that of~\cite{thdaude,Daude:2010aa,Nicolas:2004aa}~: we will adopt the point of view of a class of observers for which the two horizons are asymptotic and will show in Section~\ref{mourre} that a conjugate operator in the sense of Mourre theory~\cite{Amrein:1996aa,Mourre:1981aa} can be constructed in an analogous fashion to that in the exterior of a Kerr black hole as in~\cite{thdaude,Nicolas:2004aa}. Furthermore, it has already been noted, for example in~\cite{Belgiorno:2010aa}, that the presence of the simple horizon is enough to ensure that the usual proof of the absence of eigenvalues -- via a Grönwall inequality exploiting the separability of the Dirac equation -- follows through without modification. However, our results do not follow directly from these works due to long-range potentials at the extreme horizon and a significantly perturbed angular operator. In particular, the decomposition of the Hilbert space into spin harmonics, essential to the reduction to the spherically symmetric case treated in~\cite{Daude:2010aa} is no longer stable.  A key ingredient to our analysis, carried out in Section~\ref{sec:Hx} is constructing operators at both asymptotic ends with similar adapted decompositions and of which the full Dirac operator is a short-range perturbation. 
Furthermore, it is worth noting that since the mass terms do not survive at either of the horizons, despite constituting a long-range potential near the double one, some of the arguments in~\cite{Daude:2010aa} can be simplified.
%
%
%
\subsection{The Kerr-de Sitter metric}
Throughout this text, we will mainly use the usual Boyer-Lindquist like coordinates $(t,r,\theta,\varphi)$ in which the Kerr-de Sitter metric is known to be (signature $(+,-,-,-)$):
\begin{equation} \label{metric} g = \frac{\Delta_r}{\Xi^2\rho^2}[\textrm{d}t-a\sin^2\theta\textrm{d}\varphi]^2 - \frac{\rho^2}{\Delta_r}\textrm{d}r^2 -\frac{\rho^2}{\Delta_\theta} \textrm{d}\theta^2 - \frac{\Delta_\theta \sin^2\theta}{\rho^2\Xi^2}[(r^2+a^2)\textrm{d}\varphi-a\textrm{d}t]^2, \end{equation}
where:
\begin{equation}
\begin{aligned} 
&l^2=\frac{\Lambda}{3}, && \Delta_r=r^2-2Mr+a^2-l^2r^2(r^2+a^2), \\
&\Xi=1+a^2l^2,&& \Delta_\theta=1+a^2l^2\cos^2\theta,\\
& \rho^2=r^2+a^2\cos^2\theta. 
\end{aligned}
\end{equation}
It depends on three parameters $a,M,\Lambda$, the angular momentum per unit mass of the black hole, the mass of the black hole and the cosmological constant, respectively. We will always assume $l > 0$.

The above expression is singular when $\Delta_r=0$ or $\rho=0$, however, the manifold can be analytically extended across the singularities $\{\Delta_r=0\}$. In such an extension, the roots of $\Delta_r$ give rise to null hypersurfaces that we will refer to as horizons. They will be labelled by the root $r_i$ to which they correspond as so: $\mathscr{H}_{r_i}$. If $r_i$ is a double (resp. simple) root of $\Delta_r$, $\mathscr{H}_{r_i}$ will be said to be a ``double" (resp. ``simple'') horizon. In, for instance,~\cite{Borthwick:2018aa}, it is shown that the roots of $\Delta_r$ can be labelled such that either:
\begin{enumerate}
\item $r_{--}< 0 < r_- < r_+ < r_{++}$ 
\item\label{extreme} $r_{--}<0 < r_-=r_+ < r_{++}$ 
\item $r_{--}<0 < r_-<r_+=r_{++}$ 
\item $r_{--}<0 < r_-=r_+-r_{++}$
\item $r_{--}<r_{++}, r_-,r_+ \in \mathbb{C}\setminus{\mathbb{R}}$. 
\end{enumerate}
We will refer to case~(\ref{extreme}) as extreme Kerr-de Sitter; a necessary and sufficient condition for this is:
\begin{equation}
\begin{aligned} \label{condition_al} &|a|l<2-\sqrt{3}, \\
&M^2=\frac{(1-a^2l^2)(a^4l^4+34a^2l^2+1)- \gamma^{\frac{3}{2}}}{54l^2}, \end{aligned}\end{equation}
where $\gamma=(1-a^2l^2)^2-12a^2l^2 $. In this situation the double root is given by:
\begin{equation}
r_e\footnote{We note that in~\cite{Borthwick:2018aa} it was denoted by $x$} = \frac{12a^2l^2+(1-a^2l^2)(1-a^2l^2-\sqrt{\gamma})}{18Ml^2}.
\end{equation}
For future reference, we quote the following useful properties of $r_e$ :
\begin{equation}
\left\{ \begin{array}{l} 0 \leq r_e < \frac{4}{3}\frac{a^2}{M} \\ l^2r_e^4+a^2=Mr_e\end{array}\right.
\end{equation}

Finally, we note that the other two roots $r_{++}$ and $r_{--}$ are equally those of the polynomial:
\begin{equation}X^2 +2r_eX - \frac{a^2}{l^2r_e^2}. \end{equation}
To avoid unnecessarily complicated subscripts, we will now rename the roots of $\Delta_r$ as follows: \[r_- < 0 < r_e < r_+.\]
The region, $B$, in which we will study the scattering of Dirac fields is defined in the coordinates $(t,r,\theta,\varphi)$ by $r_e<r<r_+$. In essence, $B=\mathbb{R}\times]r_e,r_+[\times S^2$, with the metric given by~\eqref{metric}, that extends analytically to the poles. It is between two horizons, one double, one simple and it is the effect of the double horizon that we wish to understand.

The scattering problem will be considered from the point of view of a stationary observer with world-line:
\[r=r_0, \theta=\theta_0, \varphi=\omega t+\phi_0, \omega \in \mathbb{R}, r_0\in]r_e,r_+[,\theta_0\in]0,\pi[,\phi_0 \in ]0,2\pi[.\]
Proper time for such an observer differs from the coordinate function $t$ only by a multiplicative constant depending on the parameters of the trajectory. For this family of observers photons travelling, say, along a principal null geodesic, which are in some sense the most direct trajectories for light to travel towards one of the horizons, will not reach it in finite time. For instance, the coordinate time $t$ necessary for a photon, emitted from $r=r_0$ at $t=t_0$, to reach $\mathscr{H}_+$ travelling along such a curve is:
\begin{equation}
 t-t_0 = \int_{r_0}^{r_+} \frac{\Xi(r^2+a^2)}{\Delta_r}\textrm{d}r = +\infty
 \end{equation}
 In fact, for our purposes, it will be appropriate to replace the coordinate $r$, by the Regge-Wheeler type coordinate $r^* = \displaystyle \int \frac{\Xi(r^2+a^2)}{\Delta_r}\textrm{d}r$ appearing in this computation. By definition:
\begin{equation} \label{eq:rrstar} \textrm{d}{r^*}=\frac{\Xi(r^2+a^2)}{\Delta_r} \textrm{d}r. \end{equation} 
 It will be useful to calculate an explicit expression for $r^*$ by a partial fraction decomposition of the integrand: 
 \begin{equation} \frac{r^2+a^2}{(r-r_-)(r-r_e)^2(r-r_+)} = \frac{\alpha}{r-r_-} + \frac{\beta}{r-r_+} + \frac{\gamma}{r-r_e} + \frac{\delta}{(r-r_e)^2}.  \end{equation}
The coefficients $\alpha,\beta,\gamma,\delta$ are given by~:
\begin{equation*} \begin{gathered} \alpha= - \frac{l}{2}\sqrt{\frac{r_e}{M}}\frac{r^2_-+a^2}{(r_e-r_-)^2}<0, \quad \beta = \frac{l}{2}\sqrt{\frac{r_e}{M}}\frac{ r^2_++a^2}{(r_+-r_e)^2}>0,\\ \delta = \frac{l^2r_e^2(r_e^2+a^2)}{3Mr_e-4a^2} <0,\quad   \gamma = -\frac{2l^2r_e^3(2r_e^2 -7Mr_e +6a^2)}{(3Mr_e-4a^2)^2} < 0.\end{gathered}\end{equation*}
The sign of $\gamma$ follows from the following relations: $$\left\{ \begin{array}{l} r_e^2l^2(r_e^2+a^2)=r_e^2+a^2 -2Mr_e, \\ 0 < 3Mr_e-4a^2 -2r_e^2l^2(r_e^2+a^2) = 7Mr_e - 6a^2 -2r_e^2. \end{array}\right.$$
The expression of $r^*$ is therefore:
\begin{equation}\label{eq:regge_wheeler}\boxed{ \begin{aligned}r^* = \frac{\Xi}{2l}\sqrt{\frac{r_e}{M}}\ln \left(\frac{| r-r_-|^{\eta_-}}{|r-r_+|^{\eta_+}}\right) +\frac{r_e^2(r_e^2+a^2)}{3Mr_e-4a^2}\frac{\Xi}{r-r_e}   \\ +\frac{2r_e^3(2r_e^2 -7Mr_e +6a^2)}{(3Mr_e-4a^2)^2} \Xi\ln |r-r_e| + R_0.\end{aligned}} \end{equation} 
Above, $R_0$ is an arbitrary real constant and $\eta_{\pm} = \frac{r_\pm^2 +a^2}{(r_e-r_\pm)^2}$.

\noindent From \eqref{eq:regge_wheeler}, one can deduce the following asymptotic equivalences: 

\begin{lemme}
\begin{equation} r_+ - r \underset{r^* \to + \infty}{\sim} e^{-\frac{2l}{\Xi \eta_+}\sqrt{\frac{M}{r_e}} r^*}, \label{eq:asymp_cosmo} \end{equation}

\begin{equation} r-r_e \underset{r^* \to - \infty}{\sim} \frac{r_e^2(r_e^2+a^2)\Xi}{3Mr_e-4a^2}\frac{1}{r^*}.\label{eq:asymp_double} \end{equation}
\end{lemme}
 
Equation~\eqref{eq:asymp_cosmo} is true for a suitable choice of $R_0$: it is the usual behaviour that we have come to expect at a simple black hole horizon. Equation~\eqref{eq:asymp_double}, on the other hand, illustrates the first notable difference in the case of the double horizon, as the decay is a lot slower. This will be a source of long-range potentials in the Dirac operator. 

\subsection{The Dirac equation}
\label{section:intro_dirac_eq}
On $B$, $\Delta_r>0$ and the coordinate $t$ is a ``time function'',  providing a foliation $(\Sigma_t)_{t\in \mathbb{R}}$ of $B$ into spacelike Cauchy hypersurfaces.  $B$ is therefore an orientable globally hyperbolic 4-manifold and as such, by a result due to R. Geroch~\cite{Geroch:1968aa,Geroch:1970aa}, possesses a global spin structure; required for the Dirac equation.

The Dirac equation is most conveniently expressed with Penrose's abstract index notation (cf.~\cite{Penrose:1984aa}). Let $\mathbb{S}^A$ be the module of sections of the two-spinor bundle $\mathbb{S}$ and, $\mathbb{S}^{A'}$, that of the pointed two spinor bundle $\mathbb{S}'$; lowered indices are used for sections of the dual bundles. We recall that $\mathbb{S}^A$ is identified with $\mathbb{S}^{A'}$ via complex conjugation and to $\mathbb{S}_A$ via the canonical symplectic form $\varepsilon_{AB}$ according to~:
\begin{equation*} \left\{\begin{array}{l} \kappa_B = \kappa^A\varepsilon_{AB}=-\varepsilon_{BA}\kappa^A \\ \kappa^{A'}=\overline{\kappa^A} \end{array}\right.  , \quad \kappa^A\in \mathbb{S}^A.\end{equation*}
The bundle $\mathbb{S}\otimes\mathbb{S}'$ can be identified with the complexified tangent bundle $\mathbb{C}\otimes TB$ and finally:  \[\varepsilon_{AB}\varepsilon_{A'B'}=g_{ab}.\]

\noindent Following~\cite{Nicolas:2002aa}, we will refer to elements of $\mathbb{S}_A\oplus \mathbb{S}^{A'}$ as \emph{Dirac spinors}, the massive Dirac equation for a spin-$\frac{1}{2}$ Dirac spinor $(\phi_A,\chi^{A'})$ is then~:
 \begin{equation}
 \label{eq:dirac}
 \left\{ \begin{array}{lcr} \nabla^{AA'} \phi_A& =& \mu \chi^{A'} \\ \nabla_{AA'} \chi^{A'} &=& -\mu \phi_A \ \end{array}, \quad\mu=\frac{m}{\sqrt{2}}\right. .
 \end{equation}

As mentioned in the introduction, it is well known that the equation has a conserved current, namely: \[ j_{AA'}=\phi_A \bar{\phi}_{A'}+\chi_{A'}\bar{\chi}_{A}.\] Thus the total charge:
\begin{equation}\label{eq:charge} Q=\int_{\Sigma_t} T^a j_a \omega_{g,\Sigma_t},\end{equation}
is conserved. $\omega_{g,\Sigma_t}=\sqrt{\frac{\Delta_r}{\Delta_\theta}}\frac{\rho\sigma}{(r^2+a^2)\Xi^2}\textrm{d}r^*\wedge(\sin\theta\textrm{d}\theta\wedge \textrm{d}\varphi)$ is the induced volume form on $\Sigma_t$\footnote{Oriented by $-\nabla t$.} and $T^a$ is colinear to $\nabla^a t$ and normalised, for convenience, such that $T^aT_a=2$. 

$Q$ defines an inner product on spinors defined on any slice\footnote{These can be thought of as either sections of the pullback bundle of $\mathbb{S}$ via the canonical injection, or, sections of the spinor bundle on $\Sigma_s$; there is an identification between them since $\dim B=4$. \iffalse The underlying reason for this is that $\mathbb{C}\otimes\mathfrak{su}_2=\mathfrak{sl}_2.$ \fi}, $\Sigma_t$, $t\in \mathbb{R}$, and gives rise to a Hilbert space $\mathscr{H}_t$.
Solving the Dirac equation can be thought of as finding a family of isometries $U(u,s): \mathscr{H}_s\mapsto\mathscr{H}_u$ such that for any $u,s,w\in \mathbb{R}$~: \[ U(s,s)=\textrm{Id},\quad U(u,s)U(s,w)=U(u,w). \]
The framework sketched here can nevertheless be significantly simplified since $\partial_t$ is a global Killing field on $B$. All slices $\Sigma_t$ are thus isometric, in particular, $B$ is homeomorphic to $\mathbb{R}\times\Sigma$ for some fixed $\Sigma$. Furthermore, the $\mathscr{H}_t$ can all be identified and so one can view the problem as an evolution problem on a fixed Hilbert space $\mathscr{H}$.
For these reasons, we will seek expressly to write the Dirac equation as a Schrödinger type equation. Moreover, we will work directly with spinor densities\footnote{An orientation on $\Sigma$ can be seen as a bundle morphisme between $\Lambda^nT^*\Sigma$ and the density bundle.} on $\Sigma$~:
\[(\phi_A,\chi^{A'})|\omega_{g,\Sigma}|^{\frac{1}{2}}.\]
After a choice of spin-frame, this means that our Hilbert space $\mathscr{H}$ can be assimilated with $L^2(\Sigma)\otimes\mathbb{C}^4=L^2(\mathbb{R}_{r^*}\times S^2)\otimes \mathbb{C}^4$ equipped with its natural inner product~:
\[(\phi,\psi)=\int \langle\phi,\psi\rangle_{\mathbb{C}^4} \textrm{d}r^*\textrm{d}\Omega, \quad \textrm{d}\Omega=\sin\theta\textrm{d}\theta\textrm{d}\varphi.\]
We refer to~\cite{Nicolas:2002aa} for a more detailed discussion on the framework outlined above.
\subsubsection{Spin frame}
In the study of the massless Dirac equation in the case of a slow Kerr black hole~\cite{Nicolas:2004aa}, it was remarked that in order to avoid some artificial long-range terms one should choose a spin-frame determined up to sign by a local orthonormal frame of $TB$ that is aligned with the foliation of $B$ , in the sense that the timelike vector should be parallel to $\nabla^a t$. Since ${\nabla^a t}^\perp = \text{span}(\partial_r, \partial_\theta, \partial_\varphi)$, we make the simplest choice:
\[\begin{gathered}{g'}_0^a= \frac{\nabla^a t}{\sqrt{|\nabla^a t\nabla_a t|}},\,\, {g'}_{1}^a \frac{\partial}{\partial x^a}= \frac{1}{\sqrt{-g_{rr}}}\partial_{r},\\ {g'}_{2}^a\frac{\partial}{\partial x^a}=\frac{1}{\sqrt{-g_{\theta\theta}}}\partial_\theta, \,\,{g'}_{3}^a\frac{\partial}{\partial x^a}=\frac{1}{\sqrt{-g_{\varphi\varphi}}}\partial_\varphi. \end{gathered}\]
With this choice of spin-frame and trivialising the density bundle on $\Sigma$ with respect to the density $|\textrm{d}r^*\wedge\textrm{d}\Omega|^{\frac{1}{2}}$ the Dirac equation can be written $\displaystyle i\frac{\partial \Phi}{\partial t}=H\Phi$ with the operator $H$ given by:
\begin{equation} \label{def_H}\begin{gathered}H= \frac{\Delta_r\sqrt{\Delta_\theta}}{\Xi\sigma}\Gamma^1D_r + \frac{\sqrt{\Delta_r}\Delta_\theta}{\Xi\sigma}\mathfrak{D}_{S^2}+\frac{a q^2\rho^2}{\sigma^2}D_\varphi \\+\frac{\sqrt{\Delta_r\Delta_\theta}}{\sigma\sin\theta}\left(\frac{\rho^2}{\sigma}-\frac{\sqrt{\Delta_\theta}}{\Xi} \right)\Gamma^3 D_\varphi -\frac{i\sqrt{\Delta_r\Delta_\theta}\rho}{\sigma\Xi} \tilde{V_1}+  \frac{\sqrt{\Delta_r\Delta_\theta}}{\Xi\sigma}\rho m\Gamma^0. \end{gathered} \end{equation}
In the above, we have adopted similar notations to \cite{thdaude}: 
\begin{equation*} D_\varphi = -i\partial_\varphi, \quad D_r = -i\partial_r, \quad D_\theta=-i\partial_\theta, \end{equation*}
$\mathfrak{D}_{S^2}$ is the Dirac operator on the sphere $S^2$~: \[ \mathfrak{D}_{S^2} = \left(D_\theta -i\frac{\textrm{cotan}\theta}{2}\right)\Gamma^2 +\frac{D_\varphi}{\sin\theta}\Gamma^3,\]
the matrices $\Gamma^i$ are defined by~: \begin{equation*} \begin{gathered}\Gamma^0=i\left(\begin{array}{cc} 0 & I_2 \\ -I_2 & 0 \end{array} \right),\, \Gamma^1=\textrm{diag}(-1,1,1,-1), \\ \Gamma^2=\left(\begin{array}{cc}  -\sigma_x&0 \\ 0 & \sigma_x \end{array} \right), \, \Gamma^3=\left(\begin{array}{cc} \sigma_y & 0 \\ 0 &-\sigma_y \end{array} \right). \end{gathered} \end{equation*} 
Defining an operation $ c \boxtimes A$ with $c\in \mathbb{C}$ and $A=\left(\begin{array}{cc}A_1 & 0 \\ 0 & A_2 \end{array}\right)$ a block-diagonal matrix by:
\[c\boxtimes A = \left(\begin{array}{cc} cA_1 & 0 \\ 0 & \bar{c} A_2 \end{array} \right),\]
the potential $\tilde{V}_1$ can be written~:
\begin{equation*}\begin{gathered}  \tilde{V}_1 =  \tilde{F} \boxtimes \Gamma^1 + \left(\tilde{G} -\frac{\textrm{cotan}\theta\sqrt{\Delta_\theta}}{2\rho}\right)\boxtimes\Gamma^2, \end{gathered},\end{equation*}
where:
\begin{equation*}
\begin{gathered}
\tilde{F}= \frac{i\sqrt{\Delta_r}a\cos\theta}{2\rho^3} -\frac{ia\sqrt{\Delta_r}(r^2+a^2)\cos\theta\Xi}{2\sigma^2\rho} \hspace{1in}\\\hspace{2in}+\frac{\sqrt{\Delta_r}a^2\sin^2\theta}{2\rho\sigma^2(r^2+a^2)}\left(\frac{\Delta'_r}{2}(r^2+a^2)-2r\Delta_r \right),
\\
\tilde{G}=\frac{ia\Delta_\theta\sin^2\theta r + \cos\theta\rho^2\Xi - 3a^2l^2\sin^2\theta\cos\theta\rho^2}{2\sqrt{\Delta_\theta}\sin\theta\rho^3} \hspace{.8in} \\ \hspace{2in}+\frac{\sqrt{\Delta_\theta}a^2\sin\theta\cos\theta}{2\rho\sigma^2}\left( (r^2+a^2)\Xi - 2Mr\right) \\\hspace{2in}-\frac{ia\sin\theta\sqrt{\Delta_\theta}}{2\sigma^2\rho}\left(2r\Delta_r - \frac{\Delta'_r}{2}(r^2+a^2) \right),
\end{gathered}
\end{equation*}
and finally~:
\[\sigma^2=\Delta_\theta(r^2+a^2)^2-\Delta_ra^2\sin^2\theta. \]

For computational purposes it is worth noting that the operation $\boxtimes$ enjoys the following properties~:
\begin{enumerate}
\item $\boxtimes$ is distributive with respect to addition,
\item It is $\mathbb{C}$-homogenous in $A$ and $\mathbb{R}$-homogenous in $c$,
\item $(c\boxtimes A)^* = \bar{c} \boxtimes A^*$,
\item If $c\in \mathbb{R}$, $c\boxtimes A= cA$., 
\item If A is hermitian, $(-i(c\boxtimes A))^* = -i(c\boxtimes A) +2i \Re(c)A$.
\end{enumerate} 

The details of the calculation leading to this expression are sketched in Appendix~\ref{app:dirac_equation}.

 \section{Analytic framework}\label{analytic_framework}
\subsection{Symbol spaces}
We recall that $\mathscr{H}$ denotes the Hilbert space~:\[L^2(\Sigma)\otimes\bbC^4  \cong L^2(\bbR_{r^*}\times S^2)\otimes\bbC^4,\] equipped with its natural scalar product. In what follows we will study the operator $H$ defined in~\eqref{def_H} on $\mathscr{H}$ as a perturbation of another operator. In order to have a succinct language in which to distinguish the asymptotic behaviour of the coefficients of $H$, we introduce the following symbol spaces:
\begin{equation*}\Pi=\left\{ f\in C^\infty(\Sigma), \partial^{\alpha_1}_r\partial^{\alpha_2}_\theta\partial^{\alpha_3}_\varphi f\circ \psi^{-1} \in L^{\infty}(]r_e,r_+[\times S^2), \alpha_i\in \mathbb{N} \right\}.\end{equation*}
For $(m,n)\in\mathbb{N}^2$:
\begin{equation*}\bm{S}^{m,n}= \left\{ f\in C^\infty(\Sigma), \partial^{\alpha_1}_{r^*}\partial^{\alpha_2}_\theta\partial^{\alpha_3}_\varphi f \circ {\psi^*}^{-1}=\left\{\begin{array}{c} \underset{r^*\to +\infty}{O}\left( e^{-m\kappa r^*} \right) \\ \underset{r^*\to -\infty}{O} \left( \frac{1}{{r^*}^{n+\alpha_1}} \right)\end{array} \right.\alpha_i\in \mathbb{N} \right\}.\end{equation*}
$\psi$ and $\psi^*$ denote the coordinate charts $(r,\theta,\varphi)$ and $(r^*,\theta,\varphi)$ respectively and $\kappa$ is defined by:
\begin{equation} \kappa=\frac{l}{\Xi\eta_+}\sqrt{\frac{M}{r_e}}.\end{equation}
By extension, if $M \in C^\infty(\Sigma)\otimes M_4(\mathbb{C})$,we will also write $M\in \bm{S}^{m,n}$ (resp. $M\in \Pi$) if the operator norm of the matrix $M$, $||M||$, is an element of  $\bm{S}^{m,n}$ (resp. $\Pi$); this is of course equivalent to the requirement that each of its components satisfies the appropriate condition. Finally, we define:
\begin{equation}\bm{S}^{\infty,n} = \bigcap_m \bm{S}^{m,n}, \quad \bm{S}^{m,\infty}= \bigcap_n \bm{S}^{m,n}.   \end{equation}
Many of the functions $f$ at hand will be naturally expressed in the coordinate chart $\psi$, the following results will enable us to infer rapidly the asymptotic behaviour of the function when expressed in the chart $\psi^{*}$. The only missing information is the relationship between partial derivatives with respect to $r$ and those with respect to $r^*$.  From~\eqref{eq:rrstar}, one has:
\begin{equation}\partial_{r^*}=\frac{\Delta_r}{\Xi(r^2+a^2)}\partial_r. \end{equation}
So the question is settled by~:
\begin{lemme}
\label{lemme:asymptotic_behaviour_symbol}
Define the map $\alpha$ on $\Sigma$ by its coordinate expression: $\alpha\circ \psi^{-1} =\frac{\Delta_r}{\Xi(r^2+a^2)}$, then $\alpha \in \bm{S}^{2,2}. $
\end{lemme}
\begin{proof}
cf. Appendix~\ref{app:asymp_behav}.
\end{proof}
One can now use the Faà di Bruno formula\footnote{See, appendix~\ref{app:fdbformula}} to show that:
\begin{lemm}
\begin{equation} 
\label{eq:daabruno}
f\in \Pi \Rightarrow f\in \bm{S}^{0,0}, \, \partial_{r^*}f \in \bm{S}^{2,2}.
\end{equation}
In particular, if $f\in \Pi$ and $f(r^*)=\underset{r^*\to -\infty}{O}(\frac{1}{r^*})$ then $f \in \bm{S}^{0,1}$.
\end{lemm}
\subsection{$\varphi$-invariance}
\label{phi-invariance}
The metric on $B$ does not depend on the coordinate $\varphi$; this invariance will be exploited in two ways in this paper. Firstly, diagonalising $D_{\varphi}$ with anti-periodic boundary conditions, any $\phi\in \mathscr{H}$ can be represented as:

\[ \phi(r,\theta,\varphi) = \sum_{p \in \mathbb{Z}+\frac{1}{2}} \phi_p(r,\theta)e^{i p \varphi}.\]
The subspaces of this Hilbert sum are stable under the action of $H$, and we could just consider the restriction of $H$ to any such subspace; this would enable us to treat the terms with factor $D_{\varphi}$ as potentials. However, some terms contain explicit coordinate singularities. To avoid technical difficulties due to this, it is more convenient to work with the operator $H^p$ formally defined on $\mathscr{H}$ by:

\begin{multline} H^p= \frac{\Delta_r\sqrt{\Delta_\theta}}{\Xi\sigma}\Gamma^1D_r + \frac{\sqrt{\Delta_r}\Delta_\theta}{\Xi\sigma}\mathfrak{D}_{S^2}-\frac{i\sqrt{\Delta_r\Delta_\theta}\rho}{\sigma\Xi} \tilde{V_1}\\+\frac{\sqrt{\Delta_r\Delta_\theta}}{\Xi\sigma}\rho m\Gamma^0 + \frac{a q^2\rho^2}{\sigma^2}p+\frac{\sqrt{\Delta_r\Delta_\theta}}{\sigma\sin\theta}\left(\frac{\rho^2}{\sigma}-\frac{\sqrt{\Delta_\theta}}{\Xi} \right)\Gamma^3 p. \end{multline}
The function $\frac{\sqrt{\Delta_r\Delta_\theta}}{\sigma\sin\theta}\left(\frac{\rho^2}{\sigma}-\frac{\sqrt{\Delta_\theta}}{\Xi} \right)$ is well-defined and bounded, because:
\[\frac{\rho^2}{\sigma}-\frac{\sqrt{\Delta_\theta}}{\Xi}=\frac{1}{\sigma\Xi}\left(\frac{\Xi^2\rho^4 - \Delta_\theta\sigma^2}{\Xi\rho^2+\sqrt{\Delta_\theta}\sigma}\right),\]
and:
\begin{equation*}\begin{aligned} \Xi^2\rho^4-\Delta_\theta\sigma^2 =a^2\sin^2\theta\bigg(\Delta_\theta\Delta_r &+2\Xi(r^2+a^2)(l^2r^2-1)\\&+a^2\sin^2\theta(\Xi^2-l^4(r^2+a^2)^2)\bigg).\end{aligned}\end{equation*}
$H^p$ coincides with $H$ on the subspace corresponding to the eigenvalue $p \in \mathbb{Z}+\frac{1}{2}$ of $D_{\varphi}$ and the coordinate singularity is absorbed into $\mathfrak{D}_{S^2}$ which is well-defined as an operator on the sphere. 

In later analysis, it will also prove convenient to rotate the coordinate system so as to cancel some of the effects of rotation at the double horizon. 
Setting $c_0= \frac{a}{r_e^2+a^2}$, the coordinate transformation~: 
\[ t'=t,\quad r^{*'}=r^*, \quad \theta'=\theta,\quad \varphi'=\varphi-c_0 t \]
Naturally, $\varphi$ and $\varphi'$ are circular coordinates. Due to the $\varphi$-invariance of the metric, $H^p$ transforms very little under this change of coordinates, in fact, it boils down to the substitution:
\[H^p \rightarrow H^p -c_0p.\]
From now on, unless otherwise stated, we will work in the \emph{rotated} coordinates. For convenience however, we will continue to call $\varphi$ the new circular coordinate $\varphi-c_0t$. Thanks to the $\varphi$-invariance of our problem this should not cause any confusion.
\subsection{A comparison operator}
\label{comparison_operator}
Almost all the operators we will study in this paper are perturbations of a single operator $H_0$ given by:
\begin{align} 
H_0 = \Gamma_1 D_{r^*} + g(r^*)\mathfrak{D}+ f(r^*). \end{align}
The functions $g$ and $f$ satisfy:
\begin{align} g(r^*)=\frac{\sqrt{\Delta_r}}{\Xi(r^2+a^2)}\in \bm{S}^{1,1}, \quad f(r^*)=\frac{ap}{r^2+a^2}-\frac{ap}{r_e^2+a^2}\in \bm{S}^{0,1}, \end{align}
whilst, the operator $\mathfrak{D}$ is defined by:
\begin{align}
 \mathfrak{D}=\Delta_\theta^{\frac{1}{4}}\mathfrak{D}_{S^2}\Delta_\theta^{\frac{1}{4}}. 
\end{align}
The structure of this comparison operator is very similar to that of those used in~\cite{Nicolas:2004aa,thdaude}, except that, here, the angular part $\mathfrak{D}$ is a perturbation of the Dirac operator on the sphere $\mathfrak{D}_{S^2}$, rather than $\mathfrak{D}_{S^2}$ itself. The spectral properties of the latter, which are well-documented\footnote{see, for example~\cite{Abrikosov:2002aa,Camporesi:1996aa,Trautman:1993aa} }, were quite essential to the analysis in~\cite{Nicolas:2004aa,thdaude}, luckily, $\mathfrak{D}$ shares many of them.
\begin{lemme}
\label{lemme:self_adjoint_angular}
Let $S$ be the self-adjoint extension in $L^2(S^2)\otimes\mathbb{C}^2$ of the operator:
\[(D_\theta - i\frac{\cot \theta}{2})\sigma_x -\frac{D_{\varphi}}{\sin \theta}\sigma_y,\]
defined on the subset of $[C^{\infty}(S^2)]^2$ with anti-periodic boundary conditions in $\varphi$. Denoting its domain $D(S)$, $\tilde{S} = \Delta_\theta^{\frac{1}{4}}S\Delta_\theta^{\frac{1}{4}}$ is self-adjoint on $D(S)$ and has compact resolvent.
\end{lemme}
\begin{proof}
$S$ has a core consisting of smooth functions on which a simple calculation shows that:
\[\tilde{S}= \sqrt{\Delta_\theta}S -\frac{i}{2}\frac{a^2l^2\cos\theta\sin\theta}{\sqrt{\Delta_\theta}}\sigma_x. \]
The expression extends to all of $D(S)$ by continuity in the graph topology.
The estimates:
\begin{equation}
\begin{aligned} 0 \leq \sqrt{\Delta}_\theta -1 \leq \frac{\Delta_\theta -1}{\sqrt{\Delta_\theta}+1} &\leq \frac{a^2l^2}{2}, \\
\left\lVert i\sigma_x\frac{a^2l^2\cos\theta\sin\theta}{2\sqrt{\Delta_\theta}}u\right\rVert^2 &\leq \frac{a^4l^4}{4} ||u||^2, \quad u\in L^2(S^2,\mathbb{C}^2),
 \end{aligned}
 \end{equation}
together imply for $u\in D(S)$:
 \begin{equation} \label{eq:kato} \left\lVert(\sqrt{\Delta_\theta}-1)Su -i\sigma_x\frac{a^2l^2\cos\theta\sin\theta}{2\sqrt{\Delta_\theta}}u\right\rVert \leq \frac{a^2l^2}{2}\left( ||Su|| + ||u||  \right). \end{equation}
It is easy to see from~\eqref{condition_al} that $\frac{a^2l^2}{2} < 1$. Thus, by the Kato-Rellich Perturbation Theorem~\cite{Lax:2002aa, Kato:1980aa}, $\tilde{S}$ is self-adjoint on ${D}(S)$. 

In order to show that $\tilde{S}$ has compact resolvent, it suffices to show that there is a $z \in \rho(\tilde{S})$ such that $R(\tilde{S},z)$ is compact, for, by the resolvent identity, the property will follow for all $z\in \rho(\tilde{S})$. In fact, in this perturbation theory setup, it is sufficient to show that there is some $z\in \rho(S)$ such that the following inequality holds:
\begin{equation} \label{eq:resolvent_inequality} \frac{a^2l^2}{2}||R(z,S)|| + \frac{a^2l^2}{2}||SR(z,S)|| <1, \end{equation} 
where $R(z,S)$ denotes the resolvent of the operator $S$ at $z$. Indeed, assuming~\eqref{eq:resolvent_inequality}, it follows from~\eqref{eq:kato} that for any $u \in L^2(S^2,\mathbb{C}^2)$:
 \[||(\tilde{S}-S)R(z,S)u|| \leq \frac{a^2l^2}{2} ||SR(z,S)u|| + \frac{a^2l^2}{2}||R(z,S)u|| < ||u||.\]
$(\tilde{S}-S)R(z,S)$ is therefore a bounded linear operator and $I+(\tilde{S}-S)R(z,S)$ is invertible with bounded inverse. Moreover: \[ \tilde{S}-zI = S +\tilde{S}-S - zI = (I + (\tilde{S}-S)R(z,S))(S-zI).\] Consequently, $\tilde{S} -zI$ has bounded inverse given by:
\[ R(z,S)(I + (\tilde{S}-S)R(z,S))^{-1}.\]
$R(z,S)$ is compact because $S$ has compact resolvent, so $(\tilde{S}-zI)^{-1}=R(z,\tilde{S})$ is compact.

We now show there is $z\in \rho(S)$ such that~\eqref{eq:resolvent_inequality} is satisfied. By self-adjointness, it suffices to seek $z$ of the form $z=ic$. A classical resolvent estimate shows then that: $||R(z,S)||\leq \frac{1}{|c|}$ so that $||R(z,S)||$ is arbitrarily small for $|c|$ large enough. Furthermore, for any $z\in\rho(S)$ we have $||SR(z,S)|| \leq 1$, since $\frac{a^2l^2}{2}< \frac{1}{2}$,~\eqref{eq:resolvent_inequality} holds for any $|c|>2$.
\end{proof}

\begin{lemme}
\label{lemme:spectreD}
Let $\tilde{S}$ be as in Lemma~\ref{lemme:self_adjoint_angular}, the following properties hold:
\begin{itemize}
\item $-\sigma(\tilde{S})= \sigma(\tilde{S})$,
\item $\sigma(\tilde{S})\cap ]-1,1[ =\emptyset$.
\end{itemize} 
In particular, the eigenvalues $(\lambda_k)_{k\in \mathbb{Z}^*}$ can be indexed by $\mathbb{Z}^*$, in such a way that $\lambda_{-k}=-\lambda_k$ for each $k\in \mathbb{Z}^*$. Furthermore, for each $k\in \mathbb{Z}^*$, there is a subset $J_k \subset \mathbb{Z}+\frac{1}{2}$, such that for each $n\in J_k$ one can find  $\psi_{k,n}(\theta,\varphi)=\left(\begin{array}{c} \alpha_{k,n}(\theta) \\ \beta_{k,n}(\theta) \end{array}\right)e^{in\varphi}\in L^2(S^2,\mathbb{C}^2), ||\psi_{k,n}||=1$, unique up to a complex phase, satisfying $\tilde{S}\psi_{k,n}=\lambda_k\psi_{k,n}$. Necessarily, these form a total orthonormal family of eigenvectors for $\tilde{S}$. \end{lemme}
\begin{proof}
%
%
%
To prove that the spectrum of $\tilde{S}$ is disjoint from the open unit interval, it is sufficient to notice that, as a quadratic form, $\tilde{S}^2\geq 1$. Indeed, for any $u\in D(S)$:
\begin{align} (\tilde{S}u,\tilde{S} u) &= (\sqrt{\Delta_\theta}S\Delta_\theta^{\frac{1}{4}}u,S\Delta_\theta^{\frac{1}{4}}u) )\geq ||u||^2,\end{align}
because $\Delta_\theta\geq 1$. The other points will be proved in a slightly more involved case in Section~\ref{sec:Hx}. 
\end{proof}
Due to the block diagonal form of $\mathfrak{D}$, the following is an immediate consequence of the above:
\begin{corollaire}
The family:
\[\left\{ \psi^{+}_{k,n} = \left(\begin{array}{c} \psi_{k,n} \\ 0 \end{array}\right), \psi^{-}_{k,n} = \left(\begin{array}{c} 0 \\  \psi_{k,n}\end{array}\right), k \in \mathbb{Z}^*, n \in J_k\right\},\]
is a total orthonormal family of eigenvectors of $\mathfrak{D}$.
\end{corollaire}

These results are sufficient to construct a natural decomposition of $\mathscr{H}$ that can be used to obtain a convenient representation of the operator $H_0$. However, we begin by noting that the subspaces $L^2(\mathbb{R})\otimes\textrm{span}\{ \psi^+_{k,n}, \psi^-_{k,n}\}, k\in \mathbb{Z}^*, n \in J_k$, are not stable under the action of $\Gamma^1$. Indeed, if $\psi$ is an eigenvector with eigenvalue $\lambda$ of $\mathfrak{D}$, then, since $\Gamma^1$ anti-commutes with $\Gamma^2$ and $\Gamma^3$, $\Gamma^1\psi$ is an eigenvector with eigenvalue $-\lambda$. In particular, the block diagonal form of $\Gamma^1$ implies that $\Gamma^1\psi_{k,n}^\pm$ and $\psi_{-k,n}^\pm$ must be colinear (because $\psi_{k,n}$ is unique up to scaling). In fact, $\Gamma^1$ being unitary and symmetric, one has $\Gamma^1\psi_{k,n}^\pm=\pm\psi_{-k,n}^\pm$. The family $\psi_{k,n}$ remains total and orthonormal if $\psi_{-k,n}$ is rescaled to absorb the sign, so one can assume that: $\Gamma^1\psi_{k,n}^\pm=\psi_{-k,n}^\pm$.
The subspaces: \[\mathscr{H}_{k,n}=L^2(\mathbb{R})\otimes\textrm{span}\left\{ \psi^+_{k,n}, \psi^+_{-k,n}, \psi^{-}_{k,n}, \psi^{-}_{-k,n} \right\}, k \in \mathbb{N^*}, n \in J_k,\] 
are then naturally stable under $\Gamma^1$ and therefore, under $H_0$, and $\displaystyle \mathscr{H}= \overset{\perp}{\bigoplus_{k,n}} \mathscr{H}_{k,n}$. 
For each $(k,n)$, $\mathscr{H}_{k,n}$ can be isometrically identified to $ [L^2(\mathbb{R})]^4$ by the map:
\begin{equation}\label{def_func_b_kn} \begin{array}{rcl} b_{k,n}: \hspace{.5in} \mathscr{H}_{k,n} &\longrightarrow& [L^2(\mathbb{R})]^4 \vspace{.1in} \\  \begin{split} u_1\psi^{+}_{k,n} + u_2 \psi^{+}_{-k,n} \\+ u_3 \psi^{-}_{k,n} + u_4 \psi^{-}_{-k,n} \end{split} &\longmapsto& \frac{1}{\sqrt{2}}\left( \begin{array}{c} u_1 - u_2 \\ u_1 + u_2 \\ u_3+u_4 \\ u_3-u_4 \end{array} \right) \end{array}. \end{equation}
Through this identification the restriction, $H_0^{k,n}$, of $H_0$ to $\mathscr{H}_{k,n}$ can be written:
\begin{equation} \label{eq:decompH0} H_0^{k,n}=\Gamma^1D_{r^*} -\lambda_{k,n}g(r^*)\Gamma^2 + f(r^*). \end{equation}
This is clearly a bounded perturbation of the self-adjoint operator $\Gamma^1D_{r^*}$ with domain $[H^1(\mathbb{R})]^4$, hence it is self-adjoint on the same domain.

We are now ready to use the lemma below\footnote{see \cite[Lemma~3.5]{Nicolas:2004aa} and Appendix~\ref{app:proof_lemme_fin_2}} to obtain a description of a domain where the formal expression for $H_0$ is self-adjoint.
\begin{lemme}
\label{lemme-decomposition}
Let $X$ be a Hilbert space and $(X_n)_{n \in \mathbb{N}}$ a family of subspaces of $X$ such that:
\begin{equation*} X=\underset{{n\in\mathbb{N}}}{\overset{\perp}{\bigoplus}} X_n, \end{equation*}
where the sum is topological. Let $(A_n)_{n \in \mathbb{N}}$ be a sequence of operators $A_n$ on $X_n$, such that for each $n$, $A_n$ is self-adjoint on its domain $D(A_n)$. Then the operator $A$ defined by: \[Ax=\sum_n A_nx_n,\] if $x=\sum x_n,  x_n \in X_n$ for any $n\in\mathbb{N}$ is self-adjoint on: \[D(A) = \left\{ x= \sum_n x_n \in X, \sum_{n\in\mathbb{N}} ||A_nx_n||^2 < \infty \right\}.\]
\end{lemme}
%

The natural domain for $H_0$, which is always meaningful in the distributional sense, is certainly $\{ u \in \mathscr{H}, H_0 u \in \mathscr{H} \}$, this, in fact, coincides with the domain of the operator given by the previous lemma: \[D(H_0)= \{ u= \sum_{k,n} u_{k,n} \in \mathscr{H}, \sum_{k,n} ||H_0^{k,n}u_{k,n}||^2 < \infty \}.\]  The proof is analogous to that of \cite[Lemma 3.5]{Nicolas:2004aa}.

Since for each $k\in\mathbb{N^*},n\in J_k$ , $D\left(H_0^{k,n}\right)$ is isometric to $[H^1(\mathbb{R})^4]$, and  $\mathscr{S}(\mathbb{R})$ is dense in $H^1(\mathbb{R})$, we deduce immediately a core for $H_0$, that we will simply denote by $\mathscr{S}$. This core will be convenient for many computations, in particular, it will justify the use of the Leibniz rule when computing commutators. More precisely:
\begin{lemme}
\label{lemme:core_S}
$\displaystyle \mathscr{S}=\overset{\perp}{\bigoplus_{k,n}}\mathscr{S}(\mathbb{R})\otimes\textrm{span}\left\{ \psi^+_{k,n}, \psi^+_{-k,n}, \psi^{-}_{k,n}, \psi^{-}_{-k,n} \right\}$ is a core for $H_0$.
\end{lemme}
\begin{proof}cf. Appendix~\ref{app:proof_lemme_fin_2}\end{proof}


\subsection{Short and long-range potentials}
The construction of the wave operators, modified or not, will mainly be based on Cook's method\footnote{See for example \cite[Chapter~37]{Lax:2002aa} } or minor variations thereof. Because of this, it will be interesting to investigate the integrability of the matrix-valued coefficients appearing in our differential operators. Amongst those, we will call ``potentials'', the parts of the order $0$ component of its symbol that vanish on the horizons. For our purposes, they will be split into merely three groups. Namely a potential $V$ is: \begin{itemize} \item short-range at $+ \infty$ (resp. $-\infty$) if:
 \begin{align}\label{eq:potentiels} \underset{r^*\geq 0, \vartheta\in S^2}{\sup} ||\langle r^*\rangle^\alpha V|| < + \infty \quad \textrm{(resp.$\underset{r^*\leq 0, \vartheta\in S^2}{\sup} ||\langle r^*\rangle^\alpha V|| < + \infty$)} \end{align} for some $\alpha>1$, \item long-range otherwise, \item of Coulomb-type at $+\infty$ (resp. $-\infty$) if $V$ is long-range there and~\eqref{eq:potentiels} holds with $\alpha=1$. \end{itemize}
The norm here is the operator norm on $M_4(\mathbb{C})$ and $\bra{.}$ denotes the Japanese bracket $\bra{r}=\sqrt{r^2+1}$. 
In relation with the symbol spaces we introduced previously, let $m,n\in \mathbb{Z}$ and suppose $V\in \bm{S}^{m,n}$, then:
\begin{itemize}
\item$m\geq 1$ $\Rightarrow$ $V$ short-range at $+\infty$,
\item$n\geq 2$ $\Rightarrow$ $V$ short-range at $-\infty$,
\item$n=1$ $\Rightarrow$ $V$ of Coulomb type at $-\infty$. 
\end{itemize}
\subsection{Self-adjointness of $H^p$}
\label{self-adjointness}
It is now relatively easy to prove the self-adjointness of $H^p$, we first introduce the function:
\begin{equation}\label{eq:defh} h(r,\theta)= \Delta_\theta^{\frac{1}{4}}\sqrt{\frac{r^2+a^2}{\sigma}}, \end{equation}
it satisfies the following properties:
\begin{align}
&|h^2-1|\leq 1-a^2l^2<1,\\
&  \label{eq:hprop2} \partial_\theta h = \Delta_r\frac{(r^2+a^2)a^2\sin\theta\cos\theta \Xi}{2h\sqrt{\Delta_\theta}\sigma^3} \in \textbf{S}^{2,2}. \end{align}
\begin{proof} 
The first property follows from the following chain of inequalities:
\begin{equation*}
\begin{aligned}
0 \leq h^2-1 &= \frac{\Delta_r a^2\sin^2\theta}{\sigma\left(\sigma +\sqrt{\Delta_\theta}(r^2+a^2)\right)}\\ &\leq \frac{\Delta_r a^2\sin^2\theta}{\sigma^2} \\
&\leq \frac{a^2}{r^2} \\&\leq \frac{a^2}{r_e^2}=\frac{6a^2l^2}{1-a^2l^2-\sqrt{(1-a^2l^2)^2-12a^2l^2}}
\leq 1-a^2l^2.\end{aligned}
\end{equation*}
By Equation~\eqref{condition_al}, $1-a^2l^2<1$, the conclusion follows.
\end{proof}
The boundedness of $\partial_{r^*} h=\frac{\Xi\Delta_r}{r^2+a^2}\partial_{r}h$ and $\partial_\theta h$ shows that $h\in B(D(H_0))$. Indeed, $[H_0,h]$ is defined on $D(H_0)$ and:
\begin{align*} [H_0,h]u=-i\Gamma^1\partial_{r^*} hu -i \frac{\sqrt{\Delta_\theta\Delta_r}}{\Xi(r^2+a^2)}\Gamma^2\partial_\theta h u, \quad u \in D(H_0).
\end{align*}
Consequently, for any $u\in {D}(H_0)$: \begin{equation} ||H_0hu|| \leq ||hH_0u|| + ||[H_0,h]u|| \leq C(||H_0u||+||u||),\end{equation}
for some constant $C\in \mathbb{R}_+^*$. The following relationship between $H_0$ and $H^p$ is therefore meaningful:  
\begin{align}\label{eq:relH} H^p= hH_0h  + V_C + V_S, \end{align}
with:
\begin{multline}
\label{eq:defvs}
V_S =-\frac{ap\sqrt{\Delta_\theta}}{\sigma} + \frac{a\Delta_\theta(r^2+a^2)p}{\sigma^2} - a\frac{\Delta_rp}{\sigma^2} +\frac{ap(h^2-1)}{r_e^2+a^2} \\+i\left[ \left(\frac{ia\Delta_\theta\sqrt{\Delta_r}}{2\sigma^3\Xi}\left( 2r\Delta_r - \frac{\Delta_r'}{2}(r^2+a^2)\right) \right)\boxtimes\Gamma^2 \right] \\-i\left[ \left(\frac{i\Delta_r\sqrt{\Delta_\theta}a\cos\theta}{2\rho^2\sigma\Xi} -\frac{ia\Delta_r\sqrt{\Delta_\theta}(r^2+a^2)\cos\theta}{2\sigma^3} \right)\boxtimes\Gamma^1\right],
\end{multline}
\begin{multline}
\label{eq:defvc}
V_C =\frac{\sqrt{\Delta_r\Delta_\theta}}{\sigma\sin\theta}\left(\frac{\rho^2}{\sigma}-\frac{\sqrt{\Delta_\theta}}{\Xi}\right)\Gamma^3p +\frac{\sqrt{\Delta_r\Delta_\theta}}{\Xi\sigma}\rho m\Gamma^0\\-i\left[ \left(\frac{ia\sqrt{\Delta_r}\sin\theta r \Delta_\theta}{2\rho^2\sigma\Xi}\right)\boxtimes\Gamma^2 \right].
\end{multline}
In the above, we have sorted the terms according to their asymptotic behaviour at $-\infty$, since at $+\infty$ all the potentials are short-range. More precisely, the terms in $V_S$ are short-range at $-\infty$ and those of $V_C$ are of Coulomb-type there. \eqref{eq:hprop2} means that $h[H_0,h]$ is short-range at both infinities.  

Using Equation~\eqref{eq:relH}, one now shows that:
\begin{lemme}
$H^p$ is self-adjoint on $D(H_0)$, for any $p\in \mathbb{Z}+\frac{1}{2}$.
\end{lemme}
\begin{proof}
It follows from Equation\eqref{eq:relH} that:
\[H^p= H_0+(h^2-1)H_0 +h[H_0,h] + V_C + V_S,\]
since $[H_0,h]$ is bounded,  $H^p$ is $H_0$-bounded and, using the fact that : \[|h^2-1|\leq 1-a^2l^2<1,\] the result follows from the Kato-Rellich Perturbation Theorem.
\end{proof}
\subsection{Further properties of $H_0$}
We briefly quote the following important properties of the simplified operator $H_0$~:
\begin{lemme}
\label{lemme:equiv_quadratic_form}
As quadratic forms on $\mathscr{S}$, $H_0^2$ and $Q=D_{r^*}^2 + g^2(r^*)\mathfrak{D}^2$ are equivalent.
\end{lemme}
Lemma~\ref{lemme:equiv_quadratic_form} has the following important consequences:
\begin{corollaire}
\label{critere_compact}
$D(H_0) \subset H^1_{\textrm{loc}}$ continuously and we have the following criterion for compactness\footnote{The criterion is a consequence of the Rellich-Kondrachov theorem. See for example~\cite{Evans:2010aa}.}:
\center{\fbox{If $f,\chi \in C_{\infty}$ then $f(r^*)\chi(H_0)$ is compact.}}
\end{corollaire}
\begin{corollaire}
\label{corollaire:DrDbounded}
$\Gamma^1D_{r^*}$ and $g(r^*)\mathfrak{D}$ are elements of $B(D(H_0), \mathscr{H})$.
\end{corollaire}
The relationship between the operators $Q$ and $H_0^2$ goes even further, one can show that: \begin{equation} \label{eq:domaine_d2_H}\boxed{D(H_0^2)=D(Q)} \end{equation}
The proof of these results is an adaptation of that of~\cite[Lemmata~4.3,~4.4,~4.6, Corollaries~4.1, 4.2 ]{Nicolas:2004aa}.

\section{Mourre theory}\label{mourre}
\subsection{Brief overview}
\label{subsec:mourre_overview}
Mourre theory is a very powerful tool for constructing analytical scattering theories. It has been used in many different situations including the quantum $N$-particle problem~\cite{Derezinski:1997aa} and for scattering of classical fields - with or without spin- in a range of black-hole type geometries~\cite{Hafner:2003aa,thdaude,Nicolas:2004aa}. The theory has been refined since E. Mourre's original article~\cite{Mourre:1981aa} following, in particular, the theoretical developments in~\cite{Amrein:1996aa}. There, it is discussed that one can substitute a certain regularity condition for some of the technical conditions in Mourre's original work. We present here a non-optimal ``working'' version of the theory.

We begin by making precise the aforementioned regularity condition:
\begin{definition}
\label{def:regularity_class}
Let $A,H$ be two self-adjoint operators on a Hilbert space $\mathscr{H}$. We will say that $H\in C^1(A)$ if for any $u\in \mathscr{H}$ the map $s \mapsto e^{isA}(H-z)^{-1}e^{-isA}u $ is of class $C^1$ for a (and therefore all) $z\in \rho(H)$.
\end{definition}

In other words, Definition~\ref{def:regularity_class} states that, in a certain sense, the resolvent of $H$ evolves smoothly under the action of $A$\footnote{This interpretation fits nicely into the Heisenberg picture, where operators evolve instead of the wave function}.  An interesting technical consequence of this regularity is that (in the form sense) the following equation makes sense on $\mathscr{H}$. 
\[[A,(H-z)^{-1}]=(H-z)^{-1}[H,A](H-z)^{-1},\]
we refer to~\cite{Amrein:1996aa} for more details.
\begin{definition}
A pair $(A,H)$ of self-adjoint operators on a Hilbert space $\mathscr{H}$ such that $H\in C^1(A)$ will be said to satisfy a Mourre estimate (with compact error) on some energy interval $I \subset \mathbb{R}$ if there is a compact operator $K$ and a strictly positive constant $\mu$ such that:
\begin{equation*}\textbf{1}_I(H)i[H,A]\textbf{1}_I(H) \geq \mu \textbf{1}_{I}(H) + K. \end{equation*}
This will be written more briefly:
\begin{equation} \label{eq:mourre_estimate}\textbf{1}_I(H)i[H,A]\textbf{1}_I(H) \geqk \mu \textbf{1}_{I}(H).\end{equation}
\end{definition}

The heart of Mourre theory is contained in the following theorem; the statement here differs from that in Mourre's original article~\cite{Mourre:1981aa}; here we follow~\cite{Daude:2010aa,Nicolas:2004aa}.
\begin{thm}[Mourre]
\label{thm:mourre}
Suppose that:
\begin{enumerate}
\item $H\in C^1(A)$,
\item $i[H,A]$ defined as a quadratic form on $D(H)\cap D(A)$ extends to an element of $B(D(H),\mathscr{H})$,
\item $[A,[A,H]]$ defined as a quadratic form on $D(H)\cap D(A)$ extends to a bounded operator from $D(H)$ to $D(H)^*$.
\item $(A,H)$ satisfy a Mourre estimate on $I\subset \mathbb{R}$ 
\end{enumerate}
Then, $H$ has no singular continuous spectrum in $I$, and $H$ has at most a finite number of eigenvalues, counted with multiplicity, in $I$.
\end{thm}
When a pair $(A,H)$ satisfy the conditions of Theorem~\ref{thm:mourre}, $A$ will be said to be a \textbf{conjugate operator} for $H$ on $I$. 

\subsection{Our conjugate operators} 
\label{section:choice_conjugate_operators}
We will now proceed to describe our choice of conjugate operators for $H_0$ and a class of perturbations of $H_0$ that will include $H^p, p \in \mathbb{Z} +\frac{1}{2}$. Mourre theory is very flexible in that the notion of conjugate operator is local in energy but also, using cut-off functions, in space; this is well-illustrated in~\cite{thdaude,Nicolas:2004aa}. As a consequence, determining a candidate for the conjugate operator of a given operator $H$ can be a very creative process, although in many examples from physics, the generator of dilatations, or minor variations thereof, is usually a good candidate. We will see that, despite the extreme blackhole geometry, our case is no exception. As in~\cite{thdaude}, the full conjugate operator will be a combination of two operators $A_+$ and $A_-$ tailored to deal with the distinct natures of the geometry at the two asymptotic ends. Throughout the sequel we separate the two infinities using smooth cut-off functions, $j_+,j_-,j_1$ satisfying:
\begin{equation}
\label{eq:cutoff_def}
\left\{\begin{array}{cccc}
j_-(t)=1 & \textrm{if } t \leq -2, & j_-(t)=0 &\textrm{if } t \geq -\frac{3}{2}, \\
j_+(t)=1 &\textrm{if } t \geq -\frac{1}{2}, & j_+(t) = 0 &\textrm{if } t \leq -1,\\
j_1(t)=1 & \textrm{if } t \geq -1, & j_1(t)=0 & \textrm{if } t \leq -\frac{3}{2}.
 \end{array} \right.
\end{equation}
$j_-$ and $j_1$ should be chosen such that their supports are disjoint.

\subsubsection{At the simple horizon}
Near $\mathscr{H}_{r_+}$, we will follow the treatment in~\cite{thdaude} and set:
\begin{equation}A_+(S)=R_+(r^*,\mathfrak{D})\Gamma^1,\end{equation}
where:
\begin{align} R_+(r^*,\mathfrak{D})=(r^*-\kappa^{-1}\ln|\mathfrak{D}|)j_+^2\left(\frac{r^*-\kappa^{-1}\ln|\mathfrak{D}|}{S}\right).\end{align}
Since $|\mathfrak{D}|\geq 1$, the same arguments in the proof of~\cite[Lemma IV.4.4]{thdaude} can be used to show that:
\begin{lemme}
\label{lemme:r+bounded}
For any $S\geq 1$, uniformly in $\lambda_k$, $k\in \mathbb{N}^*$:
\begin{align}
|R_+(r^*,\lambda_k)|&\leq C\bra{r^*}.
\end{align}
In the above, $C$ is a positive constant and $R_+(r^*,\lambda_k)$ denotes the restriction of $R_+(r^*,\mathfrak{D})$ to $\mathscr{H}_{k,n}$.
\end{lemme}
Despite the strange argument in the cut-off function, this choice is surprisingly simple and is essentially: $\Gamma^1 r^*$. This is motivated by the observation that, under the unitary transformation: $U=e^{-i\kappa^{-1}\ln(|\mathfrak{D}|)D_{r^*}}$, the toy model on $\mathbb{R}_+\times S^2$~ given by: \[\slashh=\Gamma^1D_{r^*} + e^{-\kappa r^*}\mathfrak{D} + c,\]
transforms to:
\[\hat{\slashh} = \Gamma^1D_{r^*} +e^{-\kappa r^*}\frac{\mathfrak{D}}{|\mathfrak{D}|}+c. \]
The commutator with $\Gamma^1r^*$ is then easily seen to be:
\[i[\hat{\slashh},\Gamma^1r^*]=1+2r^*e^{-\kappa r^*}\frac{\mathfrak{D}}{|\mathfrak{D}|}\Gamma^1.\]
Restricting to a compact energy interval using $\chi(H)$, $\chi\in C^\infty_0(\bbR)$, the second term will lead to a compact error by Corollary~\ref{critere_compact}.
Note that without the unitary transformation $U$ the commutator is:
\[i[\slashh,\Gamma^1r^*]=1+2r^*e^{-\kappa r^*}\mathfrak{D}\Gamma^1.\]
Here the second term is problematic, as $r^*e^{-\kappa r^*}$ does not decay faster than $e^{-\kappa r^*}$ and hence we cannot control $||r^*e^{-\kappa r^*}\mathfrak{D}||$ with $||e^{-\kappa r^*}\mathfrak{D}||$.
\subsubsection{Near the double horizon}
Let us start our discussion at $\mathscr{H}_{r_e}$ by motivating the coordinate transformation we performed in Section~\ref{phi-invariance}. 

At the double horizon ($r^*\to - \infty$), the function $g$ appearing in the expression for $H_0$ decays as $O\left(\frac{1}{-r^*}\right)$. This is significantly slower than the exponential decay at a simple horizon, and is similar to the behaviour at space-like infinity in an asymptotically flat spacetime. In fact, when $r^*\to -\infty$ the principal symbol of $H_0$ formally ressembles: $$\tilde{\slashh}=\Gamma^1D_{r^*} -\frac{C}{r^*}\mathfrak{D},$$ 
which \emph{is} the massless Dirac operator (for the spinor density) for the asymptotically flat metric on $\mathbb{R}_-^*\times S^2$: 
\[\eta=dt^2 -d{r^*}^2 -\left(\frac{r^*}{C}\right)^2\frac{1}{\Delta_\theta}\textrm{d}\sigma^2.\]
This suggests that we should try to treat the double horizon in a similar manner to space-like infinity, and in particular that $A=\frac{1}{2}\{D_{r^*},r^*\}$ should be a reasonable candidate for a conjugate operator there; indeed, \begin{equation}\label{eq:toy_commutator_double_horizon} i[\tilde{\slashh},A]=\tilde{\slashh}.\end{equation}

However, had we used the original Boyer-Lindquist like coordinates $(t,r,\theta,\varphi)$, near $r^*\to -\infty$, we would have been lead to set: 
\[\tilde{H}_0 = \Gamma^1 D_{r^*} +g(r^*)\mathfrak{D} + \tilde{f}(r^*),\]
where $\tilde{f} \in \bm{S}^{0,0}$ and $\displaystyle \lim_{r^*\to -\infty} \tilde{f}(r^*)=c_0=\frac{a}{r_e^2+a^2}$.
The corresponding toy model would hence be: $\tilde{\slashh}+c_0$. Since $A$ commutes with constants, we need to modify it to generalise Equation~\eqref{eq:toy_commutator_double_horizon}. This can be achieved simply by appending $\Gamma_1c_0r^*$ to $A$. However, in doing so, we are immediately confronted to similar issues (that are carefully avoided by the unitary transformation $U$) described above at the simple horizon. The solution relies on the morphism properties of $\exp$ and the fact that $r^*e^{k_+r^*}=\underset{r^*\to-\infty}{o}(1)$. In our situation, even if we can imagine trying to exploit the morphism properties of $t\mapsto \frac{1}{t}$, with a unitary transformation such as $\tilde{U}=e^{-\frac{i}{2}\ln|\mathfrak{D}|\{D_{r^*},r^*\}}$, the error may not be compact simply because there is no decay left ! The coordinate change performed in Section~\ref{phi-invariance} circumvents the problem entirely by shifting the potential to the simple horizon, where we know how to treat it. In the sequel we set:
\begin{equation}A_-(S)= \frac{1}{2} \{ R_-(r^*),D_{r^*}\}, \end{equation}
where,
\begin{equation}R_-(r^*)=j_-^2(\frac{r^*}{S})r^*, \end{equation}
$\{\,\cdot\,,\,\cdot\,\}$ denotes the anti-commutator and $S\geq 1$ is a real parameter.

The conjugate operator $A_I$ will vary depending on the energy interval $I$, in fact we will show that there is $S_I \in [1,+\infty)$ such that on $I$ either:
\begin{equation}
\begin{aligned}
A_+(S_I)+A_-(S_I) \quad  \textrm{if }I\subset (0,+\infty),\\
A_+(S_I)-A_-(S_I) \quad \textrm{if }I \subset (-\infty,0),
\end{aligned}
\end{equation}
is a conjugate operator on $I$.

\subsection{The technical conditions}
Despite being the key assumption in Mourre theory, the estimate~\eqref{eq:mourre_estimate} alone is not sufficient for the conclusion of Theorem~\ref{thm:mourre}. However, checking the more technical conditions of Theorem~\ref{thm:mourre} is quite involved and diverts from the core of the text without offering much insight. For this reason, we have chosen to put the outline of the proof in Appendix~\ref{app:proof_mourre}. We nevertheless record here the appropriate conclusions of these developments:
\begin{prop}
\label{prop:technical_mourre1}
For any $S\geq 1$, $A_\pm(S)$ and $A_+(S) \pm A_-(S)$ are essentially self-adjoint on $\mathscr{S}$.
\end{prop}
\begin{prop}
\label{prop:technical_mourre2}
Let $H$ be an operator on $\mathscr{H}$ defined by:
\begin{equation} \label{eq:formeH} H=hH_0h +V, \end{equation}
where\footnote{These assumptions are not optimal}:
\begin{itemize}\item $V$ is a matrix-valued potential such that $V\in \bm{S}^{1,1}$ \item $h\in C_b^{\infty}(\mathbb{R}_{r^*}\times]0,\pi[)$ such that $h>0$, $|h^2-1|\leq c<1$, $\partial_{r^*} h, \partial_\theta h$, $h^2-1\in \bm{S}^{1,1}$. \end{itemize} 
Any such operator is self-adjoint on $\mathscr{H}$ with domain $D(H_0)$ by the Kato-Rellich theorem.
Furthermore for any $A\in \{A_\pm(S), A_+(S)\pm A_-(S)\}$:
\begin{enumerate}
\item The quadratic forms $i[H,A]$ and $i[[H,A],A]$ on $D(H)\cap D(A)$ extend to elements of $\mathcal{B}(D(H),\mathscr{H})$,
\item $H\in C^2(A)$.
\end{enumerate}
\end{prop}
We record here for future reference a further property of the operators $H$ in the preceding proposition:
\begin{lemme}\label{lemme:domaine_H2}
$D(H^2)=D(H_0^2).$
\end{lemme}
For a proof we refer to~\cite[Lemma~4.6]{Nicolas:2004aa}.
The functions $h$ and $V=V_C+V_S$ of $H^p$ satisfy slightly better conditions than those above:
\begin{lemme}\label{lemme:behaviour_potentials}\phantom{a}
\begin{itemize}
\item The function $h$ defined by Equation~\eqref{eq:defh} satisfies: \[h^2-1,\partial_{r^*}h,\partial_\theta h \in \bm{S}^{2,2}.\]
\item Let $V_C$ and $V_S$ be as in Equations~\eqref{eq:defvc} and~\eqref{eq:defvs} then:
\[V_C \in \bm{S}^{1,1}, V_S \in \bm{S}^{1,2}. \]
\end{itemize}
\end{lemme}

\subsection{Mourre estimates for $H_0$} 
We shall now move on to derive Mourre inequalities, naturally, we will treat $\mathscr{H}_{r_e}$ and $\mathscr{H}_{r_+}$ separately.
\subsubsection{Near the double horizon}
We begin with:
\begin{lemme}
\label{step1_mourrein}
Let $\chi \in C^{\infty}_0(\mathbb{R})$ then for any $S\geq 1$;
\begin{equation}\chi(H_0)i[H_0,A_-(S)]\chi(H_0) \eqk \chi(H_0)j_-(\frac{r^*}{S})H_0j_-(\frac{r^*}{S}) \chi(H_0), \end{equation}
where $\eqk$ is used to denote equality up to a compact error.
\end{lemme}
\begin{proof} 
One has:
\begin{equation}
\begin{aligned} i[H_0,A_-(S)]=&\Gamma^1 R'_-(r^*)D_{r^*} - \frac{i}{2}\Gamma^1R^{''}_-(r^*) \\&- R_-(r^*)g'(r^*)\mathfrak{D}-R_-(r^*)f'(r^*) \\ =& j_-(\frac{r^*}{S})\Gamma^1D_{r^*}j_-(\frac{r^*}{S}) + 2r^*j_-(\frac{r^*}{S})j'_-(\frac{r^*}{S})\Gamma^1 D_{r^*}  \\ &\quad\underline{-\frac{i}{S}j'_-(\frac{r^*}{S})j_-(\frac{r^*}{S})}-\underline{\frac{ir^*}{S^2}(j'_-(\frac{r^*}{S}))^2}\\ &-\underline{\frac{i r^*}{S^2}j^{''}_-(\frac{r^*}{S})j_-(\frac{r^*}{S})}-R_-(r^*)\left(g'(r^*)\mathfrak{D} +f'(r^*)\right).\end{aligned}\end{equation}
Note that if $0 \leq \chi \leq 1$ is a smooth function with compact support in $\mathbb{R}$, since $j'$ has compact support, Corollary~\ref{critere_compact} implies that the terms underlined above will only lead to compact terms in $\chi(H_0)i[H_0,A_-(S)]\chi(H_0)$, consequently:
\begin{multline}\chi(H_0)i[H_0,A_-(S)]\chi(H_0) \eqk \chi(H_0)\bigg( j_-(\frac{r^*}{S})\Gamma^1D_{r^*}j_-(\frac{r^*}{S}) \\\hspace{1.5in}+2r^*j_-(\frac{r^*}{S})j'_-(\frac{r^*}{S})\Gamma^1 D_{r^*} \\-R_-(r^*)\left(g'(r^*)\mathfrak{D} +f'(r^*)\right)\bigg)\chi(H_0). \end{multline}
Using Corollary~\ref{corollaire:DrDbounded}, one can show that $2r^*j_-(\frac{r^*}{S})j'_-(\frac{r^*}{S})\Gamma^1 D_{r^*}\chi(H_0)$ is also compact.
Indeed, let $\gamma(r^*)=2r^*j_-(\frac{r^*}{S})j'_-(\frac{r^*}{S})$ and note that $\gamma \in C^\infty_0(\mathbb{R})$. For any $u \in \mathscr{H}$, one has: \begin{equation} \gamma(r^*) \Gamma^1D_{r^*}\chi(H_0)u = \Gamma^1D_{r^*}\gamma(r^*)\chi(H_0)u +i\Gamma^1\gamma'(r^*)\chi(H_0)u. \end{equation}
Corollary~\ref{corollaire:DrDbounded} implies that there is $C_1>0$ such that for any $u\in D(H_0)$. 
 \[||\Gamma^1D_{r^*}u|| \leq C_1(||H_0u||+||u||).\]
Hence:
\begin{equation*}
\begin{aligned}
|| \gamma(r^*)\Gamma^1D_{r^*}\chi(H_0)u || &\leq ||\Gamma^1 D_{r^*}\gamma(r^*)\chi(H_0)u|| + ||\gamma'(r^*)\chi(H_0)u||,
\\ &\leq C_1||H_0\gamma(r^*)\chi(H_0)u|| + C_1||\gamma(r^*)\chi(H_0)u|| \\&\hspace{1.5in}+ ||\gamma'(r^*)\chi(H_0)u||,
\\ &\leq C_1||\gamma(r^*)H_0\chi(H_0)u|| + C_1||\gamma(r^*)\chi(H_0)u||  \\&\hspace{1in}+ (1+C_1)||\gamma'(r^*)\chi(H_0)u||.
\end{aligned}
\end{equation*}
According to Corollary~\ref{critere_compact} the operators $\gamma(r^*)H_0\chi(H_0)$, $\gamma(r^*)\chi(H_0)$ and $\gamma'(r^*)\chi(H_0)$ are all compact and so it follows from a simple extraction argument that $\gamma(r^*)\Gamma^1D_{r^*}\chi(H_0)$ must be too.
Thus:
\begin{multline}\label{eq:step1_mourre}\chi(H_0)i[H_0,A_-(S)]\chi(H_0) \eqk \chi(H_0) j_-(\frac{r^*}{S})\Gamma^1D_{r^*}j_-(\frac{r^*}{S}) \\-R_-(r^*)\left(g'(r^*)\mathfrak{D} +f'(r^*)\right)\chi(H_0).\end{multline}
Now,~\eqref{eq:step1_mourre} can be rewritten:
\begin{equation*}\begin{aligned}\chi(H_0)i[H_0,A_-(S)]\chi(H_0)\eqk &\chi(H_0)j_-(\frac{r^*}{S})H_0j_-(\frac{r^*}{S})\chi(H_0)\\ &-\chi(H_0)j^2_-(\frac{r^*}{S})\left( g(r^*)+r^*g'(r^*)\right)\mathfrak{D}\chi(H_0) \\&- \chi(H_0)j_-^2(\frac{r^*}{S})\left(f(r^*)+ r^*f'(r^*)\right)\chi(H_0).\end{aligned}\end{equation*}
Since $f(r^*) +r^*f'(r^*) \to 0$ when $r^*\to -\infty$, it follows from Corollary~\ref{critere_compact} that the terms in the last line of the previous equation are compact. The compactness of those on the middle line is also a consequence of Corollary~\ref{critere_compact}, because near the double horizon $r_* \to - \infty$ $(r\to r_e)$ one has:
\begin{align*} r^*g'(r^*)+g(r^*) = \left(1+ \frac{r^*}{\Xi(r_e^2+a^2)}\frac{\Delta_r'}{2} + O\left(\frac{1}{r^*}\right)\right)g(r^*), \end{align*}
and:
\begin{align*} 
\Delta_r'&= 2l^2(r-r_e)(r_e-r_-)(r_+-r_e) + O((r-r_e)^2),
\\&= -2\frac{(3Mr_e-4a^2)(r-r_e)}{r_e^2} + O((r-r_e)^2).
\end{align*}
Using~\eqref{eq:asymp_double} we obtain that:
\begin{equation*} \Delta_r'= -2\frac{(r_e^2+a^2)\Xi }{r^*}+o(\frac{1}{r^*}).  \end{equation*}
From which it follows:
\begin{equation} r^*g'(r^*)+g(r^*)= o(g(r^*)). \end{equation}
Therefore, there is a continuous function $\varepsilon \in C_\infty(\mathbb{R})$ such that:
\begin{align*} ||j_-^2(\frac{r^*}{S})(r^*g'(r^*)+g(r^*))\mathfrak{D}\chi(H_0)|| &= ||g(r^*)\mathfrak{D}\varepsilon(r^*)\chi(H_0)||, \\&\leq ||H_0\varepsilon(r^*)\chi(H_0)|| + ||\varepsilon(r^*)\chi(H_0)||. \end{align*}
Compactness then follows with a similar argument as before.
\end{proof}

We are now ready to prove:
\begin{prop}
\label{prop:mourre_estimate_DH}
Let $\chi$ be of a compact support contained in $(0, +\infty)$ and $\mu>0$ be such that $\supp \chi \subset [\mu, +\infty)$ then for any $S\geq 1$:
\begin{equation}\label{eq:m_ineq_1} \chi(H_0)i[H_0,A_-(S)]\chi(H_0) \geqk \mu \chi(H_0)j_-^2(\frac{r^*}{S})\chi(H_0).\end{equation}
\end{prop}
The result holds also if $\supp\chi \subset (-\infty,0)$, if we replace $A_-(S)$ by $-A_-(S)$.
\begin{proof}
Using Lemma~\ref{step1_mourrein}, it is sufficient to prove that:
\[\chi(H_0)j_-(\frac{r^*}{S})H_0j_-(\frac{r^*}{S}) \chi(H_0) \geqk \mu\chi(H_0)j_-^2(\frac{r^*}{S})\chi(H_0).\]
Our first step is to note that, although $\chi(H_0)$ and $j_-(\frac{r^*}{S})$ do not commute, their commutator is a compact operator. This can be seen using the Helffer-Sjöstrand formula~\cite[Proposition 7.2]{Helffer:1987aa}, for one has:
\begin{multline}\label{eq:use_hs_1} \left[\chi(H_0), j_-(\frac{r^*}{S})\right]=\frac{i}{2\pi} \int_{\mathbb{C}} \partial_{\bar{z}}\tilde{\chi}(z)[(H_0-z)^{-1},j_-(\frac{r^*}{S})]\textrm{d}z\wedge\textrm{d}\bar{z},
\\= -\frac{i}{2\pi} \int \partial_{\bar{z}}\tilde{\chi}(z)(H_0-z)^{-1}[H_0,j_-(\frac{r^*}{S})](H_0-z)^{-1}\textrm{d}z \wedge \textrm{d}{\bar{z}}. \end{multline}
The second equation makes sense since $j_-(\frac{r^*}{S})$ is bounded and $[H_0,j_-(\frac{r^*}{S})]$ extends to a bounded operator on $\mathscr{H}$. Furthermore, the integral~\label{eq:use_hs_1} exists in the norm topology, so the compactness of the commutator follows from that of the integrand which, again, is a consequence of Corollary~\ref{critere_compact} since: \[[H_0,j_-(\frac{r^*}{S})]=-\frac{i}{S}\Gamma^1j_-'(\frac{r^*}{S}).\]
Now $\chi(H_0)j_-(\frac{r^*}{S})H_0j_-(\frac{r^*}{S})\chi(H_0)$ is equal to
\begin{equation*}\begin{gathered} \label{eq:m_ineq_2}j_-(\frac{r^*}{S})\chi(H_0)H_0\chi(H_0)j_-(\frac{r^*}{S}) \\+\underline{j_-(\frac{r^*}{S})\chi(H_0)H_0[j_-(\frac{r^*}{S}),\chi(H_0)] } +\underline{ [\chi(H_0),j_-(\frac{r^*}{S})]H_0j_-(\frac{r^*}{S})\chi(H_0)}.\end{gathered}\end{equation*}
The underlined terms form a symmetric compact operator and denoting\footnote{following the notations of~\cite{Lax:2002aa}.} $\mathbb{E}$ the operator-valued spectral measure, for any $u \in \mathscr{H}$:
\begin{equation*}\begin{aligned}  (j_-(\frac{r^*}{S})\chi(H_0)H_0\chi(H_0)j_-(\frac{r^*}{S})u,u) &=(\chi(H_0)H_0\chi(H_0)j_-(\frac{r^*}{S})u, j_-(\frac{r^*}{S})u), \\&= \int t\chi^2(t) (\mathbb{E}(\textrm{d}t)j_-(\frac{r^*}{S})u, j_-(\frac{r^*}{S})u), \\ &\geq \mu(j_-(\frac{r^*}{S})\chi(H_0)^2j_-(\frac{r^*}{S})u,u).  \end{aligned}\end{equation*}
In other words:
\begin{equation}\begin{aligned}\label{eq:m_ineq_3} j_-(\frac{r^*}{S})\chi(H_0)H_0\chi(H_0)j_-(\frac{r^*}{S}) &\geq \mu j_-(\frac{r^*}{S})\chi(H_0)^2j_-(\frac{r^*}{S}), 
\\ &\geqk \mu \chi(H_0)j^2_-(\frac{r^*}{S})\chi(H_0), \end{aligned}\end{equation}
where we have used once more the compactness of the commutator $[\chi(H_0),j_-(\frac{r^*}{S})]$.
Similar arguments prove the final point.
\end{proof}
\subsubsection{At the simple horizon}
The decomposition of the Hilbert space constructed in Section~\ref{comparison_operator} and the results discussed there concerning the properties of the eigenvalues, mean that the proof of the Mourre estimate at the simple horizon in~\cite{thdaude}, applies to our case without any essential modification. Hence we quote without proof:
\begin{prop}[{\cite[Lemma~IV.4.11]{thdaude}}]
\label{lemme:IV411daude}
Let $\lambda_0 \in \mathbb{R}$, then there are $\chi \in C^{\infty}_0(\mathbb{R})$ such that $\lambda_0 \in \supp\chi$ and $\mu \in \mathbb{R}_+^*$ such that:
\begin{equation} \chi(H_0)i[H_0,A_+(S)]\chi(H_0) \geqk \mu \chi(H_0)j^2_1(\frac{r^*}{S})\chi(H_0), \end{equation}
for large enough $S \in \mathbb{R}_+^*$.
\end{prop}
\begin{rem}
It is interesting to remark the difference in the formulation of Propositions~\ref{prop:mourre_estimate_DH}~ and~\ref{lemme:IV411daude}. Only the latter truly restricts the size of the neighbourhood on which we have a Mourre estimate, Proposition~\ref{prop:mourre_estimate_DH} on the other hand, simply forbids a Mourre estimate on a neighbourhood of $0$.
\end{rem}
Combining the two previous results leads to:
\begin{prop}
\label{prop:mourreH0}
Let $\lambda_0 \in \mathbb{R}^*$:
\begin{itemize}
\item If $\lambda_0 >0$, then one can find an interval $I \subset (0,+\infty)$ containing $\lambda_0$ and $\mu>0$ such that:
\begin{equation} \bm{1}_I(H_0)i[H_0,A_+(S) + A_-(S)]\bm{1}_I(H_0) \geqk \mu \bm{1}_I(H_0), \end{equation}
for large enough $S\in\mathbb{R}_+^*$.
\item If $\lambda_0<0$, then one can find an interval $I\subset (-\infty,0)$ containing $\lambda_0$ and $\mu>0$ such that:
\begin{equation}
\bm{1}_I(H_0)i[H_0,A_+(S) - A_-(S)]\bm{1}_I(H_0) \geqk \mu \bm{1}_I(H_0), 
\end{equation}
for large enough $S\in\mathbb{R}_+^*$.
\end{itemize}
\end{prop}

\subsection{Mourre estimate for $H$}
Now that we have at our disposition a Mourre estimate for $H_0$, we can deduce from it Mourre estimates for any operator $H$ satisfying~\eqref{eq:formeH}. Their spectral theory is closely related to that of $H_0$ as illustrated by the following lemma.
\begin{lemme}
\label{lemme:weyl}
For any $\chi \in C^{\infty}_0(\mathbb{R})$, $(H_0-i)^{-1} -(H-i)^{-1}$ and  $\chi(H_0)-\chi(H)$ are compact. In particular, $H_0$ and $H$ have the same essential spectrum.(Weyl's Theorem).
\end{lemme}
\begin{proof}
One has for any $z\in \mathbb{C}\setminus \mathbb{R}$:
 \begin{align*} (H_0-z)^{-1} -(H-z)^{-1}&= (H-z)^{-1}(H-H_0)(H_0-z)^{-1},\\&=(H-z)^{-1}((h^2-1)H_0+\tilde{V})(H_0-z)^{-1},
 \end{align*}
for some matrix $\tilde{V}$ whose coefficients are in $C_\infty(\mathbb{R})$. Compactness of $(H_0-i)^{-1} -(H-i)^{-1}$ is, once more, a consequence of Corollary~\ref{critere_compact}. That of $\chi(H_0)-\chi(H)$ follows from this since the Helffer-Sjöstrand formula\footnote{see~\cite[Proposition 7.2]{Helffer:1987aa} } leads to:
\begin{equation}
\label{eq:helffer}
\chi(H)-\chi(H_0)=\frac{i}{2\pi} \int \partial_{\bar{z}}\tilde{\chi}(z)\left((H-z)^{-1}-(H_0-z)^{-1} \right)\textrm{d}z\wedge\textrm{d}\bar{z},
\end{equation}
the integral converges in norm so compactness of the integrand implies that of the integral.
\end{proof}
An immediate consequence of Lemma~\ref{lemme:weyl} is that for any $\chi \in C^{\infty}_0(\mathbb{R})$:
\begin{equation} \chi(H)[iH,A(S)]\chi(H)\eqk \chi(H_0)[iH,A(S)]\chi(H_0). \end{equation}
Now, writing $H=H_0 + (h^2-1)H_0 + h[H_0,h] + V$, let us consider: \[\chi(H_0)[(h^2-1)H_0+h[H_0,h]+ V,A_\pm(S)]\chi(H_0),\] we will in fact find that it is compact, so that:
\begin{equation} \chi(H)[iH,A]\chi(H)\eqk \chi(H_0)[iH_0,A]\chi(H_0). \end{equation}
We recall our main tool:
\begin{coro*}{Corollary~\ref{critere_compact}, Section~\ref{analytic_framework}}.
\center{\fbox{If $f,\chi \in C_{\infty}$ then $f(r^*)\chi(H_0)$ is compact.}}
\end{coro*}

\medskip

To simplify notations we drop the dependence on $S$ of the operator $A_-$. Consider first:
\begin{equation} [(h^2-1)H_0,A_\pm]=(h^2-1)[H_0,A_\pm]-[A_\pm,h^2-1]H_0.\end{equation}
$(h^2-1)\in \bm{S}^{1,1}$ so, by Corollary~\ref{critere_compact}, $(h^2-1)\chi(H_0)$ is compact. Therefore, so is: $\chi(H_0)(h^2-1)=((h^2-1)\chi(H_0))^{*}$. Since $[H_0,A_\pm] \in B(D(H_0),\mathscr{H})$, we conclude that  $\chi(H_0)(h^2-1)[H_0,A_\pm]\chi(H_0)$ is compact. Moreover: \[[A_-,h^2-1]=-iR_-(r^*)2hh' \in \bm{S}^{\infty,1},\] so $[A_-,h^2-1]\chi(H_0)$ is also compact. 

\noindent Next we consider the term:
\[[A_+,h^2-1]=\Gamma^1(R_+(r^*,\mathfrak{D})(h^2-1) - ((R_+(r^*,\mathfrak{D})(h^2-1))^*).\]
Note that: \[\begin{aligned}R_+(r^*,\mathfrak{D})(h^2-1)&=R_+(r^*,\mathfrak{D})\bra{r^*}^{-1}\bra{r^*}(h^2-1),\\&=R_+(r^*,\mathfrak{D})\bra{r^*}^{-1}j^2_1(\frac{r^*}{S})\bra{r^*}(h^2-1).\end{aligned}\]
The last equality is a consequence of the choice of support of $j_1$ and $j_+$: recall that $j_1(t)=1$ for $t\geq -1$ and $r^*\geq -S$ when $j_+(\frac{r^*-\kappa^{-1}\ln|\mathfrak{D}|}{S})\neq 0$ so $j^2(\frac{r^*}{S})=1$ whenever the term is non-zero.
$\bra{r^*}j^2_1(\frac{r^*}{S})(h^2-1)\chi(H_0)$ is therefore compact because $j^2_1(\frac{r^*}{S})(h^2-1)\in\bm{S}^{1,\infty}$.
\newline \noindent Additionally, Lemma~\ref{lemme:r+bounded} implies that $R_+(r^*,\mathfrak{D})\bra{r^*}^{-1}$ extends to a bounded operator on $\mathscr{H}$. The compactness of $\chi(H_0)[(h^2-1)H_0,A_{\pm}]\chi(H_0)$ follows.
The term:
\[[V,A_+]=[VR_+(r^*,\mathfrak{D})\Gamma^1 - \Gamma^1R_+(r^*,\mathfrak{D})V],\]
is treated identically: \[R_+(r^*,\mathfrak{D})V=R_+(r^*,\mathfrak{D})\bra{r^*}^{-1}\bra{r^*}j_1^2(\frac{r^*}{S})V,\] and $j_1^2(\frac{r^*}{S})V \in \bm{S}^{1,\infty}$ so $R_+(r^*,\mathfrak{D})V\chi(H_0)$ is compact.
Lastly using:
\[\begin{gathered}[V,A_-]=iR_-(r^*)V' \in \bm{S}^{\infty,1},\\ h[H_0,h]=-ih(\Gamma^1\partial_{r^*}h +\sqrt{\Delta_\theta}g(r^*)\Gamma^2\partial_\theta h), \\
\begin{aligned}[h[H_0,h],A_-]=&R_-(r^*)\big[h'(\Gamma^1h' +\Gamma^2\sqrt{\Delta_\theta}g(r^*)\partial_\theta h)+h\Gamma^1\partial^2_{r^*}h \\&+h\Gamma^2\sqrt{\Delta_\theta}g(r^*)\partial_{r^*}\partial_\theta h+h\Gamma^2\sqrt{\Delta_\theta}g'(r^*)\partial_\theta h\big] \in \bm{S}^{\infty,1},\end{aligned}
\\ [h[H_0,h],A_+]=-i\Gamma^1[h\partial_{r^*}h,R_+(r^*,\mathfrak{D})]-i[h\sqrt{\Delta_\theta}g(r^*)\Gamma^2\partial_\theta h,\Gamma^1R_+(r^*,\mathfrak{D})].
\end{gathered}\]
and similar arguments as before, we conclude that the remaining terms are also compact. Therefore, we have proved the following:
\begin{prop}
Let $H$ be an operator defined by~\eqref{eq:formeH}, then the conclusion of Proposition~\ref{prop:mourreH0} is true with $H$ in place of $H_0$.
\end{prop}
\subsection{Propagation estimates and other consequences of the Mourre estimate}
\subsubsection{On the spectrum of $H_0$ and $H$}
\label{sec:conseq_mourre_spectre_H_0}
The first important consequence of the estimate above is that Theorem~\ref{thm:mourre} applies to $H$ and $H_0$, on any interval disjoint from $\{0\}$. Hence, $H$ and $H_0$ have no singular continuous spectrum and all eigenvalues, other than possibly $0$, are of finite multiplicity. In fact, $H_0$ has no eigenvalue, as the following classical ``Grönwall lemma'' argument shows. 
\begin{proof}[Proof that $H_0$ has no pure point spectrum]
We only need to seek eigenvalues for $H_0$ on each of the subspaces $\mathscr{H}_{k,n}$, which, we recall, can be identified with $[L^2(\mathbb{R})]^4$. Let $\lambda \in \mathbb{R}$ and suppose that $u\in [L^2(\mathbb{R})]^4$ satisfies: $$H^{k,n}_0u= \left(\lambda + \frac{ap}{r^2_++a^2}-\frac{ap}{r_e^2+a^2}\right) u,$$ then $u \in [H^1(\mathbb{R})]^4$ and $u$ vanishes at infinity. This is also true of the function $w: r^* \mapsto e^{-i\Gamma^1\lambda r^*}u(r^*)$. $w$ additionally satisfies:
\begin{equation*}\begin{aligned}w'(r^*)&=e^{-i\Gamma^1\lambda r^*}(-i\Gamma^1)(\lambda u(r^*) - \Gamma^1D_{r^*}u(r^*)),\\&=e^{-i\Gamma^1\lambda r^*}(-i\Gamma^1)I(r^*)e^{i\Gamma^1\lambda r^*}w(r^*),\end{aligned}\end{equation*}
where: $I(r^*)=\left(-\lambda_kg(r^*)\Gamma^2 + f(r^*)-\left(\frac{ap}{r^2_++a^2}-\frac{ap}{r_e^2+a^2}\right)\right)$. From this, we deduce:
\[||w(r^*)|| \leq \int_{r^*}^{+\infty} \left\lVert I(r^*)\right\rVert ||w(r^*)|| \textrm{d}r^*,\]
Because $|| I ||$ is integrable near $+\infty$, it follows from the integral form of Grönwall's lemma that $w=0$ and therefore $u=0$.
\end{proof}
Using the separability of the Dirac equation in Kerr-de Sitter, a modified version of this argument shows that the full Dirac operator has no eigenvalues, we refer to~\cite{Belgiorno:2009aa}. We summarise these conclusions in the following lemma:
\begin{lemme}
\label{lemme:csq_spectral_mourre}
Let $H$ be an operator defined by~\eqref{eq:formeH} then:
\begin{itemize}
\item $H$ has no singular continuous spectrum,
\item $\sigma_{\textrm{ess}}(H)=\mathbb{R}$,
\item $\sigma_{\textrm{pp}}(H)\subset \{0\}$ and if $0$ is an eigenvalue then it has infinite multiplicity.\footnote{$\sigma_{\textrm{pp}}(H)$, the \emph{pure-point} spectrum, is the set of all eigenvalues of $H$. It is not to be confused with the \textrm{discrete} spectrum, $\sigma_\textrm{disc}(H)=\mathbb{R}\setminus \sigma_{\textrm{ess}}(H)$, the set of all isolated eigenvalues with finite multiplicity.}
\end{itemize}
\end{lemme}

\subsubsection{Strict Mourre estimates}
\label{sec:strict_mourre}
Let $H\in C^1(A)$, $(H,A)$ is said to satisfy a \emph{strict} Mourre estimate on some interval $I\subset \bbR$, if it satisfies a Mourre estimate with vanishing compact error. 
This slightly stricter condition will be required shortly for the important conclusion of Theorem~\ref{thm:sigal}. Nevertheless, if $(H,A)$ satisfies a Mourre estimate on some open interval $I\subset \bbR$, then for any $\lambda \in I$ that is \emph{not} an eigenvalue of $H$, one can find a small neighbourhood $J=(-\varepsilon +\lambda, \lambda -\varepsilon)$ of $\lambda \in I$ such that it satisfies a strict Mourre estimate on $J$. To see this we give a simplified version of the argument in the proof of~\cite[Lemma~7.2.12]{Amrein:1996aa}: let, for any $n$ large enough such that $(-\frac{1}{n}+\lambda,\lambda+\frac{1}{n}) \subset I$, $E_n=\mathbb{E}((-\frac{1}{n}+\lambda,\lambda+\frac{1}{n}))$; where $\mathbb{E}$ is the spectral measure of $H$. Then: \[\slim_{n\to \infty} E_n=\mathbb{E}(\{\lambda\})=0,\]
as $\lambda$ is not an eigenvalue. It follows that for any compact operator $K$: \[\lim_{n\to \infty} E_nKE_n = 0.\]
Therefore, if $\varepsilon >0$, one can find $N$, such that for any $n\geq N$ :
\[ |(E_nKE_n|x)\leq \varepsilon ||x||^2, \]
so that for $n\geq N$:
\[ E_nKE_n \geq - \varepsilon \Rightarrow E_nKE_n \geq -\varepsilon E_n. \]
Hence, if $\bm{1}_I(H)i[H,A]\bm{1}_I(H)\geq \mu + K$, then:
\[ E_ni[H,A]E_n \geq (\mu -\varepsilon)E_n.\]
Consequently, on small enough intervals around any non-eigenvalue, one has a strict Mourre estimate for any $\nu \in (0,\mu)$.

In the case of $H$ and $H_0$, the only possible eigenvalue is $0$. All our estimates avoid this point, therefore they can all be upgraded to strict estimates on small enough intervals around any point of $\bbR^*$.

\subsubsection{Minimal velocity estimate}
One of the most powerful consequences of the hypotheses of Mourre theory, largely discussed and optimised in~\cite{Amrein:1996aa}, is that it leads to a (generalised) limiting absorption principle. In our case, thanks to Proposition~\ref{prop:technical_mourre2}, $H_0,H \in C^2(A)$, and we directly have access to an abstract propagation estimate due to Sigal-Soffer~\cite{Sigal:1988aa}:
\begin{thm}
\label{thm:sigal}
Let $(H,A)$ be a pair of self-adjoint operators on a Hilbert space $\mathscr{H}$. Suppose that $A$ is a conjugate operator for $H$ on $I\subset \mathbb{R}$ and that $H\in C^{1+\varepsilon}(A),(\varepsilon\in\mathbb{R}_+^*)$. Let $\mu \in \mathbb{R}_+^*$ be such that: 
$$\textbf{1}_I(H)i[H,A]\textbf{1}_I(H) \geq \mu \textbf{1}_{I}(H).$$
Then, for any $b,\chi \in C^{\infty}_{0}(\mathbb{R})$ such that $\supp \chi\subset I$ and $\supp\,b \subset (-\infty,\mu)$ one has:
\begin{equation}
\begin{aligned}\forall u \in \mathscr{H}, &\int_1^{+\infty} \left\lVert b\left(\frac{A}{t}\right)\chi(H)e^{-iHt}u\right\rVert^2 \frac{\textrm{d}t}{t} \leq C ||u||^2,
\\ &\slim_{t\to+\infty}b\left( \frac{A}{t}\right)\chi(H)e^{-iHt} =0.\end{aligned}
\end{equation}
\end{thm}

The importance of Theorem~\ref{thm:sigal} is more obvious when the conjugate operator can be replaced by simpler operators that help to understand the propagation of fields. In~\cite[Lemma~IV.4.13]{Daude:2010aa}, it is shown that in the case of the operators under consideration here, $A$ can be replaced with $|r^*|$, and we obtain:

\begin{prop}
\label{prop:mve}
Let $\chi \in C_{0}^{\infty}(\mathbb{R})$ such that $\supp \chi \cap \{0\} = \emptyset$, then for any $H$ defined by Equation~\eqref{eq:formeH}, there are $\varepsilon_\chi,C\in \mathbb{R}_+^*$ such that for any $\psi \in \mathscr{H}$:

\begin{equation} \int_1^{\infty} \left\lVert \bm{1}_{[0,\varepsilon_{\chi}]}\left(\frac{|r^*|}{t} \right)\chi(H)e^{-itH}\psi \right\rVert^2 \frac{dt}{t} \leq C ||\psi||^2. \end{equation}
Furthermore:
\begin{equation} 
\label{eq:slim_mve}
\slim_{t\to +\infty}\bm{1}_{[0,\varepsilon_\chi]}\left(\frac{|r^*|}{t}\right)\chi(H)e^{-itH}=0.
\end{equation}
\end{prop}
This ``minimal velocity estimate'' means that, given a certain energy interval, all fields with energy in that interval, must be outside of the ``cone'' $\{|r^*|<\varepsilon_\chi t \}$ at late times; it will be crucial to the construction of the wave operators.

\subsubsection{Maximal velocity estimate}
Independently of Mourre theory, one can show that we also have a natural ``maximal velocity estimate'', that is a consequence of the geometry:

\begin{prop}
Let $\delta \in (0,1),  b\in C^{\infty}_0(\mathbb{R})$ be such that $\supp \,  b \cap [-1-\delta,1+\delta]=\emptyset$, then there is some constant $C\in \mathbb{R}^*$ such that for any $u\in\mathscr{H}$:
\begin{equation}\label{eq:max_ve} \int_1^{+\infty} \left\lVert b(\frac{r^*}{t})e^{-itH}u\right\rVert^2 \frac{dt}{t} \leq C ||u||^2. \end{equation}
Furthermore, for any $b\in C^{\infty}(\mathbb{R})$ such that $b\equiv 0$ on $[-1-\delta,1+\delta]$ and $b=1$ for $|r|$ large, then:
\begin{equation} \slim_{t\to \infty} b(\frac{r^*}{t})e^{-itH}=0 \end{equation}
\end{prop}
The proof is identical to that of~\cite[Proposition IV.4.4]{thdaude}.

%
%
%
%

\subsubsection{What of $t\to -\infty$?} 

Up to now, we have only discussed estimates in the far future, and have said nothing of the far past. After thorough inspection, one can convince onself that all the results here hold for $-H$ (the conjugate operator should also be replaced by its opposite), but, there is a faster way to see this. The Kerr-de Sitter metric~\eqref{metric} is invariant under the simultaneous substitutions:
\begin{equation*}
\left\{\begin{array}{lcr}
t&\to&-t\\
a &\to& -a
\end{array}\right..
\end{equation*}
This is intuitively reasonable because a time reversed black-hole will rotate in the opposite way. Consequently, all the results in the section have suitable analogs at $t\to -\infty$.

\section{Intermediate wave operators}\label{intermediatewaveoperators}
\subsection{Overall strategy}
In this section our goal is to show that, despite the long-range non-spherically symmetric potentials at the double horizon, it is still possible to reduce the scattering problem to a 1-dimensional one. To this end, we introduce the following operators:
\begin{align}
H_1&=H_0+h^{-1}V_C h^{-1},\\
H_e&=H_0+ g(r^*)\vartheta(\theta),
\end{align}
with:
\begin{align}
\vartheta(\theta)&=\frac{a^2\sin\theta}{\sqrt{\Delta_\theta}}\left(\frac{l^2r_e^2-1}{r_e^2+a^2} \right)\Gamma^3p +\rho_{e}m\Gamma^0 -\frac{a\sin\theta r_e}{2\rho^2_{e}}\sqrt{\Delta_\theta}\tilde{\gamma},\\
&\rho_{e}=r_e^2+a^2\cos^2\theta, \quad \tilde{\gamma}= \left(\begin{array}{cc}\sigma_x & 0 \\ 0 &\sigma_x\end{array} \right).
\end{align}
Finally, $V_C$ and $V_S$ are defined by equations~\eqref{eq:defvs} and~\eqref{eq:defvc}, their asymptotic behaviour is described in Lemma~\ref{lemme:behaviour_potentials}. 

Both $H_1$ and $H_e$ are of the prescribed form~\eqref{eq:formeH}, hence the theory presented in Section~\ref{mourre} applies to them.
We will show that we can compare the full operator $H\equiv H^p= hH_0h+V_S+V_C$ to simplified dynamics as so:

\[ H \underset{r^* \to \pm \infty}{\longrightarrow} \left\{\begin{array}{l} H_1 \underset{r^* \to +\infty}{\longrightarrow} H_0 \\  H_1 \underset{r^* \to -\infty}{\longrightarrow}H_e\end{array}\right.\]
\subsection{First comparison}
\label{section:compar1}
The first step is to compare $H$ to $H_1$. Here, there is no distinction between the behaviour at the different horizons because: \[\begin{aligned}H-H_1&=(h^2-1)H_0 +h[H_0,h]+V_S +(h^2-1)h^{-2}V_C \\&\equiv \underset{\in \bm{S}^{2,2}}{(h^2-1)}H_1 +V_S + h[H_0,h], \end{aligned} \] and $V_S+h[H_0,h]\equiv \tilde{V}_S$ is short-range. Proposition~\ref{prop:mve} is the key to prove:
\begin{prop}
\label{prop:comparaison1exist}
The generalised wave-operators:
\begin{equation}
\begin{aligned}
\Omega^1_{\pm}&=\slim_{t\to\pm \infty} e^{itH_1}e^{-itH}P_c(H), \\
\tilde{\Omega}^1_{\pm}&=\slim_{t\to \pm \infty} e^{itH}e^{-itH_1}P_c(H_1),
\end{aligned}
\end{equation}
exist, where, for any self-adjoint operator $B$, $P_c(B)$ denotes the projection onto the absolutely continuous subspace of $B$.
\end{prop}

\begin{proof}
We show the existence of the first limit at $ t \to +\infty$ the other cases are similar. We begin by remarking that: \[\displaystyle \overline{\bigcup_{ \underset{\supp \chi \cap \{0\} = \emptyset }{\chi \in C^{\infty}_0(\mathbb{R})}} \chi(H)\mathscr{H}}=P_c(H)\mathscr{H},\] so it is sufficient to prove the existence of the limit:
\[\slim_{t\to +\infty} e^{itH_1}\chi(H)e^{-itH},\]
for every $\chi \in C^{\infty}_0(\mathbb{R}), \supp \chi \cap \{0\}=\emptyset$. Consider then such a $\chi$ and let $\varepsilon_\chi$ be defined by Proposition~\ref{prop:mve}. Choose $j_0 \in C^{\infty}_0(\mathbb{R})$ such that $\supp j_0 \subset (-\varepsilon_\chi,\varepsilon_\chi)$ and $j_0\equiv 1$ on a neighbourhood of $0$. Set $j=1-j_0$. \eqref{eq:slim_mve} implies that:
\[\slim_{t\to\infty} e^{itH_1}j_0(\frac{r^{*}}{t})e^{-itH}\chi(H)=0.\]
It remains to prove the existence of:
\[\slim_{t\to\infty} e^{itH_1}j(\frac{r^*}{t})\chi(H)e^{-itH}.\]
For this, we apply the methods of Cook and Kato\footnote{see for example~\cite{Derezinski:1997aa,Lax:2002aa}}, who remarked that the convergence, for every $u$ in a dense set of $\mathscr{H}$, of the integral: \[ \int_1^{+\infty} \frac{\textrm{d}}{\textrm{d}t} \left(e^{itH_1}j(\frac{r^*}{t})\chi(H)e^{-itH}u\right), \]
was a sufficient condition for the limit to exist. To prove the convergence of the integral, there are two model arguments that will both be illustrated on this simple example. 
To begin with, let $u\in D(H)=D(H_1)$ then $\frac{\textrm{d}}{\textrm{d}t} \left(e^{itH_1}j(\frac{r^*}{t})\chi(H)e^{-itH}\right)u$ equates to:
\begin{align*} &e^{itH_1}\left(iH_1j(\frac{r^*}{t}) -\frac{r^*}{t^2}j'(\frac{r^*}{t}) -j(\frac{r^*}{t})iH \right)\chi(H)e^{-itH}u\\
&=e^{itH_1}\left( ij(\frac{r^*}{t})(H_1-H) +\frac{1}{t}(\Gamma^1-\frac{r^*}{t})j'(\frac{r^*}{t}) \right)\chi(H)e^{-itH}. \end{align*}
The treatment of the first term, illustrates the first type of argument. Consider first: \[H_1 - H= (h^2-1)H_1 + \tilde{V}_S.\] On $\supp j$, one must have $ |r^*| \geq \varepsilon t$ for some $\varepsilon \in \mathbb{R}_+^*$, thus, $\frac{1}{|r^*|} \leq \frac{1}{\varepsilon t}$ on $\supp j$. Consequently, $j(\frac{r^*}{t})(h^2-1) = O(t^{-2})$ and $j(\frac{r^*}{t})\tilde{V}_S=O(t^{-2})$. Because $H_1\chi(H)$ is bounded, the term~: \[e^{itH_1}j(\frac{r^*}{t})(H_1-H)\chi(H)e^{-itH}u,\] is therefore integrable.

The final term, $e^{itH_1}\frac{1}{t}\left(\Gamma^1-\frac{r^*}{t}\right)j'(\frac{r^*}{t})\chi(H)e^{-itH}u$, that is not clearly integrable in the sense of Lebesgue, requires a different treatment, which will serve as illustration for the second type of argument we use. Lebesgue integrability is in fact sufficient, but not necessary; the key to Cook's argument is simply that for any $\varepsilon$ and any $t_1,t_2$ sufficiently large:
\[\left\lVert e^{it_2H_1}j(\frac{r^*}{t_2})\chi(H)e^{-it_2H}- e^{it_1H_1}j(\frac{r^*}{t_1})\chi(H)e^{-it_1H}\right\rVert < \varepsilon.\]
Moreover, by the Hahn-Banach theorem, there is $v\in \mathscr{H}, ||v|| \leq 1$ such that:
\begin{align*}||e^{it_2H_1}j(\frac{r^*}{t_2})\chi(H)e^{-it_2H}u- e^{it_1H_1}j(\frac{r^*}{t_1})\chi(H)e^{-it_1H}u||\hspace{1in}\\ =(v,e^{it_2H_1}j(\frac{r^*}{t_2})\chi(H)e^{-it_2H}u- e^{it_1H_1}j(\frac{r^*}{t_1})\chi(H)e^{-it_1H}u),\\=\int_{t_1}^{t_2} \left(v,\frac{\textrm{d}}{\textrm{d}t}\left(e^{itH_1}j(\frac{r^*}{t})\chi(H)e^{-itH}u \right)\right)\textrm{d}t. \end{align*}
So, one only needs to verify that for $t_1,t_2$ sufficiently large the integral: \[\int_{t_1}^{t_2} \left(v,\frac{\textrm{d}}{\textrm{d}t}\left(e^{itH_1}j(\frac{r^*}{t})\chi(H)e^{-itH}u \right)\right)\textrm{d}t,\] can be made arbitrarily small.
Choose now $\tilde{\chi} \in C^{\infty}_0(\mathbb{R})$ such that $\supp \tilde{\chi} \cap \{0\} = \emptyset$ and $\tilde{\chi}\chi=\chi$, $\tilde{j} \in C^{\infty}_0(\mathbb{R})$, that vanishes on a neighbourhood of zero and satisfies $\tilde{j}j'=j'$. Notice then that:
\begin{equation*} \begin{aligned}\frac{1}{t}(\Gamma^1-\frac{r^*}{t})j'(\frac{r^*}{t})\chi(H)=&\tilde{\chi}(H_1)\tilde{j}(\frac{r^*}{t})\frac{1}{t}(\Gamma^1-\frac{r^*}{t})\tilde{j}(\frac{r^*}{t})j'(\frac{r^*}{t})\chi(H) \\&+\frac{1}{t}(\Gamma^1-\frac{r^*}{t})\tilde{j}(\frac{r^*}{t})j'(\frac{r^*}{t})(\tilde{\chi}(H)-\tilde{\chi}(H_1))\chi(H) \\&+ \frac{1}{t}[(\Gamma^1-\frac{r^*}{t}),\tilde{\chi}(H_1)]j'(\frac{r^*}{t})\chi(H) \\ &+ \frac{1}{t}(\Gamma^1-\frac{r^*}{t})[j'(\frac{r^*}{t}),\tilde{\chi}(H_1)]\chi(H).\end{aligned} \end{equation*}
The last three terms are $O(t^{-2})$ so are integrable, this is not changed by multiplying to the left with $e^{itH_1}$ and to the right with $e^{-itH}$. Now, for any $v\in \mathscr{H}$, one certainly has:
\begin{align*} |(v,e^{itH_1}\frac{1}{t}\tilde{\chi}(H_1)\tilde{j}(\frac{r^*}{t})(\Gamma^1-\frac{r^*}{t})\tilde{j}(\frac{r^*}{t})j'(\frac{r^*}{t})\chi(H)e^{-itH}u)|\hspace{.6in}\\=\left|\left(\frac{1}{\sqrt{t}}\tilde{j}(\frac{r^*}{t})(\Gamma^1-\frac{r^*}{t})\tilde{j}(\frac{r^*}{t})\tilde{\chi}(H_1)e^{-itH_1}v, \frac{1}{\sqrt{t}}j'(\frac{r^*}{t})\chi(H)e^{-itH}u\right)\right|,\\ \leq K  \left\lVert\frac{1}{\sqrt{t}}\tilde{j}(\frac{r^*}{t})\tilde{\chi}(H_1)e^{-itH_1}v\right\rVert \left\lVert \frac{1}{\sqrt{t}}j'(\frac{r^*}{t})\chi(H)e^{-itH}u\right\rVert, \end{align*}
for some $K\in \mathbb{R}_+^*$. In the above we have used the fact that: \[\tilde{j}(\frac{r^*}{t})(\Gamma^1-\frac{r^*}{t})\in B(\mathscr{H}).\]
Applying the Cauchy-Schwarz inequality, we get the following estimate:
\begin{multline*}\int_{t_1}^{t_2} \left|(v,e^{itH_1}\frac{1}{t}\Gamma^1\tilde{\chi}(H_1)\tilde{j}(\frac{r^*}{t})j'(\frac{r^*}{t})\chi(H)e^{-itH}u)\right|\textrm{d}t \\  \leq K \left( \int_{t_1}^{t_2} \left\lVert\tilde{j}(\frac{r^*}{t})\tilde{\chi}(H_1)e^{-itH_1}v\right\rVert^2\frac{\textrm{d}t}{t} \right)^{\frac{1}{2}}\left(\int_{t_1}^{t_2}\left\lVert j'(\frac{r^*}{t})\chi(H)e^{-itH}u\right\rVert^2 \frac{\textrm{d}t}{t} \right)^{\frac{1}{2}}.\end{multline*}
However, it follows from Proposition~\ref{prop:mve} that there is some constant $C\in \mathbb{R}_+^*$ such that:
\begin{equation*}\begin{gathered} \left( \int_{t_1}^{t_2} \left\lVert\tilde{j}(\frac{r^*}{t})\tilde{\chi}(H_1)e^{-itH_1}v\right\rVert^2\frac{\textrm{d}t}{t} \right)^{\frac{1}{2}}\left(\int_{t_1}^{t_2}\left\lVert j'(\frac{r^*}{t})\chi(H)e^{-itH}u\right\rVert^2 \frac{\textrm{d}t}{t} \right)^{\frac{1}{2}} \hspace{1in}\\ \leq C||v||\left(\int_{t_1}^{t_2}\left\lVert j'(\frac{r^*}{t})\chi(H)e^{-itH}u\right\rVert^2 \frac{\textrm{d}t}{t} \right)^{\frac{1}{2}}, \\ \leq C\left(\int_{t_1}^{t_2}\left\lVert j'(\frac{r^*}{t})\chi(H)e^{-itH}u\right\rVert^2 \frac{\textrm{d}t}{t} \right)^{\frac{1}{2}}.\end{gathered}\end{equation*}
In the last inequality we have specialised to the case where $||v|| \leq 1$. This quantity can be made arbitrarily small, for large enough $t_1,t_2$, again by Proposition~\ref{prop:mve}. The existence of the limit then follows.
\end{proof}

\subsection{Second comparison}
\label{sec:second_comp}
Our aim now is to show that asymptotically the dynamics of $H_1$ can again be simplified. However, the comparisons we will make in this section depend on the asymptotic region we consider. 
We will separate incoming and outgoing states using cut-off functions, $c_\pm$, that are assumed to satisfy: $c_\pm \in C^{\infty}(\mathbb{R})$, $c_\pm\equiv 1$ in a neighbourhood of $\pm \infty$ and that vanish in a neighbourhood of $\mp \infty$. We then seek to show that the following limits exist:
\begin{equation}
\label{eq:second_comp}
\begin{aligned}
\Omega^2_{\pm,\mathscr{H}_{r_+} }=\slim_{t\to \pm \infty}~& e^{iH_0t}c_+(r^*)e^{-iH_1t}P_c(H_1),\\
\tilde{\Omega}^2_{\pm,\mathscr{H}_{r_+}}=\slim_{t\to \pm \infty}~& e^{iH_1t}c_+(r^*)e^{-iH_0t},\\
\Omega^2_{\pm,\mathscr{H}_{r_e}}=\slim_{t \to \pm \infty}~&e^{iH_et}c_-(r^*)e^{-iH_1t}P_c(H_1),\\
\tilde{\Omega}^2_{\pm,\mathscr{H}_{r_e}}=\slim_{t \to \pm \infty}~&e^{iH_1t}c_-(r^*)e^{-iH_et}P_c(H_e).
\end{aligned}
\end{equation}
This appears to introduce a certain arbitrariness into the construction, the following lemma shows that this is not the case:
\begin{lemme}
If the limits~\eqref{eq:second_comp} exist, then they are independent of the choice of cut-off functions $c_\pm$.
\end{lemme}
\begin{proof}
The main point is that two such functions can differ on at most a compact set, i.e. their difference is an element of $C^\infty_0(\mathbb{R})$. So let us prove that if $c\in C^{\infty}_0(\mathbb{R})$, then, for instance: 
\[\slim_{t\to + \infty} e^{iH_0t}c(r^*)e^{-iH_1t}P_c(H_1) = 0,\]
the other cases will be similar. As before, by density, we only need to prove that:
\[\slim_{t\to + \infty} e^{iH_0t}c(r^*)\chi(H_1)e^{-iH_1t}= 0,\]
for any $\chi \in C^{\infty}_0(\mathbb{R}), \supp \chi \cap \{0\} = \emptyset$.  

Let $\chi$ be as so and let $M\in \mathbb{R}_+^*$ be such that $\supp c \subset [-M,M]$. Choose $j_0 \in C^{\infty}_0(\mathbb{R})$ with support contained in $(-\varepsilon_\chi,\varepsilon_\chi)$ such that, say, $j_0(s)=1$ for any $s\in [-\frac{\varepsilon_\chi}{2},\frac{\varepsilon_\chi}{2}]$. Then, for any $t\geq 1$, $j_0(\frac{r^*}{t})=1$ for any $|r^*|\leq \frac{\varepsilon_\chi}{2} t$. Hence, for $t \geq \frac{2M}{\varepsilon_\chi} $, \[c(r^*)=c(r^*)j_0(\frac{r^*}{t}),\textrm{ for any $r^* \in \mathbb{R}.$}\]
It follows that: \[\begin{aligned} \slim_{t\to +\infty} e^{iH_0t}c(r^*)\chi(H_1)e^{-iH_1t}=& \slim_{t\to+\infty} e^{iH_0t}c(r^*)j_0(\frac{r^*}{t})\chi(H_1)e^{-iH_1t}, \end{aligned}\] which vanishes by Proposition~\ref{prop:mve}.
\end{proof}

We now argue that the limits~\eqref{eq:second_comp} exist, with emphasis on:
\begin{equation}\label{eq:comp2prove} \slim_{t \to + \infty} e^{iH_e t} c_-(r^*)e^{-iH_1t}P_c(H_1), \end{equation}
the other cases being similar. 

\begin{lemme}
$H_1 -H_e$ is short-range near the double horizon.
\end{lemme}
\begin{proof}
Note that:
\begin{equation} h^{-2}V_C= g\underbrace{\left(\frac{\Xi}{\sin \theta}\left(\frac{\rho^2}{\sigma} - \frac{\sqrt{\Delta_\theta}}{\Xi} \right)\Gamma^3p +\rho m \Gamma^0 - \frac{a\sin\theta r}{2\rho^2}\sqrt{\Delta_\theta}\tilde{\gamma}\right)}_{\Theta(r,\theta)}, \end{equation}
and $\Theta(r_e,\theta)=\vartheta(\theta)$. Thus, $\Theta(r,\theta) - \vartheta(\theta) = \underset{r\to r_e}{o}(r-r_e)=\underset{r^*\to -\infty}{o}({r^*}^{-1})$, which leads to:
\[h^{-2}V_C - g\vartheta(\theta)= \underset{r^* \to \infty}{O} (\frac{1}{{r^*}^2}).\]
\end{proof}
\begin{proof}[Proof of the existence of~\eqref{eq:comp2prove}]
As before, we only need to prove the existence of: \[ \displaystyle \slim_{t\to +\infty}e^{iH_et} c_-(r^*)\chi(H_1)e^{-iH_1t},\]
for any $\chi \in C^{\infty}_0(\mathbb{R})$ with $\supp \chi \cap \{0\} =\emptyset$. Let $\chi$ be as so, and $j_0,j$ be as in the proof of Proposition~\ref{prop:comparaison1exist}, then: \[ \slim_{t \to +\infty}e^{iH_et}c_-(r^*)j_0(\frac{r^*}{t})\chi(H_1)e^{-iH_1t}=0,\] and we must prove the existence of $\displaystyle \slim_{t \to +\infty}e^{iH_et}c_-(r^*)j(\frac{r^*}{t})\chi(H_1)e^{-iH_1t}$.
To simplify notations, set $M(t)=e^{iH_et}c_-(r^*)j(\frac{r^*}{t})\chi(H_1)e^{-iH_1t}$, its derivative, $M'(t)$, is given by:
\begin{equation*}e^{iH_et}\left( iH_ec_-(r^*)j(\frac{r^*}{t}) -\frac{r^*}{t^2}c_-(r^*)j'(\frac{r^*}{t}) -c_-(r^*)j(\frac{r^*}{t})iH_1 \right)\chi(H_1)e^{-iH_1t}.\end{equation*}
The term between parentheses is:
\begin{equation*}\begin{gathered} c_-(r^*)j(\frac{r^*}{t})i(H_e-H_1) +\Gamma^1(c_-(r^*)j(\frac{r^*}{t}))' -\frac{r^*}{t^2}c_-(r^*)j'(\frac{r^*}{t})\\ = c_-(r^*)j(\frac{r^*}{t})i(H_e-H_1) +\Gamma^1(c'_-(r^*)j(\frac{r^*}{t})) +\frac{1}{t}c_-(r^*)(\Gamma^1-\frac{r^*}{t})j'(\frac{r^*}{t}). \end{gathered}\end{equation*}
The only new term compared with the proof of Proposition~\ref{prop:comparaison1exist} is: \[\Gamma^1(c'_-(r^*)j(\frac{r^*}{t})),\] however this vanishes when $t$ is sufficiently large because $c'$ has compact support and $j$ vanishes on a neighbourhood of $0$. Moreover, since $H_e-H_1$ is short-range near the double horizon and $c_-$ vanishes on a neighbourhood of $+\infty$,  the first two terms are $O(t^{-2})$ and hence integrable. The last term is treated as at the end of the proof of Proposition~\ref{prop:comparaison1exist}.
\end{proof}
\subsection{The operator $H_e$}\label{sec:Hx}
The expression of $H_e$ suggests that we seek to understand the precise spectral theory of the operator, defined on the sphere  by:
\begin{equation} \mathfrak{D}_e=\mathfrak{D} + \vartheta(\theta).\end{equation}
In particular, we would like to show that there is a Hilbert space decomposition of $L^2(S^2)\otimes\mathbb{C}^4$ which enables us to decompose the full Hilbert space $\mathscr{H}$ into an orthogonal sum of stable subspaces, that can be used to study $H_e$.  
Since $\vartheta(\theta)$ is a bounded operator it is an immediate consequence of the Kato-Rellich perturbation theorem that $\mathfrak{D}_e$ has compact resolvent. However, we require a slightly more thorough understanding of the structure of the spectral subspaces and in particular how $\Gamma^1$ acts on them. 
\subsubsection{Dimension of spectral subspaces} 
Decompose $L^2(S^2)\otimes \mathbb{C}^4$ in the usual manner by diagonalising $D_\varphi$ with anti-periodic boundary conditions, and consider the restriction $\mathfrak{D}_e^n$ of $\mathfrak{D}_e$ to the subspace with eigenvalue $n\in \mathbb{Z}+\frac{1}{2}$. In the following $E_\lambda$ will denote the spectral subspace of $\lambda\in \mathbb{R}$ for this restricted operator.

 An element $f$ in this subspace is an eigenvector with eigenvalue $\lambda \in \mathbb{R}$ of $\mathfrak{D}_e^n$ if and only if it is a solution to the first order ordinary differential equation:
 \begin{multline}\label{eq:valprop} \sqrt{\Delta_\theta}\Gamma^2 D_\theta f -\frac{i}{2}\left(\frac{\Delta'_\theta}{2\sqrt{\Delta_\theta}}+\cot \theta\right)\Gamma^2f-ar_e\sin\theta\frac{\sqrt{\Delta_\theta}}{2\rho^2_e}\tilde{\gamma}f  \\+\left(\frac{\sqrt{\Delta_\theta}}{\sin \theta}n +\frac{a^2\sin\theta}{\sqrt{\Delta_\theta}}\frac{l^2r_e^2-1}{r_e^2+a^2}p  \right)\Gamma^3f  + \rho_em\Gamma^0f -\lambda f =0.\end{multline}
Note that since $\Gamma^1$ anti-commutes with $\Gamma^0,\Gamma^3,\Gamma^2$ and $\tilde{\gamma}$, if $f$ is a solution to~\eqref{eq:valprop} then $\Gamma^1f$ is a solution to the analogous equation for $-\lambda$, in fact,  $\Gamma^1$ is an isometry between $ E_\lambda$ and $E_{-\lambda}$. 
The study of~\eqref{eq:valprop} is slightly easier after the substitution $z=\cos \theta$, after which we obtain:
\begin{equation} \label{eq:angular}a_1(z)\Gamma^2D_z +a_2(z)\Gamma^2f+a_3(z)\Gamma^3f +a_5(z)\tilde{\gamma}f +a_0(z)\Gamma^0f-\lambda f=0,\end{equation}
where:
\begin{equation}
\label{eq:angular_coeff}
\begin{gathered}
a_0(z)=\rho_e m, \quad a_1(z)=-\sqrt{\Delta_\theta}\sqrt{1-z^2},\\
a_2(z)= -\frac{i}{2}\left(-a^2l^2\frac{z\sqrt{1-z^2}}{\sqrt{\Delta_\theta}}+\frac{z}{\sqrt{1-z^2}}\right),\\
a_3(z)= \left(\frac{\sqrt{\Delta_\theta}}{\sqrt{1-z^2}}n +\frac{a^2\sqrt{1-z^2}}{\sqrt{\Delta_\theta}}\frac{l^2r_e^2-1}{r_e^2+a^2}p  \right),\\
a_5(z)=-a\sqrt{1-z^2}\frac{\sqrt{\Delta_\theta}}{2\rho^2_e}.
\end{gathered}
\end{equation}
Save the expressions $\sqrt{1-z^2}, \frac{1}{\sqrt{1-z^2}}$, all other functions appearing in the coefficients~\eqref{eq:angular_coeff} of the equation can be extended to analytic functions on a disc centered in $0$ and with radius $1+\varepsilon$ for some $\varepsilon >0$, the reason for this is that the parameters satisfy: $|al|<2-\sqrt{3}<1$ and $r_e>|a|$. This suggests that~\eqref{eq:angular} extends naturally to a differential equation expressed on an open subset of the 1-dimension complex manifold $S$: \[S=\{ (z,w) \in \mathbb{C}^2, z\in B(0,1+\varepsilon), z^2+w^2=1\},\] where $z$ is used as local coordinate - the implicit function theorem implies that this can be done in a neighbourhood of any point in $S$ save $(1,0),(-1,0)$. The functions $z,w$ are globally defined and holomorphic on $S$ and~\eqref{eq:angular} can be rewritten:
\begin{multline} \label{eq:angular2} -\sqrt{\Delta_\theta}w\Gamma^2 D_z f -\frac{i}{2}\left(-a^2l^2\frac{zw}{\sqrt{\Delta_\theta}}+\frac{z}{w}\right)\Gamma^2f \\+\left(\frac{\sqrt{\Delta_\theta}}{w}n +\frac{a^2w}{\sqrt{\Delta_\theta}}\frac{l^2r_e^2-1}{r_e^2+a^2}p  \right)\Gamma^3f  -aw\frac{\sqrt{\Delta_\theta}}{2\rho^2_e}\tilde{\gamma}f + \rho_em\Gamma^0f -\lambda f =0. \end{multline}
By the Cauchy-Lipschitz theorem the set of solutions to Equation~\eqref{eq:angular2} on $S\setminus \{ (1,0),(-1,0)\}$ is a 4-dimensional vector space. The solutions to~\eqref{eq:angular} will be the restrictions to $]-1,1[$, (i.e. $z\in]-1,1[, w>0$) of those of~\eqref{eq:angular2}. Amongst these, we must pick out those in $L^2]-1,1[$. Since $\mathfrak{D}_e$ has compact resolvent we already know that they exist only for a countable number of values of $\lambda$. We will not seek the exact condition for this, but, a simple analysis of the behaviour of the solutions near a point where $w=0$ will enable us to see that the subspace of $L^2]-1,1[$ solutions is at most of dimension $2$. To this end, we switch to local coordinates defined around such a point, say, $(-1,0)$. In fact, again using the Implicit Function Theorem, one can choose $w$ as local coordinate on a neighbourhood of $(-1,0)$, the equation then becomes:
 \begin{multline} \label{eq:angular3} \sqrt{\Delta_\theta}z\Gamma^2 D_w f -\frac{i}{2}\left(-a^2l^2\frac{zw}{\sqrt{\Delta_\theta}}+\frac{z}{w}\right)\Gamma^2f\\ +\left(\frac{\sqrt{\Delta_\theta}}{w}n +\frac{a^2w}{\sqrt{\Delta_\theta}}\frac{l^2r_e^2-1}{r_e^2+a^2}p  \right)\Gamma^3f  -aw\frac{\sqrt{\Delta_\theta}}{2\rho^2_e}\tilde{\gamma}f + \rho_em\Gamma^0f -\lambda f =0. \end{multline}
 \eqref{eq:angular3} has a singular-regular point at $w=0$\footnote{see~\cite{Ince:1956aa}}, hence, one can apply the Frobenius method, i.e. there are solutions of the form $f(w)= w^{\alpha}  \sum_k a_k w^k$. Plugging this anstaz into~\eqref{eq:angular3} we find that $a_0$ must be in the null space of the map:
 \begin{equation} M(\alpha)= i(\alpha +\frac{1}{2})\Gamma^2 + n\Gamma^3. \end{equation}
 The kernel is non-trivial only if $\alpha$ satisfies:
 \begin{equation} \label{eq:indicial} (\alpha-n+\frac{1}{2})^2(\alpha+n+\frac{1}{2})^2=0. \end{equation}
  For each solution to~\eqref{eq:indicial}, the kernel of $M(\alpha)$ is of dimension 2, 
%
 and so one can generate two linearly independent solutions for each $\alpha$\footnote{Note that, since the roots of~\eqref{eq:indicial} differ by a positive integer, the anstaz will need to be modified to include possible logarithmic terms in the solution when $\alpha=-|n|-\frac{1}{2}$}. Only $\alpha=|n| -\frac{1}{2}$ can yield square integrable solutions to~\eqref{eq:angular}, thus it follows that:
 \begin{lemme} \label{lemme:dimensionvp} In the notations of this paragraph, if $n\in \mathbb{Z}+\frac{1}{2}$ and $\lambda\in \sigma(\mathfrak{D}_e^n)$, then $\dim E_{\lambda} \leq 2$.
 \end{lemme}

We now complete the proof of Lemma~\ref{lemme:spectreD} ; the eigenequation $\tilde{S}\psi_{k,n}=\lambda_k\psi_{k,n}$ is the special case of~\eqref{eq:angular}, where $r_e=p=m=0$. In this case, the equation has another symmetry that amounts to saying that $\Gamma^2$ and $\Gamma^3$ anti-commute with the matrix $P=\left(\begin{array}{cc}0 & I_2 \\ I_2 & 0 \end{array}\right)$. Hence, $P$, like $\Gamma^1$, is an isometry of $E_\lambda$ onto $E_{-\lambda}$, however for any $u \in \mathbb{C}^4\setminus\{0\}$, $Pu$ and $\Gamma^1u$ are linearly independent, so that we must have equality in Lemma~\ref{lemme:dimensionvp}. The form of the solutions follows from the block diagonal form of the equations.

\subsubsection{A reduction of $H_e$}
Denote now: \begin{itemize} \item $\sigma(\mathfrak{D}_e)\cup\{0\}=(\mu_k)_{k\in \mathbb{Z}}$, enumerated such that $\mu_{-k}=-\mu_{k}$, for each $k\in\mathbb{Z}$.  \item For each $k\in \mathbb{Z}$, $J(k)$ the set of integers $q\in \mathbb{Z}$ such that $\mu_k$ is an eigenvalue for $\mathfrak{D}_e^{q+\frac{1}{2}}$; note that also $J(k)=J(-k)$. \item If $k\in \mathbb{Z},q\in J(k)$, $E_{k,q}$ the spectral subspace of the eigenvalue $\mu_k$ of $\mathfrak{D}_e^{q+\frac{1}{2}}$. By convention, if $0\not\in \sigma(\mathfrak{D}_e)$, we set $J(0)=\{0\}$ and $E_{0,0}=\{0\}$. \item For each $k\in \mathbb{N}^*$ and fixed $q\in J(k)$, $\tilde{E}_{k,q}= L^2(\mathbb{R})\otimes(E_{k,q} \overset{\perp}{\oplus} E_{-k,q})$.
\item $\tilde{E}_{0,q}=L^2(\mathbb{R})\otimes E_{0,q}, q \in J(0)$ \end{itemize}
The subspaces $\tilde{E}_{k,q}$ are, by construction, stable under the action of $H_e$ and: \[\mathscr{H}={\bigoplus_{k \in \mathbb{N}, q\in J(k)} \tilde{E}_{k,q}}.\] 

Now, let $k \in \mathbb{N}^*, q \in J(k)$, if $(e_i)_{i\in \llbracket1,\dim{E_{k,q}}\rrbracket}$ is an orthonormal basis for $E_{k,q}$, then $(\Gamma^1e_i)_{i\in \llbracket1,\dim{E_{k,q}}\rrbracket}$ is an orthonormal basis of $E_{-k,q}$ and so, since $E_{k,q}$ and $E_{-k,q}$ are orthogonal, one can concatenate these two bases to obtain an orthonormal basis $E_{k,q}\oplus E_{-k,q}$. This enables us to identify, isometrically, $\tilde{E}_{k,q}$ with $[L^2(\mathbb{R})]^{2\dim E_{k,q}}$ via the natural isomorphism:
\[((u_i)_{i\in \llbracket1,\dim{E_{k,q}}\rrbracket},(v_i)_{i\in \llbracket1,\dim{E_{k,q}}\rrbracket}) \longmapsto \sum_{i=1}^{\dim{E_{k,q}}} (u_i +v_i\Gamma^1)e_i.\]
Through this isomorphism, the restriction, $H_e^{q,n}$ of $H_e$ to $\tilde{E}_{k,q}$ corresponds to the following operator:
\[\Gamma D_r^* + \mu_k g(r^*) \tilde{\Gamma}  + f(r^*),\]
where $\Gamma=\left(\begin{array}{cc} 0 & I_{\dim E_{k,q}} \\ I_{\dim E_{k,q}} & 0 \end{array} \right), \tilde{\Gamma}=\left(\begin{array}{cc} I_{\dim E_{k,q}} &0  \\ 0 & -I_{\dim E_{k,q}}  \end{array} \right)$ and satisfy the important property that $\{\Gamma,\tilde{\Gamma}\}=0$. It is easily seen to be unitarily equivalent to:
\begin{equation} \begin{aligned} \Gamma^1D_r^* -\mu_kg(r^*)\Gamma^2 + f(r^*) &\quad \textrm{if $\dim E_{k,q}=2$}, \\
-\sigma_z D_r^* +\mu_k g(r^*)\sigma_x +f(r^*) &\quad  \textrm{if $\dim E_{k,q}=1$}.\end{aligned} \end{equation}
If $0\in \sigma(\mathfrak{D}_e)$ then, $\dim E_{0,q}\in\{1,2\}$, for any $q\in J(0)$ and through the natural identification described above is of the form $\Gamma D_{r^*}+f(r^*)$ where $\Gamma$ here is just some unitary matrix.
This is in all points analogous to~\eqref{eq:decompH0}, and we will now be able to complete the scattering theory in a unified fashion. It also follows that $H_e$ has no eigenvalues by the same Grönwall lemma argument that was used for $H_0$ in Section~\ref{sec:conseq_mourre_spectre_H_0}. In short we have:
\begin{lemme}
$\sigma(H_e)=\sigma_{ac}(H_e)$, consequently, $P_c(H_e)=\textrm{Id}$.
\end{lemme}

\subsection{The spherically symmetric operators}
\label{sec:sphericalsymmop}
The final step required in order to obtain the full scattering theory is to compare $H_e$ and $H_0$ to their natural asymptotic profiles, $\Gamma^1D_{r^*} + c_\pm$ at $r^* \to\pm \infty$ respectively;  $c_+ = \frac{ap}{r_+^2+a^2}-\frac{ap}{r_e^2+a^2}$ and $c_-=0$.

In the previous paragraph, we established that the Hilbert space $\mathscr{H}$ could be decomposed into an orthogonal sum of stable subspaces on which $H_e$ reduces to a spherically symmetric operator ; this was also shown to be the case of $H_0$, in Section~\ref{comparison_operator}. Consequently, in order to construct wave operators, we only need to work on one of these subspaces. Additionally, the similarities between the reduced forms of $H_e$ and $H_0$ imply that we, in fact, only need to know how to construct the wave operators for\footnote{We choose to discuss the case where $\dim E_{k,q}=2$, but the reasoning is independent of this choice.}:
\begin{equation} \label{eq:op_sym_spherique} \mathfrak{h}=\Gamma^1 D_{r^*} -\mu g(r^*)\Gamma^2 + f(r^*), \end{equation} on $[L^2(\mathbb{R})]^4,$ and under the assumption that we have minimal/maximal velocity estimates. This is manifestly the case for our operators because the estimates are stable under restriction to a stable subspace.

The important point is that the operator $\mathfrak{h}$ in~\eqref{eq:op_sym_spherique} is formally similar to the restriction to a spherical harmonic of the (charged) Dirac operator of the Reissner-Nordström black hole given in~\cite[Equation~3.6]{Daude:2010aa}. The extreme black hole horizon ($r^*\to \infty$) can be assimilated with spacelike infinity and the symbols $f$, $g$ have the same asymptotic behaviour at both infinities as the corresponding ones in~\cite[Equation~3.6]{Daude:2010aa}. It follows that we can apply the results of~\cite{Daude:2010aa} to our case. We note that, in fact, our operator is simpler than the one studied in~\cite{thdaude,Daude:2010aa} since there are no surviving mass terms.

\medskip

Precisely, using~\cite[Propositions~5.6 and 5.7]{Daude:2010aa} we find that~:
\begin{prop}
\label{prop:microlocal}
Let $\chi \in C^{\infty}_0(\mathbb{R})$ be such that $\supp \chi \cap \{0\}=\emptyset$ and choose $0<\theta_1<\theta_2$, then there is a constant $C>0$ such that for any $u\in [L^2(\mathbb{R})]^4$:
\begin{equation}\label{eq:ml_ve} \int_1^{+\infty} \left\lVert \bm{1}_{[\theta_1,\theta_2]}(\frac{|r^*|}{t})(\Gamma^1-\frac{r^*}{t})\chi(\mathfrak{h})e^{-it\mathfrak{h}}u \right\rVert^2 \frac{dt}{t}\leq C ||u||^2.\end{equation}
Furthermore~:
\begin{equation}
\slim_{t\to +\infty} \bm{1}_{[\theta_1,\theta_2]}(\frac{|r^*|}{t})(\Gamma^1-\frac{r^*}{t})\chi(\mathfrak{h})e^{-it\mathfrak{h}}=0. 
\end{equation} 
Analogous results can be established at $t \to-\infty$, but one must replace $\Gamma^1$ with $-\Gamma^1$.
\end{prop}
\begin{rem}
\label{rem:replace_appendixD}
For the specific treatment of our operators, due to the lack of mass terms, it is possible to simplify the proofs in~\cite{Daude:2010aa}, avoiding in particular the use of pseudo-differential operators. Indeed, the classical velocity operator is, in our case, simply $\Gamma^1$. The main argument of the proof, i.e. estimating the Heisenberg derivative of a well-chosen propagation observable $\phi$ is identical. The only difficulty encountered is a troublesome long-range term that appears due to the matrix-valued coefficients of the operator: the matrices $\Gamma^1$ and $\Gamma^2$ do not commute. Nevertheless, rather large spectral subspaces of $\mathfrak{h}_0$ sit in one of the spectral spaces of $\Gamma^1$, and, restricted to these subspaces, the commutator is zero. Using the notion of locally scalar operators introduced in~\cite{Mantoui:2001aa}, one can exploit this as in~\cite{Daude:2010aa} to prove that the term is no obstruction.
Furthermore, one can adapt the same arguments to show directly the following lemma:
\begin{lemm} The following limits exist:
\label{lemme:existence_strong_cinf_lim_gamma}
\[\slim_{t\to\pm \infty} e^{it\mathfrak{h}}\Gamma^1e^{-it\mathfrak{h}}.\]
\end{lemm}
\end{rem}
Proposition~\ref{prop:microlocal} is known as a microlocal velocity estimate. It completes the asymptotic information about the operator $\frac{r^*}{t}$ - which is itself to be thought of as an approximate velocity operator - provided by minimal and maximal velocity estimates. For instance, combining the three, we show that~:
\begin{corollaire}
\label{lemme:asymptotic_velocity}
For any $J\in C_{\infty}(\mathbb{R})$:
\begin{equation}
\slim_{t\to+\infty} e^{it\mathfrak{h}}(J(\frac{r^*}{t})-J(\Gamma^1))e^{-it\mathfrak{h}}=0,
\end{equation}
\end{corollaire}
\begin{proof}
First, by density, it is sufficient to consider $J\in C^{\infty}_0(\mathbb{R})$. For such $J$,  the Helffer-Sjöstrand formula can be used to show that the following holds for any $j_0 \in C^{\infty}_0(\mathbb{R})$:
\[\begin{aligned} (J(\frac{r^*}{t})-&J(\Gamma^1))j_0(\frac{r^*}{t})\\&=\frac{i}{2\pi}\int \partial_{\bar{z}}\tilde{J}(z)(\Gamma^1-z)^{-1}\left(\frac{r^*}{t}-z\right)^{-1}(\Gamma^1-\frac{r^*}{t})j_0(\frac{r^*}{t})\textrm{d}z\wedge\textrm{d}\bar{z}\\&=B(t)(\Gamma^1-\frac{r^*}{t})j_0(\frac{r^*}{t}). \end{aligned}\]
The $B(t)$ are uniformly bounded in $t$. 
By a further density argument we only need to prove that for any $\chi \in C^{\infty}_0(\mathbb{R})$, $0 \not\in \supp \chi$:
\[\slim_{t\to+\infty} e^{it\mathfrak{h}}(J(\frac{r^*}{t})-J(\Gamma^1))\chi(\mathfrak{h})e^{-it\mathfrak{h}}=0.\]
Fix $\chi$ and introduce a smooth partition of unity, $j_1,j_2,j_3$ subordinate to the open cover:
\[U_1=\{|x| < \varepsilon_\chi - \frac{\delta}{2}\}, U_2= \{ |x|>1+\frac{\delta}{2}\}, U_3=\{  \varepsilon_\chi-\delta<|x| < 1+\delta \},\] where $\varepsilon_\chi$ is given by Proposition~\ref{prop:mve} and $\delta \in (0,2\varepsilon_\chi)$.
Then:
\[e^{it\mathfrak{h}}(J(\frac{r^*}{t})-J(\Gamma^1))\chi(\mathfrak{h})e^{-it\mathfrak{h}}=\sum_i e^{it\mathfrak{h}}B(t)(\Gamma^1-\frac{r^*}{t})j_i(\frac{r^*}{t})\chi(\mathfrak{h})e^{-it\mathfrak{h}}\]
The result now follows from the minimal, maximal and microlocal velocity estimates. 
\end{proof}
\subsubsection{Asymptotic velocity operators and wave operators for the spherically symmetric operators}
\label{sec:asymptotic_velocity_operators}
Corollary~\ref{lemme:asymptotic_velocity} can now be used to show the existence of asymptotic velocity operators which are defined as the limits\footnote{See the appendices of~\cite{Derezinski:1997aa}}~:
\[P^\pm=\underset{t\to \pm\infty}{\textrm{s\,--\,$C_\infty$\,--\,lim~}} e^{it\mathfrak{h}}\frac{r^*}{t}e^{-it\mathfrak{h}}.\]
According to~(Lemma~\ref{lemme:existence_strong_cinf_lim_gamma}), we know that the limits~: $\displaystyle \slim_{t\to \pm\infty} e^{it\mathfrak{h}}(\pm\Gamma^1)e^{-it\mathfrak{h}}$ exist and, consequently:
\begin{equation} \begin{gathered} P^\pm = \slim_{t\to\pm \infty} e^{it\mathfrak{h}}(\pm\Gamma^1)e^{-it\mathfrak{h}}, \\ \sigma(P^\pm)=\{-1,1\}.
\end{gathered}\end{equation}

The final stage of the construction is to prove the existence of the (modified) operators in the spherically symmetric case. Here, the operators $P^\pm$ can be used to distinguish between the incoming and outgoing states instead of cut-off functions. The simplicity of their spectrum means in particular that:
\[\mathscr{H} = \mathscr{H}_{\textrm{in}} \oplus \mathscr{H}_{\textrm{out}},\]
where: $\mathscr{H}_{\textrm{in}}=\bm{1}_{\{-1\}}(P^{\pm})$, $\mathscr{H}_{\textrm{out}}=\bm{1}_{\{1\}}(P^{\pm})$. 

At the simple horizon, the asymptotic dynamics is given by $\mathfrak{h}_1=\Gamma^1D_{r^*} + c_0$. The difference between this and the operator $\mathfrak{h}$ is short range when $r^* \to +\infty$. Hence, the existence of the wave operators on $\mathscr{H}_{\textrm{out}}$ can be shown in exactly the same manner as that of~\eqref{eq:comp2prove}. 

At the double horizon, it is necessary to modify slightly the comparison dynamics in order to take into account the long range potentials, as in~\cite{Daude:2010aa}, we choose to use the Dollard~\cite{Dollard:1966aa} modification; in particular, the existence of the modified wave operator is contained in the results presented in~\cite[Sections~VII.B~(Theorem~7.2),~VII.C]{Daude:2010aa}.

We briefly recall the main idea of the Dollard modification. We seek to compare $\mathfrak{h}= \Gamma^1D_{r^*} -\mu g(r^*)\Gamma^2 +f(r^*)$ to $\mathfrak{h}_0=\Gamma^1D_{r^*}$ on $\mathscr{H}_{\textrm{in}}$. Several remarks are in order: both the potentials are long-range near the double horizon and $\{\Gamma^2,\mathfrak{h}_0\}=0$. This anti-commutation property means that the corresponding term can be thought of as an ``artifical'' long-range term; it is no obstruction to the existence of wave operators. This is perhaps best understood by looking at $\mathfrak{h}^2$~:
\begin{multline*} \mathfrak{h}^2=D_{r^*}^2+\mu^2g(r^*)^2+f(r^*)^2+\Gamma_1\{D_{r^*},f(r^*)\}-2\mu f(r^*)g(r^*)\Gamma^2\\-\mu\underbrace{\{\Gamma^1,\Gamma^2\}}_{=0}g(r^*)D_{r^*}+i\mu g'(r^*)\Gamma^1\Gamma^2.\end{multline*}
We observe that there are no surviving long-range times containing $g$.

The main idea of the Dollard modification can be explained as follows: if the potential $f(r^*)$ commuted with $\mathfrak{h}_0$, one could expect on a purely formal level that:
\[\begin{aligned} e^{i\mathfrak{h}t}e^{-if(r^*)t}e^{-i\mathfrak{h}_0t}&=e^{i(\mathfrak{h}_0-\mu g(r^*)\Gamma^2)t}e^{+if(r^*)t}e^{-if(r^*)t}e^{-i\mathfrak{h}_0t}\\&=e^{i(\mathfrak{h}_0-\mu g(r^*)\Gamma^2)t}e^{-i\mathfrak{h}_0t}.\end{aligned}\]
Hence, modifying the asymptotic dynamics with $e^{itf(r^*)}$ would enable us to construct a wave operator.
Now, of course $f$ does not commute with $\mathfrak{h}_0$, but, Proposition~\ref{prop:microlocal} and Corollary~\ref{lemme:asymptotic_velocity} suggest that, in some sense, $r^* \approx \Gamma^1t$ when $t\to +\infty$, therefore it could be a good idea to attempt to approximate $f(r^*)$ with $f(\Gamma^1t)$, which \emph{does} commute with $\mathfrak{h}_0$! We are therefore lead to try the above reasoning with the dynamics $U(t,t_0)$ generated by $f(t\Gamma^1)$. In fact, the comparison only interests us for $r^*<0$, so we will consider the dynamics generated by $\tilde{f}(t\Gamma^1)=j(t\Gamma^1)f(t\Gamma^1)$ where $j \in C^{\infty}(\mathbb{R})$ is a smooth cut-off function satisfying $j(s)=0$ if $s>1$ and $j(s)=1$ if $s<\frac{1}{2}$.
Since $t\mapsto \tilde{f}(t\Gamma^1)=V(t)$ is uniformly bounded in $t$, $U(t,t_0)$ of this time-dependent operator is given by the Dyson series, or, time-ordered exponential:
\[\begin{aligned} U(t,t_0)&= \sum_{n=0}^{+\infty} \frac{(-i)^n}{n!} \int_{[t_0,t]^n} T(V(t_1)V(t_2)\dots V(t_n))\textrm{d}t_n\dots\textrm{d}t_1\\&=T\exp\left((-i)\int_{t_0}^t V(s)\textrm{d}s\right).\end{aligned}\]
In the above, the operator $T$ denotes time ordering of the operators which is defined as:
\[T(V(t_1)\dots V(t_n))=\sum_{\sigma \in \mathfrak{S}_n}\bm{1}(t_{\sigma(1)}>t_{\sigma(2)}>\dots>t_{\sigma(n)})V(t_{\sigma(1)})\dots V(t_{\sigma(n)}).\]
The uniform-boundedness of the operators $V(t)$ implies that this expansion converges in norm. 
Set $U(t)=U(t,0)$, then according to~\cite[Section~7.2]{Daude:2010aa}~:
\begin{prop}
\label{prop:wave_operators}
The following limits exist:
\begin{equation}
\begin{aligned}
&\slim_{t\to\pm\infty} e^{it\mathfrak{h}}e^{-it\mathfrak{h}_1}\bm{1}_{\{1\}}(\pm \Gamma^1),\\
&\slim_{t\to\pm\infty} e^{it\mathfrak{h}_1}e^{-it\mathfrak{h}}\bm{1}_{\{1\}}(P^{\pm}),\\
&\slim_{t\to\pm\infty} e^{it\mathfrak{h}}U(t)e^{-it\mathfrak{h}_0}\bm{1}_{\{-1\}}(\pm\Gamma^1),\\
&\slim_{t\to\pm\infty} e^{it\mathfrak{h}_0}U(t)^*e^{it\mathfrak{h}}\bm{1}_{\{-1\}}(P^{\pm}).
\end{aligned}
\end{equation}
\end{prop}
Once more, we note that the proof in~\cite[Section~7.2]{Daude:2010aa} is complicated by the presence of a mass term absent from our operators. The reader will find 
a simplified proof in Appendix~\ref{app:dollard}.

\section{The full scattering theory}\label{fullscattering}
In the previous two sections, the original scattering problem was progressively reduced to a one-dimensional problem via two intermediate comparisons. We discussed the proof of the existence of a number of strong limits that are to be identified with intermediate waves operators. In this section, we assemble these results into the scattering theory we set out to construct; the whole construction was broken up into three comparisons as illustrated in Figure~\ref{fig:tree1}.
\begin{figure}[h]
\begin{forest}
forked edges,
for tree={
grow'=0,
draw,
align=c,
rounded corners,
}
[$H$
    [$H_1$
    	[$H_e$
	[Asymptotic profiles]
	]
    	[$H_0$
	[Asymptotic profiles]
	]
    ]
] 
\end{forest}
\caption{Successive comparisons \label{fig:tree1}}
\end{figure}
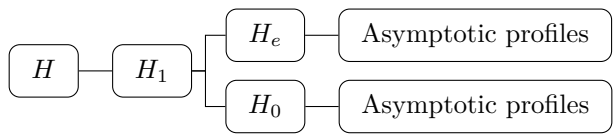
\subsection{Comparison I}
The difference between $H_1$ and $H$ being a short-range potential at both infinities, there was no obstruction to the existence of the classical wave operators (Proposition~\ref{prop:comparaison1exist}):
\begin{equation}
\begin{aligned}
\Omega^1_{\pm}&=\slim_{t\to \pm \infty} e^{itH_1}e^{-itH}P_c(H), \\
\tilde{\Omega}^1_{\pm}&=\slim_{t\to\pm\infty}e^{itH}e^{-itH_1}P_c(H_1).
\end{aligned}
\end{equation} 
The properties of these operators are well known\footnote{see, for example, \cite[Chapter~37]{Lax:2002aa}}, they satisfy:
\begin{equation}
\label{eq:propertieswo1}
\begin{gathered}
\tilde{\Omega}_{\pm}^1={\Omega_{\pm}^1}^*, \, \, \underset{\textrm{Intertwining relation}}{\Omega^1_\pm H = H_1 \Omega^1_\pm}, \\{\Omega_{\pm}^1}^*{\Omega_{\pm}^1}=P_c(H), \,\, {\Omega_{\pm}^1}{\Omega_{\pm}^1}^*=P_c(H_1),
\end{gathered}
\end{equation}
as such they are isometries between the absolutely continuous subspaces of $H$ and $H_1$; the intertwining relation shows that $H$ and $H_1$ are unitarily equivalent.
\subsection{Comparison II}
The second comparison was established in Section~\ref{sec:second_comp} and required to distinguish between states scattering to the double horizon $\mathscr{H}_{r_e}$ and those scattering to the simple horizon $\mathscr{H}_{r_+}$. This distinction was accomplished using smooth cut-off functions $c_\pm$, vanishing on a neighbourhood of $\mp \infty$ and equal to $1$ on a neighbourhood of $\pm \infty$; we will denote by $\mathscr{C}_\pm$ the subset of smooth functions with these properties. We have shown the existence of the limits, for $c_\pm \in \mathscr{C}_\pm$:
\begin{equation}
\label{eq:waveop2}
\begin{aligned}
\Omega^2_{\pm,\mathscr{H}_{r_+} }=\slim_{t\to \pm \infty}~& e^{iH_0t}c_+(r^*)e^{-iH_1t}P_c(H_1),\\
\tilde{\Omega}^2_{\pm,\mathscr{H}_{r_+}}=\slim_{t\to \pm \infty}~& e^{iH_1t}c_+(r^*)e^{-iH_0t},\\
\Omega^2_{\pm,\mathscr{H}_{r_e}}=\slim_{t \to \pm \infty}~& e^{iH_et}c_-(r^*)e^{-iH_1t}P_c(H_1),\\
\tilde{\Omega}^2_{\pm,\mathscr{H}_{r_e}}=\slim_{t \to \pm \infty}~& e^{iH_1t}c_-(r^*)e^{-iH_et}.
\end{aligned}
\end{equation}
The limits are independent of the choice of $c_\pm$; recall also that both $H_e$ and $H_0$ only have absolutely continuous spectrum. \cite[Proposition~4]{Reed:1979aa} shows that the ranges of both $\tilde{\Omega}^2_{\pm,\mathscr{H}_{r_e}}$ and $\tilde{\Omega}^2_{\pm,\mathscr{H}_{r_+}}$ are subsets of the absolutely continuous subspace of $H_1$, it follows then that:
\begin{equation} 
\tilde{\Omega}^2_{\pm, \mathscr{H}_{r_e}}={\Omega^2}^*_{\pm,\mathscr{H}_{r_e}}, \quad \tilde{\Omega}^2_{\pm, \mathscr{H}_{r_+}}={\Omega^2}^*_{\pm,\mathscr{H}_{r_+}}.
\end{equation}
One also has the intertwining relations on the absolutely continuous subspace of $H_1$: 
\begin{align}\label{eq:intertwineH0H1} H_0\Omega^2_{\pm,\mathscr{H}_{r_+}}=\Omega^2_{\pm,\mathscr{H}_{r_+}}H_1,\\
 \label{eq:intertwineHxH1} H_e\Omega^2_{\pm,\mathscr{H}_{r_e}}=\Omega^2_{\pm,\mathscr{H}_{r_e}}H_1. \end{align}
Together, Equations~\eqref{eq:propertieswo1},~\eqref{eq:intertwineH0H1} and~\eqref{eq:intertwineHxH1} give:
\begin{equation}H_0\Omega^2_{\pm,\mathscr{H}_{r_+}}\Omega^1_\pm=\Omega^2_{\pm,\mathscr{H}_{r_+}}\Omega^1_\pm H, \quad H_e\Omega^2_{\pm,\mathscr{H}_{r_e}}\Omega^1_\pm=\Omega^2_{\pm,\mathscr{H}_{r_e}}\Omega^1_\pm H. \end{equation}
Now, since the limits are independent of the choice of $c_\pm\in \mathscr{C}_\pm$, one can always choose $c_\pm$ such that $c_+^2 + c_-^2 = 1$, consequently:
\begin{equation} 
\label{eq:propertieswo2}\begin{gathered} {\Omega^2}^*_{\pm,\mathscr{H}_{r_e}}\Omega^2_{\pm,\mathscr{H}_{r_e}} + {\Omega^2}^*_{\pm,\mathscr{H}_{r_+}}\Omega^2_{\pm,\mathscr{H}_{r_+}}= P_c(H_1),\\
\end{gathered}
 \end{equation}
One could have also chosen $c_\pm$ such that their supports were disjoint, therefore, we must also have:
\begin{equation}
\Omega^2_{\pm,\mathscr{H}_{r_e}}{\Omega^2}^*_{\pm,\mathscr{H}_{r_+}}=\Omega^2_{\pm,\mathscr{H}_{r_+}}{\Omega^2}^*_{\pm,\mathscr{H}_{r_e}}=0.
\end{equation}
In other words, relation~\eqref{eq:propertieswo2} is an orthogonal sum decomposition of the absolutely continuous subspace of $H_1$ and the operators~\eqref{eq:waveop2} are partial isometries.
We therefore have a decomposition of $P_c(H_1)$ into incoming and outgoing states. In what follows, to simplify notations, we consider only the direct wave operators, analogous statements can be formulated for the reverse ones. Define:
\[X_{\textrm{in}}^{H_1}=(\ker \Omega^2_{+,\mathscr{H}_{r_e}})^{\perp}, \quad X_{\textrm{out}}^{H_1}=(\ker \Omega^2_{+,\mathscr{H}_{r_+}})^{\perp}.\]
In virtue of Equation~\eqref{eq:propertieswo2}, these subspaces have nice characterisations, indeed: $X_{\textrm{in}}^{H_1}$ is exactly $\ker \Omega^2_{+,\mathscr{H}_{r_+}}\cap P_c(H_1)\mathscr{H}$ and $\phi \in \ker \Omega^2_{+,\mathscr{H}_{r_+}}\cap P_c(H_1)\mathscr{H}$, if and only if :
\[\lim_{t\to +\infty} ||c_+(r^*)e^{-itH_1}\phi||=0,\] for any $c_+ \in \mathscr{C}_+$. In other words, the states in $X^{H_1}_{\textrm{in}}$ are exactly those whose energy is concentrated on $\mathbb{R}_-$ at late times. Similarly, $\phi \in X_{\textrm{out}}^{H_1}$ if and only if:
\[\lim_{t\to +\infty} ||c_-(r^*)e^{-itH_1}\phi||=0,\] for any $c_-\in \mathscr{C}_-$.
An important point is that $\Omega^2_{+,\mathscr{H}_{r_e}}$ maps $X_{\textrm{in}}^{H_1}$ onto a similar subspace for $H_e$ (and similarly at $\mathscr{H}_{r_+}$ for $H_0$). If $\psi$ is in the range of $\Omega^2_{+,\mathscr{H}_{r_e}}$, then there is $\phi \in X_{\textrm{in}}^{H_1}$ such that:
\[\lim_{t\to +\infty} ||e^{-itH_e}\psi - c_-(r^*)e^{-itH_1}\phi||=0,\]
for any $c_- \in \mathscr{C}_-$. Hence for any $c_+ \in \mathscr{C}_+$, one can choose $c_- \in \mathscr{C}_-$ with support disjoint from that of $c_+$ so that:
\[\begin{aligned}0&=\lim_{t\to +\infty} ||c_+(r^*)e^{-itH_e}\psi - c_+(r^*)c_-(r^*)e^{-itH_1}\phi||,\\&=\lim_{t\to +\infty} ||c_+(r^*)e^{-itH_e}\psi\rVert.\end{aligned}\]
Conversely, all such states are mapped into $X_\textrm{in}^{H_1}$ by ${\Omega^2}^*_{+,\mathscr{H}_{r_e}}$.

Incoming and outgoing subspaces for $H_e$ and $H_0$ were originally defined using the asymptotic velocity operators constructed in Section~\ref{sec:asymptotic_velocity_operators}. These operators were constructed on each of the stable subspaces of the respective orthogonal sum decompositions associated to each of the operators, they are:
\[P^+_e = \slim_{t\to+\infty}e^{itH_e}\Gamma^1e^{-itH_e},\quad P^{+}_0=\slim_{t\to+\infty}e^{itH_0}\Gamma^1e^{-itH_0},\]
and satisfy for any $J\in C_{\infty}(\mathbb{R})$:
\begin{equation}\label{eq:global_asymp_vel_op} \begin{gathered}J(P^+_e)= \slim_{t\to+\infty}e^{itH_e}J(\frac{r^*}{t})e^{-itH_e}, \\ J(P^+_0)=\slim_{t\to+\infty}e^{itH_0}J(\frac{r^*}{t})e^{-itH_0}. \end{gathered}\end{equation}
In terms of these operators, $X^{H_e}_\textrm{in}= \textrm{Ran} \bm{1}_{\mathbb{R}_-}(P^+_e)=\textrm{Ran}\bm{1}_{\{-1\}}(P^+_e)$. Using~\eqref{eq:global_asymp_vel_op}, one can show that $X^{H_e}_\textrm{in}$ as defined above coincides exactly with the image 
$\Omega^2_{+,\mathscr{H}_{r_e}}X^{H_1}_\textrm{in}$, for instance, for any $\phi \in \mathscr{H}$,
\[\bm{1}_{\{-1\}}(P^+_e)\phi =J(P^+_e)\phi = \lim_{t\to+\infty} e^{itH_e}J(\frac{r^*}{t})e^{-itH_e}\phi,\] for any $J\in C^{\infty}_0(\mathbb{R})$ such that $\supp J \subset (-\infty,0), J(-1)=1$.
Hence, for any $c_+ \in \mathscr{C}_+$:
\[\lim_{t \to +\infty} c_+(r^*)e^{-itH_e}\bm{1}_{\{-1\}}(P^+_e)\phi= \lim_{t\to+\infty} c_+(r^*)J(\frac{r^*}{t})e^{-itH_e}\phi=0.\]
The other inclusion is proved in a similar manner, one can show for example that:
 \begin{equation} \lim_{t\to +\infty} c_+(r^*)e^{-itH_e}\phi=0, \textrm{for any $c_+ \in \mathscr{C}_+$} \Rightarrow \phi \in \textrm{Ran}\bm{1}_{\{1\}}(P^+_e)^{\perp}. \end{equation} Indeed, let $\phi$ satisfy the condition and let $\psi \in \textrm{Ran}\bm{1}_{\{1\}}(P^+_e)$. A similar argument to the one above shows that for any $c_- \in \mathscr{C}_-$:
\[\lim_{t\to+\infty} c_-(r^*)e^{-itH_e}\psi=0.\] Choose now $c_+\in \mathscr{C}_+,c_-\in\mathscr{C}_-$ such that $c_+ + c_-=1$, then for $t\in \mathbb{R}$:
\begin{equation} 
\begin{aligned}
(\phi,\psi)&=(e^{-itH_e}\phi,e^{-itH_e}\psi),\\&=(c_+(r^*)e^{-itH_e}\phi,e^{-itH_e}\psi) + (e^{-itH_e}\phi,c_-(r^*)e^{-itH_e}\psi).
\end{aligned}
\end{equation}
By the Cauchy-Schwarz inequality, it follows that for any $t\in \mathbb{R}$:
\[ |(\psi,\phi)| \leq ||\psi|| ||c_+(r^*)e^{-itH_e}\phi|| + ||\phi|| ||c_-(r^*)e^{-itH_e}\psi||,\]
The right-hand side approaches $0$ as $t\to +\infty$ so that:
\[|(\psi,\phi)|=0.\]
We can therefore define a global wave operator from the absolutely continuous subspace of $H_1$ onto the \emph{external} direct sum $\textrm{Ran} \bm{1}_{\{-1\}}(P^+_e) \oplus \textrm{Ran} \bm{1}_{\{1\}}(P^+_0)$.
\begin{multline}
\begin{array}{rcl}\Omega^2_+~:~ X_\textrm{in}^{H_1}\oplus X_\textrm{out}^{H_1} &\longrightarrow & \textrm{Ran} \bm{1}_{\{-1\}}(P^+_e) \oplus \textrm{Ran} \bm{1}_{\{1\}}(P^+_0)\\ \\ (\phi_1,\phi_2) &\longmapsto& (\Omega^2_{+, \mathscr{H}_{r_e}}\phi_1, \Omega^2_{+,\mathscr{H}_{r_+}}\phi_2).\end{array}
\end{multline}

\subsection{Comparison III}
Although the results in Section~\ref{sec:sphericalsymmop} can be used to construct a scattering theory for $H_e$ and $H_0$ on the whole Hilbert space, the previous discussion shows that, for our needs, it only relevant to do this on 
$\textrm{Ran} \bm{1}_{\{-1\}}(P^+_e)$ for $H_e$ and on $\textrm{Ran} \bm{1}_{\{1\}}(P^+_0)$ for $H_0$.
The asymptotic profiles are given by:
\begin{equation}
\begin{gathered}
H_{-\infty}= \Gamma^1D_{r^*},\\
H_{+\infty}=\Gamma^1D_{r^*}+\left(\frac{a}{r_+^2+a^2}-\frac{a}{r_e^2+a^2} \right)p.
\end{gathered}
\end{equation}
The outgoing and incoming states are identical for both of these operators and given by:
\[\mathscr{H}^+= \textrm{Ran}\bm{1}_{\{1\}}(\Gamma^1),\quad \mathscr{H}^- = \textrm{Ran}\bm{1}_{\{-1\}}(\Gamma^1).\]
Due to the stability of the subspace under $\Gamma^1, H_e,H_{\pm\infty}$, the results in Section~\ref{sec:sphericalsymmop} prove that the following strong limits exist:
\begin{equation*} \begin{array}{l}\displaystyle \Omega^3_{+,\mathscr{H}_{r_+}}=\slim_{t\to+\infty} e^{itH_{+\infty}}e^{-itH_0}\bm{1}_{\mathbb{R}_+}(P^{+}_0),
\\ \displaystyle \Omega^3_{+,\mathscr{H}_{r_e}}=\slim_{t\to+\infty} e^{itH_{-\infty}}\left(T\exp\left(-i\int_0^t \tilde{f}(\Gamma^1s)\textrm{d}s\right)\right)^*e^{-itH_e}\bm{1}_{\mathbb{R}_-}(P^{+}_e), 
\vspace{.1in}
\\ \displaystyle \tilde{\Omega}^3_{+,\mathscr{H}_{r_+}}=\slim_{t\to+\infty} e^{itH_0}e^{-itH_{+\infty}}\bm{1}_{\mathbb{R}_+}(\Gamma^1),\\ \displaystyle \tilde{\Omega}^3_{+,\mathscr{H}_{r_e}}=\slim_{t\to+\infty} e^{itH_e}T\exp\left(-i\int_0^t \tilde{f}(\Gamma^1s)\textrm{d}s\right)e^{-itH_{-\infty}}\bm{1}_{\mathbb{R}_-}(\Gamma^1),
  \end{array}\end{equation*}
One also has~:~ $\tilde{\Omega}^3_{+,\mathscr{H}_{r_+}}=  {\Omega^3}^*_{+,\mathscr{H}_{r_+}}$ and similarly for $\mathscr{H}_{r_e}$, this gives rise to a unitary map:
\begin{equation*}
 \begin{array}{ccc} \Omega^3_+~:~\textrm{Ran} \bm{1}_{\{-1\}}(P^+_e) \oplus \textrm{Ran} \bm{1}_{\{1\}}(P^+_0) &\longrightarrow & \mathscr{H}^-\oplus\mathscr{H}^+ =\mathscr{H}\vspace{.1in} \\ (\phi_1,\phi_2) &\longmapsto& (\Omega^3_{+, \mathscr{H}_{r_e}}\phi_1, \Omega^3_{+,\mathscr{H}_{r_+}}\phi_2).\end{array}
\end{equation*}
Finally, composition of $\Omega^1_+,\Omega^2_+, \Omega^3_+$ yields a unitary map $W_+$ between $P_c(H) = X_\textrm{in}^{H}\oplus X_\textrm{out}^{H}$ and $\mathscr{H}$, where:
\[X_\textrm{in}^{H}= (\ker \Omega^2_{+,\mathscr{H}_{r_e}}\Omega^1_+)^\perp, \quad X_\textrm{out}^{H}= (\ker \Omega^2_{+,\mathscr{H}_{r_+}}\Omega^1_+)^\perp,\]
given by:
\begin{equation*}
\begin{array}{rcl} W_+~:~X_\textrm{in}^{H} \oplus X_\textrm{out}^{H} &\longrightarrow & \mathscr{H}^-\oplus\mathscr{H}^+ =\mathscr{H} \vspace{.1in}\\ \phi_1+\phi_2 &\longmapsto& \Omega^3_{+, \mathscr{H}_{r_e}}\Omega^2_{+,\mathscr{H}_{r_e}}\Omega^1_+\phi_1 + \Omega^3_{+,\mathscr{H}_{r_+}}\Omega^2_{+,\mathscr{H}_{r_+}}\Omega^1_+\phi_2.\end{array}
\end{equation*}

\begin{rem}
\label{rem:asymptotic_velocity_full_dynamics}
A simple application of the above result is to define the asymptotic velocity operator for the full dynamics. It is defined by the limits for $J\in C_\infty(\mathbb{R})$,
\[J(P^+)=\slim_{t\to+\infty} e^{iHt}J(\frac{r^*}{t})e^{-iHt}=W_+^*J(\Gamma^1)W_+,\]
Using the results discussed in Section~\ref{sec:asymptotic_velocity_operators}, it follows that~: 
$P^+=W_+^*\Gamma^1W_+.$
\end{rem}
\subsection{Scattering for the Dirac operator}
We now return to the notations we adopted prior to Section~\ref{intermediatewaveoperators}, where we dropped the explicit dependence of our operator $H^p$ for notational convenience. We recall from Section~\ref{phi-invariance} that $H^p$ coïncides with the full Dirac operator on each of the subspaces associated with the eigenvalue $p \in \mathbb{Z}+\frac{1}{2}$ of $D_\phi$. The global wave operators obtained in the previous section, although defined on all of $\mathscr{H}$, also depend on the parameter $p$. However the $p$-eigenspace is stable so that to obtain the scattering theory for the Dirac operator one only need to reassemble each of the harmonics. Since the Dirac operator has no pure point spectrum\footnote{see again \cite{Belgiorno:2009aa}}, there is no need to project onto the absolutely continuous subspace. Therefore, we state our final theorem~:
\begin{theo}
\label{th:final_result}
For any \(\displaystyle \phi=\!\!\sum_{p\in \mathbb{Z}+\frac{1}{2}}\phi_p(r^*,\theta)e^{ip\phi} \in \mathscr{H}\) set~:
\begin{equation} \mathscr{P}^+\phi =\!\! \sum_{p\in\mathbb{Z}+\frac{1}{2}}P^{+}_p\phi_p e^{ip\phi},  \end{equation} 
then $\mathscr{P}_+$ is a bounded symmetric operator with spectrum $\{-1,1\}$, and for any $J\in C_\infty(\mathbb{R}),$
\[J(\mathscr{P}_+)= \slim_{t\to+\infty} e^{iHt}J(\frac{r^*}{t})e^{-iHt}. \] Moreover, defining~:
\[ X_\textrm{in} = \textrm{Ran~}\bm{1}_{\{-1\}}(\mathscr{P}^+), \quad X_\textrm{out}=\textrm{Ran~}\bm{1}_{\{1\}}(\mathscr{P}^+),\]
then, $\mathscr{H}= X_\textrm{in} \oplus X_\textrm{out}$ and the operator~:
\begin{equation}\mathfrak{W}^+\phi =\!\! \sum_{p\in\mathbb{Z}+\frac{1}{2}}W_{+}^p\phi_p e^{ip\phi}, \end{equation}
is a unitary operator such that~:
\[\mathfrak{W}_+X_\textrm{in}=\mathscr{H}_- , \quad \mathfrak{W}_+ X_\textrm{out}=\mathscr{H}_+,\] and for the full Dirac operator $H$ defined by Equation~\eqref{def_H}~:
\[H_{-\infty}\mathfrak{W}_+\bm{1}_{\{-1\}}(\mathscr{P}^+)+H_{+\infty}\mathfrak{W}_+\bm{1}_{\{1\}}(\mathscr{P}^+) = \mathfrak{W}_+H. \]
\end{theo}

\section{Conclusion}

In this paper we have proposed an analytical construction for a scattering theory for particules in a region situated between a double and simple horizon of an extreme Kerr-de Sitter blackhole. The presence of the simple horizon alone simplified the problem considerably, being an obstruction to the existence of pure-point spectrum, and the existence of a conjugate operator in the sense of Mourre theory ruled out the possibility for any singular continuous spectrum.  The setting was therefore ideal for an analytic scattering theory.

We found that, from an analytical point of view, the double horizon region was analogous to that of spacelike infinity in Kerr-Newmann spacetime. The theory is in fact slightly easier because the mass terms do not persist at the horizons, meaning that things appear to boil down to the massless case. As in this case, the reasoning hinges on the ability to obtain a minimal velocity estimate.

The main difference and novelty is that the double horizon exacerbates the effects of the rotation of the black hole by complicating the structure of the angular operator; the mass also plays a lesser role here. However, this did not prove to be an essential difficulty for the analytic methods used in this paper, which is another illustration of their robustness. 

The methods used here do however have the clear disadvantage of not being very geometrical. In some sense, the study of the effects of the double horizon is reduced to the distinction between long and short-range potentials; it would be considerably more satisfying to seek a proof of the results in this paper with a clearer geometrical meaning. 


\newpage 

\appendix
\section{Expression of the Dirac equation}
\label{app:dirac_equation}
In this appendix, we will sketch the calculations that lead to the operator $H$ studied in the main text. We recall that the massive Dirac equation for a spin-$\frac{1}{2}$ Dirac spinor $(\phi_A,\chi^{A'})$ is~:
 \begin{equation}
 \label{eq:dirac}
 \left\{ \begin{array}{c} \nabla^{AA'} \phi_A = \mu \chi^{A'} \\ \nabla_{AA'} \chi^{A'} = -\mu \phi_A \ \end{array}, \quad\mu=\frac{m}{\sqrt{2}}\right. .
 \end{equation}
 To convert this into a system of four scalar equations we will use the local spin-connection forms ${\alpha^{\bm{A}}_{\,\,\bm{B}}}_a$ of a local normalised spin frame $(\varepsilon^A_{\bm{A}})_{\bm{A} \in \{0,1\}}$ defined by:
 \[{\alpha^{\bm{A}}_{\,\,\bm{B}}}_a= \varepsilon^{\bm{A}}_{B}\nabla_a \varepsilon^{B}_{\bm{B}}.\]
 Given any orthonormal frame $g^a_{\bm{a}}$ and a normalised spin frame $\varepsilon^A_{\bm{A}}$ such that the vector fields:
\begin{align} \label{eq:np_spin} l^a=\varepsilon^A_0\varepsilon^{A'}_{0'};  \quad n^a=\varepsilon^A_1 \varepsilon^{A'}_{1'}; \quad m^a =\varepsilon^A_0\varepsilon^{A'}_{1'}; \end{align}
of the Newman-Penrose tetrad $(l^a,n^a,m^a,\bar{m}^a)$ satisfy:
\begin{equation}
\label{eq:np_spin2}
\left\{ \begin{array}{c} l^a=\frac{g^a_0+g^a_1}{\sqrt{2}},  \\ n^a = \frac{g^a_0 - g^a_1}{\sqrt{2}}, \\ m^a= \frac{g^a_2 + ig^a_3}{\sqrt{2}}, \end{array}\right.\end{equation}
then the spin connection forms are given in terms of the local connection forms $\omega^{\bm{i}}_{\,\,\bm{j}}$ in the basis $g^a_{\bm{a}}$ by:
\begin{equation}\begin{gathered}
\alpha^{0}_{\,\, 0} = \frac{\omega^{0}_{\,\, 1} + i \omega^{2}_{\,\, 3}}{2}, \quad \alpha^{1}_{\,\, 0} = \frac{\omega^{2}_{\,\, 0} + \omega^{2}_{\,\, 1}}{2} + i \frac{\omega^{3}_{\,\, 0} + \omega^{3}_{\,\, 1}}{2}, \\
\alpha^{0}_{\,\, 1} = \frac{\omega^{2}_{\,\, 0} - \omega^{2}_{\,\, 1}}{2} - i \frac{\omega^{3}_{\,\, 0} - \omega^{3}_{\,\, 1}}{2}. \end{gathered}\end{equation} 
A spin connection is a $\mathfrak{sl}(2,\mathbb{C})$-valued one-form, so necessarily: \[\alpha^{1}_{\,\, 1} = -\alpha^{0}_{\,\,0}.\] In terms of the covariant derivative, this is equivalent to the requirement that $\nabla_a \varepsilon_{AB} = 0$. 
The forms ${\alpha^{\bm{A'}}_{\,\, \bm{B'}}}_a= \varepsilon^{\bm{A'}}_{B'}(\nabla_a \varepsilon^{B'}_{\bm{B'}})$ satisfy:
\begin{equation} \label{eq:complex_spin_connection} {\alpha^{\bm{A'}}_{\,\, \bm{B'}}}_a= {\overline{\alpha^{\bm{A}}_{\,\,\bm{B}}}}_a \end{equation}
\begin{rem}
It should be remarked that our conventions differ slightly from those in \cite{Penrose:1984aa}, namely, we identify $\mathbb{R}^4$ to $H(2,\mathbb{C})$ via the isomorphism:
\[ \varphi : \begin{array}{ccc} \mathbb{R}^4 &\longrightarrow& H(2,\mathbb{C})\\ \left(\begin{array}{c} x_0 \\ x_1 \\ x_2 \\ x_3 \end{array}\right) &\longmapsto& \left(\begin{array}{cc} x_0+x_1 & x_2-ix_3 \\ x_2+ix_3 & x_0-x_1 \end{array}\right) \end{array}\]
\end{rem} 
\begin{rem} Consider the Lie group morphism $ \Lambda :SL(2,\mathbb{C}) \rightarrow SO_+(1,3)$ defined by associating to any $A\in SL(2,\mathbb{C})$ the matrix $\Lambda(A)$ of the linear map $u$ defined by $u(x) = \varphi^{-1}(A \varphi(x)A^*), x \in \mathbb{R}^4$ expressed in the canonical basis of $\mathbb{R}^4$. Then, viewing $\bm{\omega}=(\omega^{\bm{i}}_{\,\, \bm{j}} )_{\bm{i},\bm{j} \in \llbracket 0,3 \rrbracket}$ and $\bm{\alpha}=(\alpha^{\bm{A}}_{\,\, \bm{B}})_{\bm{A},\bm{B} \in \{0,1\}} $ as matrix valued one-forms, it follows that for any $(p,v) \in TM$:
\[ \bm{\alpha}_p(v) = \Lambda^{-1}_{*} (\bm{\omega}_p(v)),\] where $\Lambda_{*}$ is the Lie algebra isomorphism induced by $\Lambda$.
\end{rem}
Once a choice of spin-frame has been made, equation~\eqref{eq:dirac} can be written as four scalar equations in terms of the components $\phi^{\bm{A}},\chi_{\bm{A}'}$ of the spinor fields. For instance, the equation:
\[ \nabla_{AA'}\phi^A = - \mu \chi_{A'},\]
becomes,
\begin{equation*}
\nabla_{\bm{AA'}}\phi^{\bm{A}} + \phi^{\bm{A}}{\alpha^{\bm{B}}_{\,\, \bm{A}}}_{\,CC'}\varepsilon^C_{\bm{B}}\varepsilon^{C'}_{\bm{A'}} = -\mu\chi_{\bm{A'}}.
\end{equation*}
For $\bm{A}=0'$, this translates to :
\[ l^a\nabla_a \phi^0 + \bar{m}^a\nabla_a\phi^1 +\phi^0\left( {\alpha^0_{\,\, 0}}_a l^a + {\alpha^1_{\,\,0}}_a \bar{m}^a \right) + \phi^1\left( {\alpha^0_{\,\,1}}_al^a + {\alpha^1_{\,\, 1}}_a \bar{m}^a \right) = - \mu \chi_{0'},\]
or, equivalently:
\[ l^a\nabla_a \phi_1 -\bar{m}^a\nabla_a\phi_0 +\phi_1\left( {\alpha^0_{\,\, 0}}_a l^a + {\alpha^1_{\,\,0}}_a \bar{m}^a \right) - \phi_0\left( {\alpha^0_{\,\,1}}_al^a + {\alpha^1_{\,\, 1}}_a \bar{m}^a \right) = \mu \chi^{1'}. \]
Overall, we obtain the following system of equations for the components:
\begin{equation}
\label{eq:dirac2}
\left\{\begin{array}{llll}   
l^a \nabla_a \chi^{0'} + m^a \nabla_a \chi^{1'} + \chi^{0'}\overline{F} + \chi^{1'}\overline{G}= - \mu \phi_0, \\
\bar{m}^a \nabla_a \chi^{0'} + n^a\nabla_a \chi^{1'} + \chi^{0'}\overline{G_1}+\chi^{1'} \overline{F_1} = - \mu \phi_1,\\
m^a\nabla_a\phi_1 - n^a \nabla_a\phi_0 + \phi_1G_1 - \phi_0 F_1= - \mu \chi^{0'},\\
l^a\nabla_a \phi_1 -\bar{m}^a\nabla_a\phi_0 +\phi_1 F - \phi_0G = \mu \chi^{1'},
 \end{array} \right.
\end{equation}
where we have defined:
\begin{align*} F&= {\alpha^0_{\,\, 0}}_al^a + { \alpha^1_{\,\, 0}}_a \bar{m}^a, \\ G&={\alpha^0_{\,\,1}}_a l^a + {\alpha^1_{\,\, 1}}_a \bar{m}^a, \\ F_1 &= {\alpha^0_{\,\, 1}}_a m^a + {\alpha^1_{\,\, 1}}_a n^a, \\ G_1&={\alpha^0_{\,\, 0}}_a m^a + {\alpha^1_{\,\, 0}}_a n^a, \end{align*}
and used the fact that, by~\eqref{eq:complex_spin_connection}, for any complex vector fields $u^a, v^a$:
\[{\alpha^{\bm{A'}}_{\,\, \bm{B'}}}_a \bar{u}^a + {\alpha^{\bm{C'}}_{\,\, \bm{D'}}}_a\bar{v^a} = \overline{{\alpha^{\bm{A}}_{\,\, \bm{B}}}_a u^a + {\alpha^{\bm{C}}_{\,\, \bm{D}}}_av^a }.\]

\subsection{Dirac equation in the ``Boyer-Lindquist'' frame}
We will first use the results in~\cite{Borthwick:2018aa} to write the Dirac equation in a normalised spin-frame linked by equations~\eqref{eq:np_spin} and~\eqref{eq:np_spin2} to the orthonormal frame~:
 \begin{equation}
\begin{aligned} g_0^a\frac{\partial}{\partial x^a}&= \frac{\Xi}{\rho \sqrt{\Delta_r}}\left((r^2+a^2)\partial_t+ a\partial_\varphi \right), \\ g_1^a\frac{\partial}{\partial x^a}&=\frac{\sqrt{\Delta_r}}{\rho} \partial_r, \,\, g_2^a\frac{\partial}{\partial x^a} = \frac{\sqrt{\Delta_\theta}}{\rho}\partial_\theta, \\ g_3^a \frac{\partial}{\partial x^a}&= \frac{\Xi}{\sin \theta \sqrt{\Delta_\theta} \rho}\left( \partial_\varphi + a \sin^2\theta \partial_t \right). \end{aligned}\end{equation}
In such a frame, the expressions for $F,G,F_1,G_1$ are given by: 
\begin{equation*}\begin{aligned} &F= \frac{1}{2\sqrt{2}\sqrt{\Delta_r} \rho^3}\left( \frac{\Delta_r'}{2}\rho^2 + \Delta_r\tilde{r} \right),\\ 
&G= \frac{1}{2\sqrt{2} \sqrt{\Delta_\theta} \sin\theta \rho^3} \left( ia\Delta_\theta\sin^2\theta\tilde{r}+\cos\theta \rho^2(1+a^2l^2\cos(2\theta)) \right),\\ &F_1=-F,\\ &G_1=G,\end{aligned}\end{equation*}
where $\Delta_r' = \frac{\partial \Delta_r}{\partial r}$ and $\tilde{r}=(r+ia\cos\theta)$. 

\noindent In matrix form, with $\bm{\psi} = {}^t \left(\phi_0,\phi_1,\chi^{0'},\chi^{1'}\right)$,~\eqref{eq:dirac2} is then:
\[ i (\gamma^\mu\partial_\mu + V)\bm{\psi} = m \bm{\psi}. \]
In the above: 
\[ V = \sqrt{2}\left(\begin{array}{cccc} 0 & 0 & i\bar{F} & i \bar{G} \\ 
0 &0 & i\bar{G} & -i \bar{F} \\ iF & iG & 0 & 0 \\ iG & -iF & 0 &0 \end{array} \right), \]
\begin{align*} \gamma^t =\frac{\Xi(r^2+a^2)}{\sqrt{\Delta_r\rho^2}}\left(\begin{array}{cc} 0 & iI_2 \\ -iI_2 & 0 \end{array}\right) -i\frac{a\sin\theta\Xi}{\sqrt{\Delta_\theta\rho^2}}\left(\begin{array}{cc} 0 & \sigma_y \\ \sigma_y & 0  \end{array}\right), &\\ \gamma^r=i\sqrt{\frac{\Delta_r}{\rho^2}}\left(\begin{array}{cc} 0 & \sigma_z \\ \sigma_z & 0 \end{array}\right), \quad \gamma^\theta = i \sqrt{\frac{\Delta_\theta}{\rho^2}} \left(\begin{array}{cc} 0 & \sigma_x \\ \sigma_x & 0 \end{array}\right),\hspace{.1in}\\ \gamma^\varphi=\frac{a\Xi}{\sqrt{\Delta_r\rho^2}}\left(\begin{array}{cc} 0 &iI_2 \\ -iI_2 & 0  \end{array}\right)-i\frac{\Xi}{\sqrt{\Delta_\theta\rho^2}\sin\theta}\left(\begin{array}{cc} 0 & \sigma_y \\ \sigma_y & 0  \end{array}\right).\end{align*}
The $\gamma^\mu$ are the so-called ``gamma matrices'' that satisfy the Clifford algebra anti-commutation relations~: \[\{\gamma^\mu, \gamma^\nu \} = 2g^{\mu\nu}\textrm{Id}_4.\] $\sigma_x,\sigma_y,\sigma_z$ are the Pauli matrices,
\begin{align*}\sigma_x=\left(\begin{array}{cc}0 & 1 \\ 1 & 0 \end{array}\right),\quad \sigma_y=\left(\begin{array}{cc} 0 & -i \\ i & 0 \end{array}\right),\quad \sigma_z=\left(\begin{array}{cc} 1 & 0 \\ 0 & -1 \end{array}\right)= -i\sigma_x\sigma_y. \end{align*}

\subsubsection{Change of spin-frame}
In the main text, we chose instead to work in a normalised spin-frame defined by Equations~\eqref{eq:np_spin} and~\eqref{eq:np_spin2} from the orthonormal frame~:
\[\begin{gathered}{g'}_0^a= \frac{\nabla^a t}{\sqrt{|\nabla^a t\nabla_a t|}},\,\, {g'}_{1}^a \frac{\partial}{\partial x^a}= \frac{1}{\sqrt{-g_{rr}}}\partial_{r},\\ {g'}_{2}^a\frac{\partial}{\partial x^a}=\frac{1}{\sqrt{-g_{\theta\theta}}}\partial_\theta, \,\,{g'}_{3}^a\frac{\partial}{\partial x^a}=\frac{1}{\sqrt{-g_{\varphi\varphi}}}\partial_\varphi \end{gathered}\]
The matrix $P$ of the Lorentz transformation $L_a^b$ that sends $g_{\bm{a}}^a$ to ${g'}_{\bm{a}}^a$ is given by:
\begin{equation} P = M_{g'^a_{\bm{a}}, g^b_{\bm{b}}}(Id)= \left(\begin{array}{cccc} \frac{\sqrt{\Delta_\theta}(r^2+a^2)}{\sigma} & 0 & 0 & -\frac{a\sin{\theta}\sqrt{\Delta_r}}{\sigma} \\ 0 & 1 &0 &0 \\ 0 &0 &1 &0 \\ -\frac{a\sin{\theta}\sqrt{\Delta_r}}{\sigma} & 0 & 0 &\frac{\sqrt{\Delta_\theta}(r^2+a^2)}{\sigma}  \end{array}\right),\end{equation}
where we have defined $\sigma^2 = \Delta_\theta(r^2+a^2)^2 - \Delta_r a^2\sin^2\theta$. 
Up to sign, the spin transformation $S \in SL(2;\mathbb{C})$ that corresponds to $P$ is:
\begin{equation}
S = \left( \begin{array}{cc} \sqrt{\frac{\sigma_+}{2\sigma}} & \frac{ia\sin\theta\sqrt{\Delta_r}}{\sqrt{2\sigma\sigma_+}} \\ -\frac{ia\sin\theta\sqrt{\Delta_r}}{\sqrt{2\sigma\sigma_+}} & \sqrt{\frac{\sigma_+}{2\sigma}} \end{array} \right),
\end{equation}
in the above formula $\sigma_+ = \sigma + \sqrt{\Delta_\theta}(r^2+a^2)$. It is useful to note that $\sigma_+$ satisfies:
\[\sigma_+^2-a^2\sin^2\theta \Delta_r =2\sigma\sigma_+. \]
The appropriate change of basis matrix in ${\mathbb{S}_{A}}_p\oplus \mathbb{S}^{A'}_p$ at each point $p$ of block II is given by:
\begin{equation}
\tilde{P} = \left( \begin{array}{cc} {}^tS^{-1} & 0 \\ 0 & \bar{S}  \end{array}\right) = \sqrt{\frac{\sigma_+}{2\sigma}}I_4 + \frac{a\sin\theta\sqrt{\Delta_r}}{\sqrt{2\sigma\sigma_+}}\left(\begin{array}{cc} -\sigma_y & 0 \\ 0 & \sigma_y \end{array}\right).
\end{equation}
The equation satisfied by $\bm{\psi'}=\tilde{P}^{-1} \bm{\psi}$ is hence:
\begin{equation} i\tilde{P}^{-1}(\gamma^\mu\partial_\mu + V)\tilde{P} \bm{\psi'} = m\bm{\psi'}. \end{equation}
The left-hand side is:
\[i\tilde{P}^{-1}(\gamma^\mu\partial_\mu + V)\tilde{P} = i\left(\tilde{\gamma}^\mu\partial_\mu + \tilde{V} + \tilde{P}^{-1}\gamma^\mu\frac{\partial \tilde{P}}{\partial x^\mu}\right),\]
where: 
\begin{equation}
\begin{gathered}
 \gamma^r=\tilde{\gamma}^r, \quad \gamma^\theta=\tilde{\gamma}^\theta, \quad \tilde{V} = V, \quad \tilde{\gamma^t} = \frac{\Xi\sigma}{\sqrt{\Delta_r\Delta_\theta}\rho}\left(\begin{array}{cc} 0 & iI_2 \\ -iI_2 & 0 \end{array} \right) ,\\ \tilde{\gamma}^\varphi = \frac{a\Xi q^2\rho}{\sigma\sqrt{\Delta_r\Delta_\theta}}\left(\begin{array}{cc} 0 & iI_2 \\ -iI_2 & 0 \end{array}\right) -i\frac{\Xi\rho}{\sigma\sin\theta}\left(\begin{array}{cc} 0 & \sigma_y \\ \sigma_y & 0\end{array}\right), \\ \tilde{P}^{-1}\gamma^r\frac{\partial \tilde{P}}{\partial r} = \frac{\sqrt{\Delta_r}}{\rho}f_r\left( \begin{array}{cc} 0 & -\sigma_x \\ \sigma_x & 0 \end{array} \right), \\ \tilde{P}^{-1}\gamma^\theta\frac{\partial \tilde{P}}{\partial \theta}=\frac{\sqrt{\Delta_\theta}}{\rho}f_\theta \left(\begin{array}{cc} 0 & \sigma_z \\ -\sigma_z & 0 \end{array} \right). \end{gathered}\end{equation} 
In the above formulae, we have introduced the following notations:
 \begin{equation*}\begin{gathered}q^2 = (\Delta_\theta(r^2+a^2)-\Delta_r)\rho^{-2}, \\ f_r=\frac{a\sin\theta\sqrt{\Delta_\theta}}{2\sigma^2\sqrt{\Delta_r}}\left(-\frac{\Delta'_r}{2}(r^2+a^2)+2r\Delta_r \right), f_\theta =-\frac{a\sqrt{\Delta_r}(r^2+a^2)\cos\theta \Xi}{2\sigma^2\sqrt{\Delta_\theta}}.\end{gathered}\end{equation*}
\noindent The final step is to replace the spinor with the spinor density. After trivialising the density bundle, this amounts to the replacement~:
\begin{equation}\label{eq:spinor_density} \Phi=\underbrace{\left(\frac{\Delta_r \rho^2\sigma^2}{\Delta_\theta(r^2+a^2)^2\Xi^4} \right)^{\frac{1}{4}} }_{\alpha(r,\theta)^{-1}}\bm{\psi'}. \end{equation} $\Phi$ satisfies almost the same equation as $\bm{\psi'}$ except for two additional terms:
\[ i\gamma^{1}\partial_r(\ln \alpha(r,\theta)) \Phi + i\gamma^2 \partial_\theta(\ln \alpha(r,\theta)) \Phi   \]
\noindent Overall the equation becomes:
\begin{equation}\label{eq:dirac3} i\tilde{\gamma}^0\partial_t \Phi + i\tilde{\gamma}^1\partial_r \Phi + i\tilde{\gamma}^2\partial_\theta \Phi +i\tilde{\gamma}^3\partial_\varphi \Phi + iV_1\Phi = m\Phi, \end{equation}
\noindent with:
\begin{equation}
V_1 = \left(\begin{array}{cccc} 0 & 0 & i\bar{\tilde{F}} & i \bar{\tilde{G}} \\ 
0 &0 & i\bar{\tilde{G}} & -i \bar{\tilde{F}} \\ i\tilde{F} & i\tilde{G} & 0 & 0 \\ i\tilde{G} & -i\tilde{F} & 0 &0 \end{array} \right),
\end{equation}
\begin{equation}
\begin{aligned}
\tilde{F}&= \sqrt{2}F + i \frac{\sqrt{\Delta_\theta}}{\rho}f_\theta + \frac{\sqrt{\Delta_r}}{\rho}\partial_r \ln \alpha(r,\theta), \\
\tilde{G}&= \sqrt{2}G - i \frac{\sqrt{\Delta_r}}{\rho}f_r + \frac{\sqrt{\Delta_\theta}}{\rho}\partial_\theta \ln \alpha(r,\theta).
\end{aligned}
\end{equation}
%
\noindent Rewriting \eqref{eq:dirac3} as an evolution equation, and introducing $\mathfrak{D}_{S^2}$, the Dirac operator on the 2-sphere, we obtain the following form of the Dirac equation:
\begin{multline}  i\partial_t \Phi + i\frac{\Delta_r\sqrt{\Delta_\theta}}{\Xi\sigma}\Gamma^1\partial_r\Phi - \frac{\sqrt{\Delta_r}\Delta_\theta}{\Xi\sigma}\mathfrak{D}_{S^2}\Phi +\frac{ia q^2\rho^2}{\sigma^2}\partial_\varphi\Phi \\ +\frac{i\sqrt{\Delta_r\Delta_\theta}}{\sigma\sin\theta}\left(\frac{\rho^2}{\sigma}-\frac{\sqrt{\Delta_\theta}}{\Xi} \right)\Gamma^3 \partial_\varphi\Phi +\frac{i\sqrt{\Delta_r\Delta_\theta}\rho}{\sigma\Xi} \tilde{V_1}\Phi = \frac{\sqrt{\Delta_r\Delta_\theta}}{\Xi\sigma}\rho m\Gamma^0\Phi \end{multline}
\noindent This leads to the operator $H$ given in the main text.

\section{Proofs of the technical results of Section~\ref{analytic_framework}}
\subsection{Proof of Lemma~\ref{lemme:asymptotic_behaviour_symbol}}
\label{app:asymp_behav}
\begin{proof}
Remark first that, from equations~\eqref{eq:asymp_cosmo} and \eqref{eq:asymp_double}, since $r_e$ is a double root of the polynomial $\Delta_r$, we have:
\begin{equation}\label{eq:asympdeltar} \begin{gathered} \Delta_r=\underset{r^*\to-\infty}{O}\left(\frac{1}{{r^*}^2} \right), \quad
\Delta_r=\underset{r^*\to+\infty}{O}\left(e^{-2 \kappa r^*}\right), \\  \Delta'_r=\underset{r^*\to-\infty}{O}\left(\frac{1}{r^*}\right). \end{gathered} \end{equation}
Hence:
\begin{equation} \begin{aligned} &\alpha(r^*)=\underset{r^*\to-\infty}{O}\left(\frac{1}{{r^*}^2} \right), &&
\alpha(r^*)=\underset{r^*\to+\infty}{O}\left(e^{-2 \kappa r^*}\right), \\
& \partial_r \alpha(r^*)=\underset{r^*\to-\infty}{O}\left(\frac{1}{r^*}\right), &&
 \partial_r \alpha(r^*)=\underset{r^*\to+\infty}{O}(1).
 \end{aligned}
\end{equation}
For any $n\geq 2$, it is easy to see that $\partial^n_r \alpha(r^*) = O(1).$
Now, $\partial_{r^*}\alpha(r^*)=\alpha(r^*)\partial_r \alpha(r^*) $, so we have the correct behaviour at infinity after the first derivative. We claim that for $n\geq 1$:
\begin{equation}\label{eq:formederivalpha} \partial^n_{r^*}\alpha(r^*) = \sum_{k=1}^{n} f_k(r^*)(\partial_r\alpha(r^*))^{\beta_k}(\alpha(r^*))^k ,\end{equation}
where $\alpha_k \in \mathbb{N}$, $f_k \in \Pi$ and $ \beta_k + 2k \geq n+2$ for each $k \in \llbracket 1,n \rrbracket$.
This is obvious for $n=1$ and if such a relationship is true for some $n\geq 1$, after differentiation one has:
\begin{equation}\begin{aligned} \partial^{n+1}_{r^*}\alpha(r^*) = &\sum_{k=1}^{n} \partial_r f_k(r^*)(\partial_r\alpha(r^*))^{\beta_k}(\alpha(r^*))^{k+1} \\&+\beta_kf_k(r^*)\partial^2_r\alpha(r^*)(\partial_r\alpha(r^*))^{\alpha_k-1}(\alpha(r^*))^{k+1}
\\ &+\sum_{k=1}^n f_k(r^*)(\partial_r\alpha(r^*))^{\beta_k +1}(\alpha(r^*))^k. \end{aligned}\end{equation}
Therefore, $\partial^{n+1}_{r^*}\alpha(r^*)$ satisfies~\eqref{eq:formederivalpha}, with:
\begin{equation*} \begin{gathered} \tilde{\beta}_{n+1}= \max(0,\beta_n-1), \\ \tilde{f}_{n+1}=\partial_r f_n(\partial_r\alpha)^{\beta_n-\tilde{\beta}_{n+1}} + \beta_nf_n\partial^2_r\alpha,\\ \tilde{f}_1 = f_1=1, \,\, \tilde{\beta}_1=\beta_1+1=n+1.
\end{gathered} \end{equation*}
For $k \in \llbracket 2, n \rrbracket$,
\begin{equation*} \begin{gathered}
\\ \tilde{\beta}_k=\min(\beta_{k}+1,\max(0,\beta_{k-1}-1)),
\\ \begin{aligned}\tilde{f}_k=\partial_rf_{k-1}(\partial_r\alpha)^{\beta_{k-1}-\tilde{\beta}_k}+\beta_{k-1}f_{k-1}\partial^2_r\alpha(\partial_r \alpha(r^*))^{\beta_{k-1}-1-\tilde{\beta}_k}\\+f_k(\partial_r\alpha)^{\beta_k+1-\tilde{\beta}_k}.\end{aligned}
\end{gathered} \end{equation*}
The $\tilde{f}_k$ clearly satisfy the required hypothesis; if $\tilde{\beta}_k \neq 0$, then, either $\tilde{\beta_k}=\beta_{k}+1$ or $\tilde{\beta}_k=\beta_{k-1}-1$. In the first case, then:
\[\tilde{\beta}_k +2k \geq n+4,\]
in the second case:
\[\tilde{\beta}_k + 2k \geq n+2 +2 -1 = n+3.\]
If $\tilde{\beta}_k=0$, then necessarily this implies $\beta_{k-1} \leq 1$. By hypothesis, $\beta_{k-1}$ satisfies: $\beta_{k-1} + 2k  \geq n+4 $,  so, $2k \geq n+3$, and the hypothesis is equally satisfied. Hence, the result follows by induction.
The asymptotics can now be read from~\eqref{eq:formederivalpha}, each term in the sum is $O(\alpha)=O(e^{-2\kappa r^*})$ at $r^* \to +\infty$ and every term in the sum is $O({r^*}^{-(n+2)})$ at $r^*\to -\infty$.
\end{proof}

\subsection{Fàa di Bruno formula}
\label{app:fdbformula}
Let $f,g\in C^{\infty}(\mathbb{R})$, then for any $n\geq 1$:
\begin{equation*}
\begin{gathered}
(f\circ g)^{(n)}= \!\!\!\!\!\!\!\!\!\!\!\!\sum_{(m_1,\dots,m_n)\in I_n} \frac{n!f^{(m_1+\dots m_n)}\circ g}{m_1! 1!^{m_1}m_2!2!^{m_2}\dots m_n! n!^{m_n}} \prod_{j=1}^n\left( g^{(j)}\right)^{m_j}, \\
I_n=\{(m_1,\dots,m_n)\in \mathbb{N}^n, \sum_{j=1}^n jm_j = n \}.\end{gathered} \end{equation*}

\subsection{Proofs of Lemmata~\ref{lemme-decomposition},~\ref{lemme:core_S} and~\ref{lemme:equiv_quadratic_form}}
\label{app:proof_lemme_fin_2}
\begin{proof}[Proof of Lemma~\ref{lemme-decomposition}]
It is clear that $A$ is densely defined. In order to show that $A$ is closed, denote by $P_k$ the orthogonal projection onto $X_k$ for each $k\in \mathbb{N}$ and suppose that $(x^m)_{m\in\mathbb{N}}$ is a sequence of points of $D(A)$ such that $x^m \to x$ and $Ax^m \to y$ in $X$. Then for any $k\in \mathbb{N}$, $P_k x^m \to P_k x$ and $P_kAx^m=A_kP_kx^m \to P_k y$ by definition, but since $A_k$ is closed, it follows that $P_kx\in D(A_k)$ and $P_ky=A_kP_kx$. Thus~: \[\sum_k ||A_kP_kx||^2 = \sum_k ||P_ky||^2 <+\infty,\] so $x\in D(A)$ and $Ax=y$.

To prove that $A$ is self-adjoint we show that $A+z$ has dense range for any $z\in \mathbb{C}\setminus \mathbb{R}$. Let $y\in X$ be such that $(Ax+zx,y)=0$ for any $x\in D(A)$. In particular, for each $k\in \mathbb{N}$, and every $x\in D(A_k)$, $(A_k x+zx,P_ky)=0$, but then, since $A_k$ is self-adjoint, $P_ky=0$ for any $k\in \mathbb{N}$, i.e. $y=0$.
\end{proof}

\medskip

\begin{proof}[Proof of Lemma~\ref{lemme:core_S}]
For any $k,n$, $H_0^{k,n}$ is self-adjoint on $D(H_0^{k,n})=b_{k,n}^{-1}([H^1(\mathbb{R})]^4)$ ($b_{k,n}$ is defined by equation~\eqref{def_func_b_kn})and $[\mathscr{S}(\mathbb{R})]^4$ is dense in $[H^1(\mathbb{R})]^4$. Denote by $P_{k,n}$ the orthogonal projection onto $\mathscr{H}_{k,n}$. Let $u\in D(H_0)$ and $\varepsilon \in \mathbb{R}_+^*$. For each $(k,n)$, one can find $\phi_{k,n}\in [\mathscr{S}(\mathbb{R})]^4$ such that: \[||b_{k,n}P_{k,n}\psi-\phi_{k,n}||_{[H^1(\mathbb{R})]^4}\leq \varepsilon \frac{2^{-\frac{k+n+2}{2}}}{C_{k}},\] where $C_k=\lambda_k||g||_{\infty}+||f||_{\infty}+1$, it follows that:
\begin{equation}
\begin{aligned} &\sum_{k,n} ||P_{k,n}\psi - b^{-1}_{k,n}(\phi_{k,n})||^2 \leq \varepsilon^2, \\&
\sum_{k,n} ||H_0(P_{k,n}\psi - b^{-1}_{k,n}(\phi_{k,n}))||^2 \leq \varepsilon^2.
\end{aligned}
\end{equation}
Therefore, $\displaystyle \sum_{k,n} P_{k,n}\psi - b^{-1}_{k,n}(\phi_{k,n})$ converges to some $y\in D(H_0)$.
 Set $\phi=\psi-y$, then $||\phi-\psi|| + ||H_0(\phi-\psi)|| \leq 2\varepsilon$, and for every $k,n$:
 \[P_{k,n}\phi=P_{k,n}\psi - P_{k,n}y=b_{k,n}^{-1}(\phi_{k,n}),\]
 i.e. $\phi \in \mathscr{S}$. $\varepsilon$ being arbitrary this concludes the proof.
\end{proof}

\medskip

\begin{proof}[Proof of Lemma~\ref{lemme:equiv_quadratic_form}]
On $\mathscr{S}$, the following equation makes sense:
\[H_0^2=D_{r^*}^2 + g(r^*)^2\mathfrak{D}^2 + f(r^*)^2 + \frac{\Gamma^1}{i} g'(r^*)\mathfrak{D}+2g(r^*)\mathfrak{D}f(r^*) + \{f(r^*),\Gamma^1D_{r^*}\}.
 \]
Furthermore, for any $u \in \mathscr{S}$:
\begin{equation}
\begin{aligned} |(\{f(r^*),\Gamma^1D_{r^*}\}u,u) |&\leq |(\Gamma^1D_{r^*}u,f(r^*)u)| +|(f(r^*)u,\Gamma^1 D_{r^*}u)|,\\ &\leq 2 ||\Gamma^1D_{r^*}u|| ||f(r^*)u||,
\\& \leq 2||f||_{\infty}||\Gamma^1D_{r^*}u||||u||, \\ &\leq \frac{1}{2} ||\Gamma^1D_{r^*}u||^2 + 2||f||_{\infty}^2||u||^2. \end{aligned}
\end{equation}
It follows that:
\[ \frac{1}{2} D_{r^*}^2 +2||f||^2_\infty \geq \{f(r^*),\Gamma^1D_{r^*}\} \geq - \frac{1}{2}D_{r^*}^2 - 2||f||_{\infty}^2.\]
Exploiting the fact that $|g'(r^*)| \leq C|g(r^*)|$ for some $C>0$, one has:
\begin{equation}
\begin{aligned} |(\frac{\Gamma^1}{i}g'(r^*)\mathfrak{D}u,u)| &\leq ||g'(r^*)\mathfrak{D}u|| ||u||,
\\ &\leq \frac{1}{4C^2} ||g'(r^*)\mathfrak{D}u||^2 + C^2||u||^2,
\\ & \leq \frac{1}{4} ||g(r^*)\mathfrak{D}u||^2 + C^2||u||^2. \label{eq:justifier}
 \end{aligned}
 \end{equation}
We thus conclude that:
\[ \frac{1}{4}g^2(r^*)\mathfrak{D}^2 +C^2 \geq \frac{\Gamma^1}{i}g'(r^*)\mathfrak{D} \geq - \frac{1}{4}g(r^*)^2 \mathfrak{D}^2 - C^2.\]
In~\eqref{eq:justifier}, we have used the fact that:
\[g'^2(r^*)\mathfrak{D}^2 \leq C^2g^2(r^*)\mathfrak{D}^2.\]
This follows from the functional calculus, since, if $Z$ is an even function in the second variable:
\[\begin{gathered} (Z(r^*,\mathfrak{D})u,u) = \sum_{k,n} \int Z(r^*,\lambda_{k,n}) ||u_{k,n}||^2_{\mathbb{C}^4} \textrm{d}r^*,\\ u= \sum_{k,n} b_{k,n}^{-1}u_{k,n}, u_{k,n}\in [L^2(\mathbb{R})]^4,\end{gathered}\]
and so inequalities valid for $Z$ pass to the operators, here: \[Z(x,y)= g'(x)^2y^2,\] which clearly satisfies:
$Z(x,y)\leq C^2g(x)^2y^2$.
Finally: 
\begin{align*}
|(2g(r^*)\mathfrak{D}f(r^*)u,u)|=& 2|(g(r^*)\mathfrak{D}u,f(r^*)u)|, \\\leq& 2||f||_{\infty}||g(r^*)\mathfrak{D}u||||u||, \\ 
\leq& \frac{1}{4} ||g(r^*)\mathfrak{D}u||^2 + 4 ||f||^2_{\infty}||u||^2.
 \end{align*}
Thus: 
\[ \frac{1}{4}g(r^*)^2\mathfrak{D}^2 +4||f||_{\infty}^2 \geq 2g(r^*)f(r^*)\mathfrak{D} \geq -\frac{1}{4}g(r^*)^2\mathfrak{D}^2 -4||f||_\infty^2,\]
and therefore: 
\[H_0^2 \geq \frac{1}{2}(D^2_{r^*} + g(r^*)^2\mathfrak{D}^2) - C',\]
where $C'= 7||f||_\infty^2 +C^2 >0$.
Overall : \[ \frac{1}{2}Q -C' \leq H_0^2 \leq 2Q+C',\]
which concludes the proof.
\end{proof}

\section{Proof of the technical conditions of Mourre theory}
\label{app:proof_mourre}
\subsection{Preliminary discussion}
In this appendix, the reader will find proofs of Propositions~\ref{prop:technical_mourre1} and~\ref{prop:technical_mourre2}.
We begin by remarking that the condition $H\in C^1(A)$ is quite difficult to check directly, despite the following characterisation: \begin{thm}[{\cite[Theorem 6.2.10]{Amrein:1996aa}}] \label{thm:6210} $H\in C^1(A)$ if and only if the following two conditions are satisfied:
\begin{itemize}
\item there is $c\in \mathbb{R}_+$ such that for all $u\in D(A)\cap D(H)$:
\begin{equation} \label{eq:cond1_class_c1}|(Au,Hu)-(Hu,Au)|\leq c(||Hu||^2+||u||^2),\end{equation}
\item for some $z\in\rho(H)$ the set: \[\{ u \in D(A), (H-z)^{-1}u \in D(A) \textrm{ and } (H-\bar{z})^{-1}u \in D(A) \},\] is a core for $A$. 
\end{itemize}
\end{thm}
To overcome this, there is a useful scheme, based on Nelson's commutator theorem~\cite[Theorems X.36, X.37]{Reed:1975aa}, that greatly simplifies the proof that $H\in C^1(A)$ in many cases.
We first recall Nelson's theorem:
\begin{thm}[Nelson]
\label{thm:nelson}
Let $N$ be a self-adjoint operator with $N\geq 1$. Let $A$ be a symmetric operator with domain $D$ that is also a core for $N$. Suppose that:
\begin{itemize}
\item For some $c$ and all $\psi\in D$,
\begin{equation}
\label{eq:nelson_1}
||A\psi|| \leq c||N\psi||.
\end{equation}
\item For some $d$ and all $\psi \in D$:
\begin{equation}
|(A\psi,N\psi)-(N\psi,A\psi)|\leq d ||N^{\frac{1}{2}}\psi||^2.
\label{eq:nelson_2}
\end{equation}
Then $A$ is essentially self-adjoint on $D$ and its closure is essentially self-adjoint on any other core for $N$.
\end{itemize}
\end{thm}
\begin{rem}
Note that it follows that $D(N)\subset D(\bar{A})$ and $A$ is essentially self-adjoint on $D(N)$.
\end{rem}
The scheme is to find a third operator $N$ - that we will refer to as the \emph{comparison operator} - whose domain is a core for both $H$ and $A$; which we establish using Nelson's lemma. We then seek to apply the following:
\begin{thm}[{\cite[Lemma 3.2.2]{Gerard:2002aa}}]
\label{thm:gerard}
Let $(H,H_0,N)$ be a triplet of self-adjoint operators on $\mathscr{H}$, with $N\geq 1$, $A$ a symmetric operator on $D(N)$. Assume that:
\begin{enumerate} \item $D(H)=D(H_0)\supset D(N)$,
\item $D(N)$ is stable under the action of $(H-z)^{-1}$, 
\item $H_0$ and $A$ satisfy \eqref{eq:nelson_1} and~\eqref{eq:nelson_2},
\item for some $c>0$ and any $u \in D(N)$, \eqref{eq:cond1_class_c1} is satisfied.
\end{enumerate}
Then:
\begin{itemize}
\item $D(N)$ is dense in $D(A)\cap D(H)$ with norm $||Hu||+||Au||+||u||$,
\item the quadratic form $i[H,A]$ defined on $D(A)\cap D(H)$ is the unique extension of $i[H,A]$ on $D(N)$,
\item $H\in C^1(A)$.
\end{itemize}
\end{thm}

\subsection{The comparison operator $N$}
Before identifying the comparison operator $N$, we begin with an important stability lemma:
\begin{lemme}
\label{lemme:stabilite} 
For any $n\in \mathbb{N}^*,z\in\rho(H_0)$, the domain of $\bra{r^*}^n$ is stable under the resolvent $(H_0-z)^{-1}$ and $\chi(H_0)$ for any $\chi\in C^{\infty}_0(\mathbb{R})$. The statement remains true if $H_0$ is replaced with $H$.
\end{lemme}
The proof is identical to that of~\cite[Proposition IV.3.2]{thdaude} and will not be repeated here. This lemma is very important for scattering purposes since it is an indication of how decay rates behave under the action of $H$, but it also serves to justify the use of the following comparison operator\footnote{That has an almost uncanny ressemblance to the harmonic oscillator...}:
\begin{equation} 
\label{eq:def_N} N=D^2_{r^*} + g(r^*)^2\mathfrak{D}^2 + \bra{r^*}^2=Q+\bra{r^*}^2.
\end{equation}
Decomposing $\mathscr{H}$ as in Section~\ref{self-adjointness}, Lemma~\ref{lemme-decomposition} and equation~\eqref{eq:domaine_d2_H} imply that:
\begin{equation} 
\label{eq:caract_domaine_N}D(N)=D(Q)\cap D(\bra{r^*}^2)=D(H_0^2)\cap D(\bra{r^*}^2).
\end{equation}
Finally~\eqref{eq:caract_domaine_N} and Lemmata~\ref{lemme:stabilite} and~\ref{lemme:domaine_H2}, together lead to:
\begin{equation}\begin{aligned}\forall z\in \rho(H_0), \quad (H_0-z)^{-1}D(N) &\subset D(N),
\\ \forall z \in \rho(H), \quad (H-z)^{-1}D(N) &\subset D(N).\end{aligned}\end{equation}
Thus, the first two conditions of Theorem~\ref{thm:gerard} are satisfied by the triplet $(H,H_0,N)$. 
\subsection{Nelson's lemma}
We will now check that $H_0$ and $A_{\pm}(S)$ satisfy the hypotheses of Theorem~\ref{thm:nelson}. To simplify notations, we will omit to specify the dependence on the parameter $S$ of the operator $A_{\pm}$ in this paragraph, as all the results discussed here hold for any $S\geq 1$.
As a first step, we deduce immediately the following useful estimates from~\eqref{eq:def_N}:
\begin{lemme}\label{lemme:ineq_N}
For any $u\in D(N):$
\begin{equation}
\begin{aligned}
||\Gamma^1D_{r^*}u||&\leq ||N^{\frac{1}{2}}u||, & ||g(r^*)\mathfrak{D}u||&\leq ||N^{\frac{1}{2}}u||,
\\ ||r^*u||&\leq ||N^{\frac{1}{2}}u||, & ||u||&\leq ||N^{\frac{1}{2}}u||.
\end{aligned}
\end{equation}
\end{lemme}

\begin{lemme}
With $N$ as comparison operator, $H_0$ satisfies Equations~\eqref{eq:nelson_1} and \eqref{eq:nelson_2}.
\end{lemme}
\begin{proof}
Fix $u\in D(N)$, from Lemma~\ref{lemme:ineq_N}, we have: 
\begin{equation} \begin{aligned}\label{eq:D_nelson_1} ||H_0u|| &\leq ||\Gamma^1D_{r^*}u|| + ||g(r^*)\mathfrak{D}u|| + ||f(r^*)u||,\\&\leq (2+||f||_{\infty})||N^{\frac{1}{2}}u||, \\&\leq(2+||f||_{\infty})||Nu||, \end{aligned} \end{equation}
this proves~\eqref{eq:nelson_1}. Moreover:
\begin{equation}
\begin{aligned} 
|([N,H_0]u,u)| \leq&  2|(\Gamma^1r^*u,u)|+2|(\Gamma^1g'(r^*)g(r^*){\mathfrak{D}}^2u,u)|\\&\hspace{.1in} +2||f'||_{\infty}||D_{r^*}u||||u||
+2||D_{r^*}u||||g'(r^*)\mathfrak{D}u||,
\\ \label{eq:intermediate_step}\leq& 2\big( ||r^*u||||u|| +C||g(r^*)\mathfrak{D}u||^2 \\ &\hspace{.1in}+ ||f'||_{\infty}||D_{r^*}u||||u||  + C||D_{r^*}u||||g(r^*)\mathfrak{D}u|| \big),
\\ \leq& 2(1 +||f'||_{\infty} + 2C)||N^{\frac{1}{2}}u||^2.
\end{aligned}
\end{equation}
In~\eqref{eq:intermediate_step}, we have used the fact that there is $C\in\mathbb{R}^*_+$ such that: \[|g'(r^*)|\leq C|g(r^*)|,\] and the functional calculus as in the proof of Lemma~\ref{lemme:equiv_quadratic_form}.
\end{proof}

In order to establish analogous estimates for $A_-$, we will also need the following estimates: 
\begin{lemme}
\label{lemme:r_estimate}
For any $u\in D(N)$,
\begin{equation}
\begin{aligned}
||{r^*}^2u||^2 &\leq ||Nu||^2 + ||u||^2,\\
||Qu||^2 & \leq ||Nu||^2 + ||u||^2.
\end{aligned}
\end{equation}
\end{lemme}
\begin{proof}
As usual, we will prove it for $u\in \mathscr{S}$. One has:
\begin{equation}
\begin{aligned}
||Nu||^2=&(N^2u,u) \\=&||Qu||^2 + ||{r^*}^2u||^2+||u||^2 +(Qu,{r^*}^2u) \\&+ ({r^*}^2u,Qu) + 2(Qu,u)+2||r^*u||^2.
\end{aligned}
\end{equation}
Since, for any $v\in D(Q)$, $(Qv,v) = ||\Gamma^1D_{r^*}^2v||^2 + ||g(r^*)\mathfrak{D}v||^2 \geq 0$, it follows that:
\begin{align}
||Nu||^2 &\geq ||Qu||^2 + ||{r^*}^2u||^2 + ||u||^2 + (Qu,{r^*}^2u) + ({r^*}^2u,Qu).
\end{align}
Now,
\begin{equation} (Qu,{r^*}^2u)=(r^*Qu,r^*u)= (Qr^*u,r^*u) + (2iD_{r^*}u,r^*u),\end{equation}
and so, adding the hermitian conjugate $({r^*}^2u,Qu)$, one obtains:
\begin{equation*}
\begin{aligned}(Qu,{r^*}^2u) + ({r^*}^2u,Qu)&=2(Qr^*u,r^*u) +(2ir^*D_{r^*}u,u) -(2iD_{r^*}r^*u,u)  
\\ &= 2(Qr^*u,r^*u) - 2||u||^2 \geq -2||u||^2.
\end{aligned}
\end{equation*} 
Hence, 
\begin{align}
||Nu||^2 \geq ||Qu||^2 + ||{r^*}^2u||^2 - ||u||^2.
\end{align}
\end{proof}

\begin{lemme}
\label{lemme:second_derivative_estimate}
There is a constant $d>0$ such that for any $u \in D(Q)=D(H_0^2)$,
\begin{equation} 
||D_{r^*}^2u||^2\leq d( ||Qu||^2 + ||u||^2 ).
\end{equation}
\end{lemme}

\begin{proof}
As quadratic forms on $\mathscr{S}$:
\begin{equation}
\begin{aligned}
Q^2=&D^4_{r^*} + (g^2(r^*){\mathfrak{D}}^2)^2 + D^2_{r^*}g^2(r^*)\mathfrak{D}^2 + g^2(r^*)\mathfrak{D}^2 D^2_{r^*}
,\\=& D^4_{r^*} + (g^2(r^*){\mathfrak{D}}^2)^2 +2\mathfrak{D}g(r^*)D_{r^*}^2g(r^*)\mathfrak{D} \\&
+ [D_{r^*}^2,g(r^*)]g(r^*){\mathfrak{D}}^2 - g(r^*){\mathfrak{D}}^2[D_{r^*}^2,g(r^*)],
\\ \geq& D^4_{r^*} + (g^2(r^*){\mathfrak{D}}^2)^2 + [[D^2_{r^*},g],g]{\mathfrak{D}}^2,
\\ =& D^4_{r^*} + (g^2(r^*){\mathfrak{D}}^2)^2 -i[\{D_{r^*},g'\},g(r^*)]{\mathfrak{D}}^2,
\\=&D^4_{r^*} + (g^2(r^*){\mathfrak{D}}^2)^2  -2(g'(r^*))^2{\mathfrak{D}}^2,
\\ \geq& D^4_{r^*} +(g^2(r^*){\mathfrak{D}}^2)^2 -2C^2g(r^*)^2{\mathfrak{D}}^2,
\\ \geq& D^4_{r^*} +\frac{1}{2}(g^2(r^*){\mathfrak{D}}^2)^2 -2C^4,
\\  \geq& D^4_{r^*} -2C^4.
\end{aligned}
\end{equation}
where we have used the fact that $|g'(r^*)| \leq C |g(r^*)|$.
\end{proof}
\noindent Combining Lemmata~\ref{lemme:r_estimate}~and~\ref{lemme:second_derivative_estimate} yields:
\begin{corollaire}
${r^*}^2,D_{r^*}^2 \in B(D(N),\mathscr{H})$.
\end{corollaire}
We are now ready to prove:
\begin{lemme}
$A_-$ satisfies \eqref{eq:nelson_1} and \eqref{eq:nelson_2}.
\end{lemme}
\begin{proof}
Until now we have not discussed the domain of $A_-$ and will simply consider it as being defined for $u\in\mathscr{S}$, which is a core for $N$. Then, the following estimates hold:
\begin{equation*}
\begin{aligned}
||A_-u||^2=&(R_-(r^*)D_{r^*}u,R_-(r^*)D_{r^*}u) + \frac{1}{4}||R'_-(r^*)u||^2 \\&
-\frac{1}{2}\left( (R_-(r^*)D_{r^*}u, iR'_-(r^*)u) + (iR'_-(r^*)u,R_-(r^*)D_{r^*}u)\right),
\\\leq& (R_-(r^*)D_{r^*}u,R_-(r^*)D_{r^*}u) \\&+ ||R'_-(r^*)R_-(r^*)u||||D_{r^*}u|| +\frac{1}{4}||R'_-(r^*)u||^2.
\end{aligned}
\end{equation*}
Since $R'_-(r^*)$ is a bounded operator, using Lemma~\ref{lemme:ineq_N} one can see that:
\begin{equation*}
\begin{aligned}
 ||R'_-(r^*)R_-(r^*)u||||D_{r^*}u|| + \frac{1}{4}||R'_-(r^*)u||^2 &\leq ||R'_-||_{\infty}||N^\frac{1}{2}u||^2 + \frac{1}{4}||R'_-||^2_{\infty} ||u||^2
\\&\leq ||R'_-||_{\infty}(1+||R'_-||_\infty)||Nu||^2.
\end{aligned}
\end{equation*}
Moreover, by Lemmata~\ref{lemme:r_estimate}~and~\ref{lemme:second_derivative_estimate}:
\begin{align*} |(R_-(r^*)D_{r^*}u,R_-(r^*)D_{r^*}u)| &= |(R_-^2(r^*)u,D_{r^*}^2 u) + 2(iR'_-(r^*)R_-(r^*)u,D_{r^*}u)|
\\\nonumber &\leq \sqrt{6d}||Nu||^2+2||R'_-||_\infty||Nu||^2.
\end{align*}
Combining the above gives~\eqref{eq:nelson_1}. To prove~\eqref{eq:nelson_2} we start with the following estimates:
\begin{equation*}
\begin{aligned} 
|([N,A_-]u,u)| =&\bigg|(-\frac{i}{2}(R^{(3)}_-(r^*)u,u) -i(\{D_{r^*}^2,R_-'(r^*)\}u,u)\\&+2i({r^*}^2j_-^2(\frac{r^*}{S})u,u)+(2ig'(r^*)g(r^*)R_-(r^*)\mathfrak{D}^2u,u) \bigg|,
\\=&\bigg|-i(\{D_{r^*},R'_-(r^*)D_{r^*}\}u,u) -\frac{1}{2}(\{D_{r^*},R_-''(r^*)\}u,u)
\\&+2i({r^*}^2j_-^2(\frac{r^*}{S})u,u)+(2ig'(r^*)g(r^*)R_-(r^*)\mathfrak{D}^2u,u) \bigg|,
\\ \leq& 2||D_{r^*}u|| \left(||R'||_\infty||D_{r^*}u|| + \frac{1}{2}||R''||_\infty ||u|| \right) \\&+2||j_-(\frac{r^*}{S})r^*u||^2
 +2||g(r^*)\mathfrak{D}u|| || g'(r^*)R_-(r^*)\mathfrak{D}u||.
\end{aligned}
\end{equation*}
The only term that may pose problem is:
\begin{equation}|| R_-(r^*)g'(r^*)\mathfrak{D}u||. \end{equation}
However, 
\begin{equation} 
\label{problem_term1}
R_-(r^*)g'(r^*)=g(r^*)j_-^2(r^*)r^*\left( \frac{\frac{\Delta'_r}{2}}{\Xi(r^2+a^2)}-\frac{2r\Delta_r}{\Xi(r^2+a^2)^2} \right),
\end{equation}
and the term between brackets is $\underset{r^* \to - \infty}{O}(\frac{1}{r^*})$ because when $r^*\to -\infty$, $r$ approaches $r_e$, the double root of $\Delta_r$, hence, both $\Delta_r$ and $\Delta'_r$ are at least $\underset{r \to x}{O}(r-r_e)$ and $r-r_e= \underset{r^* \to - \infty}{O}(\frac{1}{r^*})$  \footnote{Note that in~\eqref{problem_term1} $\Delta'_r=\frac{\partial \Delta_r}{\partial r}$}.
In conclusion, there is $C\in \mathbb{R}^*, |R_-(r^*)g'(r^*)|\leq C|g(r^*)|$ and thus, by the functional calculus:
\begin{equation} 
\label{eq:estimate_rdrg}
|| R_-(r^*)g'(r^*)\mathfrak{D}u|| \leq C ||g(r^*)\mathfrak{D}u||.
\end{equation}
Overall, 
\begin{align} 
|([N,A_-]u,u)| \leq \left(||R''_-||_{\infty} +2\left(||R'_-||_{\infty} +C +1\right)\right)||N^{\frac{1}{2}}u||^2
\end{align}
\end{proof}

According to the above result, we can conclude that $A_-$ is essentially self-adjoint on $D(N)$; the analogous result for $A_+$ is proved in~\cite[Lemma IV.4.5]{thdaude}, the arguments are identical. Theorem~\ref{thm:nelson} also applies to $A=A_+\pm A_-$. In all cases, we will consider the operators and their domains as being defined by the conclusion of Theorem~\ref{thm:nelson}.

\subsection{Proof that $H_0,H\in C^1(A)$}
In order to prove that $H,H_0\in C^1(A)$, we require one more estimate that will be the object of this section. According to Theorem~\ref{thm:gerard} it is sufficient to prove that for some $c>0$ and any $u\in D(N)$ one has the estimate:
\begin{equation} \label{eq:H1_eq1} |(Hu,A_\pm u)-(A_\pm u,Hu)| \leq c(||Hu||^2 + ||u||^2). \end{equation}
As before, we will focus our attention on $A_-$ and refer to~\cite[Lemma IV.4.7]{thdaude} for $A_+$. In order to apply Mourre theory, we will additionally need to show that $i[H,A]$ extends to a bounded operator from $D(H)=D(H_0)$\footnote{This equality is to be understood to imply that the graph norms are equivalent.} to $\mathscr{H}$. Both of these are covered by the following estimates, established, first, on the common core $\mathscr{S}$; we begin with $H_0$.

Let $u\in\mathscr{S}$, then:
\begin{equation*}
\begin{aligned} ||i[H_0,A_-]u||=&\big\lVert\Gamma^1 R'_-(r^*)D_{r^*}u - \frac{i}{2}\Gamma^1R^{''}_-(r^*)u \\&- R_-(r^*)g'(r^*)\mathfrak{D}u-R_-(r^*)f'(r^*)u\big\rVert,
\\ \leq& ||R'_-||_{\infty}||D_{r^*}u|| +\frac{1}{2} ||R_-''||_{\infty}||u|| \\&+||R_-(r^*)g'(r^*)\mathfrak{D}u|| + ||R_-f'||_{\infty}||u||.
\end{aligned}
\end{equation*}
Using~\eqref{eq:estimate_rdrg} and Corollary~\ref{corollaire:DrDbounded}, we thus conclude that for some $c>0$ and any $u \in \mathscr{S}$:
\begin{equation}  ||i[H_0,A_-]u|| \leq c(||H_0u|| + ||u||).\end{equation} 
Hence, $i[H_0,A_-]$ extends uniquely to an element of $B(D(H_0),\mathscr{H}$) and~\eqref{eq:H1_eq1} holds. In order to establish the analogous result for $H$, we write:
\begin{equation*} 
[H,A_-]= h[H_0,A_-]h +i(hH_0R_-(r^*)h'+R_-(r^*)h'H_0h) + iR_-(r^*)V'.
\end{equation*}
Since $h,R'_-(r^*) \in B(D(H_0))$, $h[H_0,A_-]h$ and $R'_-(r^*)h'H_0h$ extend to elements of $B(D(H_0),\mathscr{H})$. For similar reasons to $h$, $R_-(r^*)h'\in B(D(H_0))$ also, and, using~\eqref{eq:daabruno}, $R_-(r^*)V'\in B(\mathscr{H})$. It follows then that $[H_0,A_-]$ extends to a bounded operator $D(H_0) \to \mathscr{H}$.

Assembling all the results above, we have thus shown that $H_0,H\in C^1(A)$ and that the first two assumptions of Theorem~\ref{thm:mourre} are satisfied. It remains to verify the final assumption regarding the double commutator.

\subsection{The double commutator assumption}
Theorem~\ref{thm:mourre} only requires that the double commutator extends to a bounded operator from $D(H)$ to $D(H)^*$, this section will be devoted to showing a slightly stronger result: 
\begin{lemme}
\label{lemme:double_commutators}
$[A,[A,H_0]]$ and $[A,[A,H]]$ extend to elements of $B(D(H),\mathscr{H})$.
\end{lemme} 
The consequence will be that $H$ and $H_0$ are in fact $C^2(A)$ (see~\cite[Chapter 5]{Amrein:1996aa}), proving the final point of Proposition~\ref{prop:technical_mourre2}. Beginning with $H_0$, it is sufficient to prove this for the four double commutators $[A_\pm,[A_\pm,H_0]]$ separately; we will mainly concentrate on $A_-$, but it will also be informative to consider the mixed terms $[A_\pm,[A_\mp,H_0]]$.

\paragraph{(a) $[[H_0,A_-],A_-]$}
A short calculation shows that:
\begin{equation}\begin{aligned} \label{eq:dc1}(-i)[i[H_0,A_-],A_-]=&(-i)\bigg( -\frac{1}{2}\Gamma^1R'_-(r^*)R''_-(r^*)-i(R'_-(r^*))^2\Gamma^1D_{r^*} \\&+ iR_-(r^*)R''_-(r^*)\Gamma^1D_{r^*}   -\frac{i}{2}\Gamma^1R_-(r^*)R'''(r^*) \\&-iR_-(r^*)\left(\underline{(R_-(r^*)g'(r^*))'\mathfrak{D}} + (R_-(r^*)f'(r^*))' \right)\bigg).\end{aligned}
\end{equation}
Many of the terms in~\eqref{eq:dc1} extend clearly to elements of $B(D(H),\mathscr{H})$, either because they are bounded on $\mathscr{H}$ or using Corollary~\ref{corollaire:DrDbounded}. The term that merits comment is underlined; it expands as follows:
\begin{equation} R_-(r^*)g''(r^*)\mathfrak{D} +R'_-(r^*)g'(r^*)\mathfrak{D}. \end{equation}
We have already shown how to deal with the second term, and the first is treated very similarly as it is easily seen that $|g''(r^*)|\leq C|g(r^*)|$ for some $C\in\mathbb{R_+^*}$.
\newline
\paragraph{(b) $[i[H_0,A_-],A_+]$}
This double commutator, as a quadratic form on $\mathscr{S}$, can be computed as:
\begin{align*}
(-i)[i[H_0,A_-],A_+]=&(-i)\bigg( [\Gamma^1R'_-(r^*)D_{r^*},A_+] \\&- 2 R_-(r^*)g'(r^*)R_+(r^*,\mathfrak{D})\Gamma^1\mathfrak{D}\bigg).
\end{align*}
The first term vanishes, since on $\mathscr{S}$ it can be evaluated as:
\[[\Gamma^1R'_-(r^*)D_{r^*},A_+]=-R_-'(r^*)R_+'(r^*,\mathfrak{D}),\]
and $j_+$ and $j_-$ have disjoint support (cf.~\eqref{eq:cutoff_def}).
The second term, which, on first glance, seems difficult to control, will equally vanish entirely due to our choice cut-off functions $j_+,j_-,j_1$. To see this, recall that: 
\[R_+(r^*,\mathfrak{D})=(r^*-\kappa^{-1}\ln |\mathfrak{D}| )j_+^2\left(\frac{r^*-\kappa^{-1}\ln |\mathfrak{D}|}{S}\right).\]
Hence, since $j_1$ satisfies $j_1(t) = 1, t\geq -1$, then:
\begin{equation} \label{eq:r+prop} R_+(r^*,\mathfrak{D})=j_1^2(\frac{r^*}{S})R_+(r^*,\mathfrak{D}). \end{equation}
It follows that: \[2 R_-(r^*)g'(r^*)R_+(r^*,\mathfrak{D})\Gamma^1\mathfrak{D} = 2R_-(r^*)j_1^2(\frac{r^*}{S})g'(r^*)R_+(r^*,\mathfrak{D})\Gamma^1\mathfrak{D},\] but, $j_-$ and $j_1$ are chosen such that $\supp j_- \cap \supp j_1 =\emptyset$, therefore this term vanishes.

\paragraph{(c) $[i[H_0,A_+],A_-]$}
Here, we start from\footnote{In this equation $R'(r^*,\mathfrak{D})$ denotes the operator obtained after differentiating with respect to $r^*$}: 
\[i[H_0,A_+]=R_+'(r^*,\mathfrak{D}) + 2ig(r^*)\mathfrak{D}R_+(r^*,\mathfrak{D})\Gamma^1,\]
this leads to:
\[[i[H_0,A_+],A_-]=R_+^{''}(r^*,\mathfrak{D})R_-(r^*) + 2i\left(g(r^*)R_+(r^*,\mathfrak{D})\right)'R_-(r^*)\mathfrak{D}\Gamma^1 .\]
Since~\eqref{eq:r+prop} is equally true if $R_+(r^*,\mathfrak{D})$ is replaced by its first or second derivative with respect to $r^*$, one can argue as before and find that this double commutator vanishes entirely. We refer to~\cite{thdaude} for the appropriate treatment of $[[H_0,A_+],A_+]$.

This concludes the proof that $(H_0,A)$ satisfies the first hypotheses of Mourre theory. To show that this is equally true of $(H,A)$, we proceed as before using~\eqref{eq:relH}. For example:
\begin{multline*} [[H,A_-],A_-]= h[[H_0,A_-],A_-]h+2ih[H_0,A_-]R_-(r^*)h'\\+ 2iR_-(r^*)h'[H_0,A_-]h 
-2h'R_-(r^*)H_0R_-(r^*)h' \\-hH_0R(r^*)(R_-(r^*)h')' -R_-(r^*)(R_-(r^*)h')'H_0h  - R_-(r^*)(R_-(r^*)V')'.
\end{multline*}
This extends to an element of $B(D(H),\mathscr{H})$, thanks to the decay of $h', V'$, etc. Similar computations show that this is equally true of the other double commutators. The reader may be concerned that a long-range potentiel may jeopardise our efforts in the mixed commutators, causing unbounded terms to appear. However, this is not the case since either commutation with $A_-$ introduces the necessary decay through differentiation or terms vanish entirely due to the choice that $j_1$ and $j_-$ have disjoint supports. For the first point, more precisely, if, for instance, $f\in \mathscr{S}(\mathbb{R})$, then :
\[[f(r^*),A_-]=iR_-(r^*)f'(r^*).\]
In all cases encountered, $f$, when expressed as a function of $r$, has bounded derivative and therefore, at least, $[f(r^*),A_-]=O(\frac{1}{r^*})$.

\section{Proof of the existence of the Dollard modified wave operators}\label{app:dollard}
In this appendix we shall prove the existence of~:
\begin{equation}
\label{eq:wavedollard1}
s-\lim_{t\to +\infty} e^{it\mathfrak{h}}U(t)e^{-it\mathfrak{h}_0}\bm{1}_{\{-1\}}(\Gamma^1).
\end{equation}
\begin{proof}[Proof of the existence of~\eqref{eq:wavedollard1}]
The asymptotic velocity operator is simply $\Gamma^1$ for $\mathfrak{h}_0$ which is the reason why we use it to split incoming and outgoing states for $\mathfrak{h}_0$. The first step is to replace the projection with an operator that is more convenient to work with. First of all, for any $J\in C^{\infty}_0(\mathbb{R})$ such that, $\supp J \subset (-\infty,0)$ and $J(-1)=1$, $J(\Gamma^1)=\bm{1}_{\{-1\}}(\Gamma^1)$.
Furthermore for each $t$, one has: 
\begin{equation}
\begin{aligned} e^{it\mathfrak{h}}U(t)J(\frac{r^*}{t})e^{-it\mathfrak{h}_0}&= e^{it\mathfrak{h}}U(t)e^{-it\mathfrak{h}_0}(e^{it\mathfrak{h}_0}J(\frac{r^*}{t})e^{-it\mathfrak{h}_0}-J(\Gamma^1))\\ &\quad + e^{it\mathfrak{h}}U(t)e^{-it\mathfrak{h}_0} J(\Gamma^1).\end{aligned} \end{equation}
Now, $e^{it\mathfrak{h}}U(t)e^{-it\mathfrak{h}_0}$ is uniformly bounded in $t$ so applying\footnote{although it is simpler for $\mathfrak{h}_0$} Corollary~\ref{lemme:asymptotic_velocity} to $\mathfrak{h}_0$, we find that the strong limit of the first term exists and is $0$, so, using another classical density argument we only need to prove the existence of:
\[s-\lim_{t\to+\infty} e^{it\mathfrak{h}}U(t)J(\frac{r^*}{t})e^{-it\mathfrak{h}_0}\chi(\mathfrak{h}_0),\]
for any $\chi \in C^{\infty}_0(\mathbb{R}), 0 \not\in \supp \chi$, this in particular implies that $\chi \equiv 0$ on a neighbourhood of $0$.

Once more, we use Cook's method and to that end we calculate the derivative; one finds: 
\[\begin{aligned} e^{it\mathfrak{h}}\left(iJ(\frac{r^*}{t})(-\mu)g(r^*)\Gamma^2 +\frac{1}{t}J'(\frac{r^*}{t})(\Gamma^1- \frac{r^*}{t})\right)\chi(\mathfrak{h}_0)U(t)e^{-it\mathfrak{h}_0}\\ +e^{it\mathfrak{h}}\left(iJ(\frac{r^*}{t})f(r^*) - iJ(\frac{r^*}{t})\tilde{f}(t\Gamma^1)\right)\chi(\mathfrak{h}_0)U(t)e^{-it\mathfrak{h}_0}.\end{aligned}\]
The term involving $J'$ can be treated by the second method explained in the proof of Proposition~\ref{prop:comparaison1exist}; we will not repeat the reasoning here.

Let us examine the first term:
\[T_1=e^{it\mathfrak{h}}(iJ(\frac{r^*}{t})(-\mu)g(r^*)\Gamma^2 e^{-it\mathfrak{h}_0}U(t)\chi(\mathfrak{h}_0),\] where we have used the fact that $\Gamma^1$ commutes with $\mathfrak{h}_0$, hence $U(t)$ commutes with $\chi(\mathfrak{h}_0)$ and $e^{-it\mathfrak{h}_0}$.
Since $\Gamma^2$ anti-commutes with $\Gamma^1$, $\Gamma^2U(t)=\tilde{U}(t)\Gamma^2$, where\footnote{The operators under consideration here are all bounded, the series defining $U(t)$ converges in norm and $\tilde{f}$ is continuous and bounded, so one only needs to check the anti-commutation property on polynomials.} $\tilde{U}(t)=T\exp(i\tilde{f}(-\Gamma^1t))$, so one can rewrite $T_1$ as follows:
\[T_1=e^{it\mathfrak{h}}iJ(\frac{r^*}{t})(-\mu)g(r^*)\tilde{U}(t)e^{-it\mathfrak{h}_0}e^{it\mathfrak{h}_0}\Gamma^2 e^{-it\mathfrak{h}_0}\chi(\mathfrak{h}_0).$$
Set $E(t)= \int_0^t e^{is\mathfrak{h}_0}\Gamma^2e^{-is\mathfrak{h}_0}\chi(\mathfrak{h}_0)\textrm{d}s$. $\Gamma^2$ anti-commutes with $\mathfrak{h}_0$, therefore:  $$E(t)=\Gamma^2 \int_0^t e^{-2is\mathfrak{h}_0}\chi(\mathfrak{h}_0)\textrm{d}s.\]
However, it follows from the bounded functional calculus that:
\[\left \lVert \int_0^t e^{-2is\mathfrak{h}_0}\chi(\mathfrak{h}_0)\textrm{d}s \right \rVert = \sup_{\lambda \in \mathbb{R}} \left|\int_0^t e^{-2is\lambda}\chi(\lambda)\textrm{d}s \right|.\]
Since $\chi\equiv 0$ on a neighbourhood of $0$, this is finite and bounded independently of $t$, so $E(t)$ is a uniformly bounded function of $t$.
Now, for any $t_1,t_2 \geq 1$,
\begin{multline}\int_{t_1}^{t_2}T_1(t)\textrm{dt}= \left[e^{it\mathfrak{h}}iJ(\frac{r^*}{t})(-\mu)g(r^*)\tilde{U}(t)e^{-it\mathfrak{h}_0}E(t) \right]_{t_1}^{t_2} \\ - \int_{t_1}^{t_2}\partial_t\left(e^{it\mathfrak{h}}(iJ(\frac{r^*}{t})(-\mu)g(r^*)\tilde{U}(t)e^{-it\mathfrak{h}_0}\right)E(t)\textrm{d}t \end{multline}
Since $J$ vanishes on a neighbourhood of $0$, and $E(t)$ is uniformly bounded, the term in the squared brackets vanishes as $t_1,t_2 \to +\infty$:
\begin{equation}
\begin{aligned}
\left\lVert e^{it\mathfrak{h}}iJ(\frac{r^*}{t})(-\mu)g(r^*)\tilde{U}(t)e^{-it\mathfrak{h}_0}E(t) \right\rVert &= O\left(|g(r^*)|J(\frac{r^*}{t})\right)\\&=O\left(\frac{1}{t}\right).
\end{aligned}
\end{equation}
Additionally, due to the further derivative, the integrand in the second term is $O(t^{-2})$ and hence integrable. 
It remains to treat the final terms:
\[T_2=e^{it\mathfrak{h}}\left(iJ(\frac{r^*}{t})f(r^*) - iJ(\frac{r^*}{t})\tilde{f}(t\Gamma^1)\right)\chi(\mathfrak{h}_0)U(t)e^{-it\mathfrak{h}_0}.\]
Notice first that, $\supp J \subset (0,-\infty)$, so $J(\frac{r^*}{t})=J(\frac{r^*}{t})j(r^*)$ and:
\[T_2=e^{it\mathfrak{h}}iJ(\frac{r^*}{t})\left(\tilde{f}(r^*) - \tilde{f}(t\Gamma^1)\right)\chi(\mathfrak{h}_0)U(t)e^{-it\mathfrak{h}_0}.\]
It follows from~\eqref{eq:daabruno} and the subsequent remarks that $\tilde{f} \in \bm{S}^{1,1}$, and one can use the Helffer-Sjöstrand formula to obtain an expression for $(\tilde{f}(r^*) - \tilde{f}(t\Gamma^1))J(\frac{r^*}{t})$ as in the proof of Lemma~\ref{lemme:asymptotic_velocity}:
\[(\tilde{f}(r^*) - \tilde{f}(t\Gamma^1))J(\frac{r^*}{t})=B(t)(\Gamma^1-\frac{r^*}{t})J(\frac{r^*}{t})\] where $B$ is a uniformly bounded operator in $t$. The desired integrability result is hence a consequence of the microlocal velocity estimate~\eqref{eq:ml_ve}; the existence of~\eqref{eq:wavedollard1} follows.
\end{proof}

\bibliography{biblio}

\def\bysame{\leavevmode ---------\thinspace}
\makeatletter\if@francais\providecommand{\og}{<<~}\providecommand{\fg}{~>>}
\else\gdef\og{``}\gdef\fg{''}\fi\makeatother
\def\cdrandname{\&}
\providecommand\cdrnumero{no.~}
\providecommand{\cdredsname}{eds.}
\providecommand{\cdredname}{ed.}
\providecommand{\cdrchapname}{chap.}
\providecommand{\cdrmastersthesisname}{Memoir}
\providecommand{\cdrphdthesisname}{PhD Thesis}
\begin{thebibliography}{10}

\bibitem{Abrikosov:2002aa}
{\scshape A.~A. Abrikosov, Jr.}, {\og {Dirac operator on the Riemann
  Sphere}\fg}, arXiv:hep-th/021234, 2002,
  \url{https://arxiv.org/abs/hep-th/0212134}.

\bibitem{Amrein:1996aa}
{\scshape W.~O. Amrein, A.~B. de~Monvel {\normalfont
  \cdrandname}~V.~Georgescu}, \emph{{$C^0$-Groups, Commutator Methods and
  Spectral Theory of N-Body Hamiltonians}}, Modern Birkh{\"a}user Classics,
  vol. XIV, Birkh{\"a}user Basel, 1996.

\bibitem{Aretakis:2011aa}
{\scshape S.~Aretakis}, {\og {Stability and Instability of Extreme
  Reissner-Nordstr\"om Black Hole Spacetimes for Linear Scalar Perturbations
  I}\fg}, \emph{Communications in Mathematical Physics} \textbf{307} (2011),
  \cdrnumero 1, p.~17.

\bibitem{Aretakis:2011hc}
\bysame , {\og {Stability and Instability of Extreme Reissner-Nordstr\"om Black
  Hole Spacetimes for Linear Scalar Perturbations II}\fg}, \emph{Annales Henri
  Poincare} \textbf{12} (2011), p.~1491-1538,
  \url{https://arxiv.org/abs/1110.2009}.

\bibitem{Batic:2007jb}
{\scshape D.~Batic {\normalfont \cdrandname}~H.~Schmid}, {\og {The Dirac
  propagator in the extreme Kerr metric}\fg}, \emph{J. Phys. A} \textbf{40}
  (2007), p.~13443-13452, \url{https://arxiv.org/abs/gr-qc/0703023}.

\bibitem{Belgiorno:2009aa}
{\scshape F.~Belgiorno {\normalfont \cdrandname}~S.~L. Cacciatori}, {\og
  {Absence} of normalizable time-periodic solutions for the {Dirac} equation in
  {Kerr}-{Newman}-{DS} black hole background\fg}, \emph{Journal of Physics A:
  Mathematical and Theoretical} \textbf{42} (2009), \cdrnumero 13, p.~135207,
  \url{https://arxiv.org/abs/arXiv:0807.4310}.

\bibitem{Belgiorno:2010aa}
\bysame , {\og {The Dirac equation in Kerr-Newman-AdS black hole
  background}\fg}, \emph{Journal of Mathematical Physics} \textbf{51} (2010),
  \cdrnumero 3, p.~033516-033517.

\bibitem{Bizon:2012we}
{\scshape P.~Bizon {\normalfont \cdrandname}~H.~Friedrich}, {\og {A remark
  about wave equations on the extreme Reissner-Nordstr\"om black hole
  exterior}\fg}, \emph{Class. Quant. Grav.} \textbf{30} (2013), p.~065001,
  \url{https://arxiv.org/abs/1212.0729}.

\bibitem{Borthwick:2018aa}
{\scshape J.~{Borthwick}}, {\og {Maximal Kerr de Sitter spacetimes}\fg},
  \emph{Classical Quantum Gravity} \textbf{35} (2018), \cdrnumero 21,
  p.~215006.

\bibitem{Camporesi:1996aa}
{\scshape R.~Camporesi {\normalfont \cdrandname}~A.~Higuchi}, {\og {On the
  eigenfunctions of the Dirac operator on spheres and real hyperbolic
  spaces}\fg}, \emph{Journal of Geometry and Physics} \textbf{20} (1996),
  p.~1-18.

\bibitem{Dafermos:2014jwa}
{\scshape M.~Dafermos, I.~Rodnianski {\normalfont
  \cdrandname}~Y.~Shlapentokh-Rothman}, {\og {A scattering theory for the wave
  equation on Kerr black hole exteriors}\fg}, \emph{Annales Scientifiques de
  l'ENS} \textbf{51} (2018), \cdrnumero 2, p.~371-486,
  \url{https://arxiv.org/abs/1412.8379}.

\bibitem{Dappiaggi:2009fx}
{\scshape C.~Dappiaggi, V.~Moretti {\normalfont \cdrandname}~N.~Pinamonti},
  {\og {Rigorous construction and Hadamard property of the Unruh state in
  Schwarzschild spacetime}\fg}, \emph{Adv. Theor. Math. Phys.} \textbf{15}
  (2011), \cdrnumero 2, p.~355-447, \url{https://arxiv.org/abs/0907.1034}.

\bibitem{thdaude}
{\scshape T.~Daud{\'e}}, {\og {Scattering theory for Dirac fields in various
  spacetimes of the General Relativity}\fg}, Theses, {Universit{\'e} Sciences
  et Technologies - Bordeaux I}, 12 2004.

\bibitem{Daude:2010aa}
{\scshape T.~Daud\'e}, {\og {Time-dependent scattering theory for charged Dirac
  fields on a Reissner-Nordstr{\"o}m black hole}\fg}, \emph{{Journal of
  Mathematical Physics}} \textbf{51} (2010), \cdrnumero 10, p.~102504.

\bibitem{Daude:2016aa}
{\scshape T.~Daud{\'e} {\normalfont \cdrandname}~F.~Nicoleau}, {\og {Direct and
  inverse scattering at fixed energy for massless charged Dirac fields by
  Kerr-Newman-de Sitter black holes}\fg}, \emph{Memoirs of the American
  Mathematical Society} \textbf{247} (2016), \cdrnumero 1170,
  \url{https://arxiv.org/abs/hal-00841788}.

\bibitem{Derezinski:1997aa}
{\scshape J.~Derezinski {\normalfont \cdrandname}~C.~Gerard}, \emph{{Scattering
  theory of Classical and Quantum N-particle System}}, Theoretical and
  Mathematical Physics, vol. XII, Springer-Verlag Berlin Heidelberg, 1997.

\bibitem{Dollard:1966aa}
{\scshape J.~Dollard {\normalfont \cdrandname}~G.~Velo}, {\og {Asymptotic
  Behaviour of a Dirac Particle in a Coulomb Field}\fg}, \emph{Il Nuovo Cimento
  A (1965-1970)} \textbf{45} (1966), \cdrnumero 4, p.~801-812.

\bibitem{Evans:2010aa}
{\scshape L.~Evans}, \emph{Partial Differential Equations}, {Second}
  \cdredname, Graduate Studies in Mathematics, vol.~19, American Mathematical
  Society, 2010.

\bibitem{Mantoui:2001aa}
{\scshape V.~Georgescu {\normalfont \cdrandname}~M.~Mantoui}, {\og {On the
  Spectral Theory of Singular Dirac Type Hamiltonians}\fg}, \emph{{Journal of
  Operator Theory}} \textbf{46} (2001), p.~289-321.

\bibitem{Georgescu:2017vl}
{\scshape V.~Georgescu, C.~G{\'e}rard {\normalfont \cdrandname}~D.~H{\"a}fner},
  {\og {Asymptotic completeness for superradiant Klein-Gordon equations and
  applications to the De Sitter Kerr metric}\fg}, \emph{Journal of the European
  Mathematical Society} \textbf{19} (2017), \cdrnumero 8, p.~2371-2444.

\bibitem{gerard:hal-02939993}
{\scshape C.~G{\'e}rard, D.~H{\"a}fner {\normalfont \cdrandname}~M.~Wrochna},
  {\og {The Unruh state for massless fermions on Kerr spacetime and its
  Hadamard property}\fg}, {Accepted for publication in Annales Scientifiques de
  l'ENS}, 09 2020.

\bibitem{Gerard:2002aa}
{\scshape C.~Gerard {\normalfont \cdrandname}~I.~Laba}, \emph{{Multiparticle
  Quantum Scattering}}, Mathematical Surveys and Monographs, vol.~90, American
  Mathematical Society, 2002.

\bibitem{Geroch:1968aa}
{\scshape R.~Geroch}, {\og {Spinor Structure of Space-Times in General
  Relativity. I}\fg}, \emph{{Journal of Mathematical Physics}} \textbf{9}
  (1968), \cdrnumero 11, p.~1739-1744.

\bibitem{Geroch:1970aa}
\bysame , {\og {Spinor Structure of Space-Times in General Relativity. II}\fg},
  \emph{{Journal of Mathematical Physics}} \textbf{1} (1970), \cdrnumero 11,
  p.~343-348.

\bibitem{Nicolas:2004aa}
{\scshape D.~H{\"a}fner {\normalfont \cdrandname}~J.-P. Nicolas}, {\og
  Scattering of massless {Dirac} fields by a slow {Kerr} black hole\fg},
  \emph{Reviews in Mathematical Physics} \textbf{16} (2004), \cdrnumero 1,
  p.~29-123.

\bibitem{Hafner:2003aa}
{\scshape D.~H{\"a}fner}, {\og {Sur la th{\'e}orie de la diffusion pour
  l'{\'e}quation de Klein-Gordon dans la m{\'e}trique de Kerr}\fg},
  \emph{Dissertationes Mathematicae} (2003), \cdrnumero 421, p.~102.

\bibitem{Helffer:1987aa}
{\scshape B.~Helffer {\normalfont \cdrandname}~J.~Sj{\"o}strand}, {\og Equation
  de Schr{\"o}dinger avec champ magn{\'e}tique et {\'e}quation de Harper\fg},
  \emph{Journ{\'e}es {\'e}quations aux d{\'e}riv{\'e}es partielles} (1987),
  p.~1-9.

\bibitem{idelonriton:tel-01370116}
{\scshape G.~Idelon-Riton}, {\og {On scattering theory for the massive Dirac
  equation in Schwarzschild-Anti- de Sitter space-time and applications}\fg},
  \cdrphdthesisname, {Universit{\'e} Grenoble Alpes}, 07 2016.

\bibitem{Ince:1956aa}
{\scshape E.~Ince}, \emph{Ordinary Differential Equations}, Dover Publications,
  1956.

\bibitem{Kato:1980aa}
{\scshape T.~Kato}, \emph{Perturbation Theory for Linear Operators}, {Second}
  \cdredname, Classics in Mathematics, Springer-Verlag Berlin Heidelberg, 1980.

\bibitem{Lax:2002aa}
{\scshape P.~D. Lax}, \emph{Functional {Analysis}}, Pure and Applied
  Mathematics, Wiley, 2002.

\bibitem{Mourre:1981aa}
{\scshape E.~Mourre}, {\og {Absence of singular continous spectrum for certain
  self-adjoint operators}\fg}, \emph{Communications in Mathematical Physics}
  \textbf{78} (1981), \cdrnumero 3, p.~391-408.

\bibitem{Nicolas:2002aa}
{\scshape J.-P. Nicolas}, {\og Dirac fields on asymptotically flat
  space-times\fg}, \emph{Dissertationes Mathematicae} \textbf{408} (2002).

\bibitem{Penrose:1984aa}
{\scshape R.~Penrose {\normalfont \cdrandname}~W.~Rindler}, \emph{Spinors and
  space-time : {Two-spinor calculus and relativistic fields}}, Cambridge
  Monographs on Mathematical Physics, Cambridge University Press, 1984.

\bibitem{Reed:1975aa}
{\scshape M.~Reed {\normalfont \cdrandname}~B.~Simon}, \emph{{Fourier Analysis,
  Self-adjointness}}, {Methods of Modern Mathematical Physics}, vol.~II,
  Academic Press, Inc, 1975.

\bibitem{Reed:1979aa}
\bysame , \emph{{Scattering theory}}, {Methods of Modern Mathematical Physics},
  vol. III, Academic Press, Inc, 1979.

\bibitem{Sigal:1988aa}
{\scshape I.~Sigal {\normalfont \cdrandname}~A.~Soffer}, {\og Local Decay and
  Velocity bounds for Quantum Propagation\fg}, \emph{Princeton University}
  (1988).

\bibitem{Trautman:1993aa}
{\scshape A.~Trautman}, {\og {Spin structures on hypersurfaces and the spectrum
  of the Dirac operator on spheres}\fg}, in \emph{{Spinors,Twistors, Clifford
  Algebras and Quantum Deformations}} (Dordrecht) (Z.~Oziewicz, B.~Jancewicz
  {\normalfont \cdrandname}~A.~Borowiec, \cdredsname), Springer Netherlands,
  1993, p.~25-29.

\end{thebibliography}
\end{document}